\documentclass[11pt]{article}
\usepackage{amsfonts,amsmath,amsthm,amssymb,dsfont, etex}

\usepackage{tikz}
\usetikzlibrary{arrows,automata}
\usepackage{comment, times}
\usepackage[margin=1in]{geometry}
\usepackage{soul,color}
\usepackage{graphicx,float,wrapfig}
\usepackage{mathrsfs}
\usepackage[usenames,dvipsnames]{pstricks}
\usepackage{epsfig}
\usepackage{pst-grad} 
\usepackage{pst-plot} 
\usepackage[linesnumbered,boxed,ruled,vlined]{algorithm2e}

\usepackage{authblk}
\usepackage{url}

\usepackage{tabu}

\usepackage[normalem]{ulem}

\newtheorem{theo}{Theorem}[section]

\newtheorem{lemma}[theo]{Lemma}
\newtheorem{claim}{Claim}
\newtheorem{prop}[theo]{Proposition}
\newtheorem{cor}[theo]{Corollary}
\theoremstyle{definition}
\newtheorem{defi}[theo]{Definition}
\theoremstyle{plain}
\newtheorem{rem}[theo]{Remark}

\newtheorem{ft}{Fact}

\newenvironment{proofof}[1]{\begin{proof}[Proof of #1]}{\end{proof}}

\newcommand{\R}{\mathbb{R}}
\newcommand{\Ex}{\mathbb{E}}
\everymath{\displaystyle}

\newcommand{\distr}{\mathcal{D}}

\newcommand{\indicator}{\mathds{1}}

\newcommand{\negl}{\mathsf{negl}}

\newcommand{\bset}{\{0,1\}}
\newcommand{\lone}[1]{\| #1 \|_1}
\newcommand{\ltwo}[1]{\| #1 \|_2}
\newcommand{\inangle}[1]{\langle #1 \rangle}
\newcommand{\ip}[2]{\inangle{#1,#2}}

\newcommand{\BPPPATH}{\mathsf{BPP_{path}}}
\newcommand{\SZK}{\mathsf{SZK}}

\newcommand{\SZKcc}{\mathsf{SZK}^{\cc}} 

\newcommand{\HVSZK}{\mathsf{HVSZK}}
\newcommand{\NISZK}{\mathsf{NISZK}}
\newcommand{\NISZKcc}{\mathsf{NISZK}^{\cc}} 
\newcommand{\Stwocc}{\Sigma_2^{\cc}} 
\newcommand{\NIPZK}{\mathsf{NIPZK}}
\newcommand{\coNIPZK}{\mathsf{coNIPZK}}
\newcommand{\PZK}{\mathsf{PZK}}
\newcommand{\SRE}{\mathsf{SRE}}
\newcommand{\HVPZK}{\mathsf{HVPZK}}
\newcommand{\coPZK}{\mathsf{coPZK}}
\newcommand{\coHVPZK}{\mathsf{coHVPZK}}
\newcommand{\BQP}{\mathsf{BQP}}
\newcommand{\naCQP}{\mathsf{naCQP}}
\newcommand{\CQP}{\mathsf{CQP}}
\newcommand{\DQP}{\mathsf{DQP}}
\newcommand{\QMA}{\mathsf{QMA}}
\newcommand{\BPP}{\mathsf{BPP}}

\newcommand{\AM}{\mathsf{AM}}
\newcommand{\AMcc}{\mathsf{AM}^{\cc}} 

\newcommand{\coAM}{\mathsf{coAM}}
\newcommand{\coAMcc}{\mathsf{coAM}^{\cc}}
\newcommand{\PTIME}{\mathsf{P}}

\newcommand{\PH}{\mathsf{PH}}
\newcommand{\PP}{\mathsf{PP}}
\newcommand{\PPcc}{\mathsf{PP}^{\cc}} 
\newcommand{\EXP}{\mathsf{EXP}}
\newcommand{\AOPP}{\mathsf{A}_\mathsf{0}\mathsf{PP}}
\newcommand{\NP}{\mathsf{NP}}
\newcommand{\coNP}{\mathsf{coNP}}
\newcommand{\PSPACE}{\mathsf{PSPACE}}
\newcommand{\BPPpath}{\mathsf{BPP}_\mathsf{path}}
\newcommand{\PostBQP}{\mathsf{PostBQP}}
\newcommand{\co}{\mathsf{co}}

\newcommand{\UPP}{\mathsf{UPP}}
\newcommand{\UPPcc}{\mathsf{UPP}^{\cc}} 
\newcommand{\cc}{\mathsf{cc}}
\newcommand{\eps}{\epsilon}
\newcommand{\PPdt}{\mathsf{PP}^{\mathsf{dt}}}
\newcommand{\PZKdt}{\mathsf{PZK}^{\mathsf{dt}}}
\newcommand{\coPZKdt}{\mathsf{coPZK}^{\mathsf{dt}}}
\newcommand{\UPPdt}{\mathsf{UPP}^{\mathsf{dt}}}
\newcommand{\NISZKdt}{\mathsf{NISZK}^{\mathsf{dt}}}
\newcommand{\AMdt}{\mathsf{AM}^{\mathsf{dt}}}
\newcommand{\SZKdt}{\mathsf{SZK}^{\mathsf{dt}}}

\newcommand{\Collision}{\mathsf{Col}}

\newcommand{\PCollision}{\textsf{Collision}}

\newcommand{\poly}{\operatorname*{poly}}
\newcommand{\polylog}{\operatorname*{polylog}}

\newcommand{\Oracle}{\mathcal{O}}

\newcommand{\PTP}{\mathsf{PTP}}

\newcommand{\dega}{\widetilde{\mathrm{deg}}}

\newcommand{\odegap}{\mathrm{deg}^{+}}

\newcommand{\degsh}{\deg_{\pm}}

\newcommand{\Gapmajority}{\mathsf{GapMaj}}
\newcommand{\GapMaj}[3]{\mathsf{GapMaj}_{#2,#3}(#1)}
\newcommand{\GapAND}[3]{\mathsf{GapAND}_{#2,#3}(#1)}
\newcommand{\Gapand}{\mathsf{GapAND}}

\newcommand{\SDU}{\textsf{SDU}}

\newcommand{\EA}{\mathsf{EA}}

\newcommand{\AND}{\mathsf{AND}}
\newcommand{\OR}{\mathsf{OR}}

\newcommand{\cd}[2]{\mathcal{C}_{#1,#2}}

\newcommand{\True}{\textsf{TRUE}}
\newcommand{\False}{\textsf{FALSE}}

\newcommand{\dort}[1]{$#1$-orthogonalizing}

\newcommand{\sympsi}{\tilde{\psi}}
\newcommand{\symh}{\tilde{h}}
\newcommand{\symf}{\tilde{f}}
\newcommand{\symp}{\tilde{p}}
\newcommand{\symg}{\tilde{g}}
\newcommand{\hatvarphi}{\hat{\varphi}}
\newcommand{\hatmu}{\hat{\mu}}
\newcommand{\hatpsi}{\hat{\psi}}
\newcommand{\hatnu}{\hat{\nu}}
\newcommand{\hatrho}{\hat{\rho}}
\newcommand{\hatphi}{\hat{\phi}}
\newcommand{\hatRho}{\hat{P}}
\newcommand{\Rho}{P}
\newcommand{\supp}{\operatorname{supp}}

\newcommand{\lb}{n^{1/4}/\log n}
\renewcommand{\epsilon}{\varepsilon}

\def\ShowAuthNotes{1}
\ifnum\ShowAuthNotes=1
\newcommand{\authnote}[2]{\ \\ \textcolor{red}{\parbox{0.9\linewidth}{[{\footnotesize {\bf #1:} { {#2}}}]}}\newline}
\else
\newcommand{\authnote}[2]{}
\fi

\let\svfootnoterule\footnoterule
\renewcommand\footnoterule{\vfill\svfootnoterule}
\begin{document}
\title{On the Power of Statistical Zero Knowledge}
\author[1]{Adam Bouland}
\author[2]{Lijie Chen}
\author[1]{Dhiraj Holden}
\author[3]{Justin Thaler}
\author[1]{Prashant Nalini Vasudevan}
\affil[1]{CSAIL, Massachusetts Institute of Technology, Cambridge, MA USA}
\affil[2]{IIIS, Tsinghua University, Beijing, China}
\affil[3]{Georgetown University, Washington, DC USA}

\date{}

\clearpage\maketitle
\thispagestyle{empty}

\begin{abstract}
    We examine the power of statistical zero knowledge proofs (captured by the complexity
class $\SZK$) and their variants. 
First, we give the strongest known relativized evidence that $\SZK$ contains hard problems,
by exhibiting an oracle relative to which $\SZK$ (indeed, even $\NISZK$) 
is not contained in the class $\UPP$, containing those problems solvable by randomized algorithms with unbounded error. 
This answers an open question of Watrous from 2002 \cite{AaronsonPersonal}.
Second, we ``lift'' this oracle separation to the setting of communication complexity,
thereby answering a question of G{\"{o}}{\"{o}}s et al. (ICALP 2016).
Third, we give relativized evidence that \emph{perfect} zero knowledge
proofs (captured by the class $\PZK$) are weaker than general zero knowledge proofs. 
Specifically, we exhibit oracles relative to which $\SZK \not \subseteq \PZK$,
$\NISZK \not \subseteq \NIPZK$, and $\PZK \neq \coPZK$. The first of these results
answers a question raised in 1991 by Aiello and H{\aa}stad (Information and Computation), 
and the second answers a question of Lovett and Zhang (2016). 
We also describe additional applications of these results outside of structural complexity.

The technical core of our results is a stronger hardness amplification 
theorem for approximate degree, which roughly says that composing 
the gapped-majority function with any function of high approximate 
degree yields a function with high threshold degree.

\end{abstract}
\addtocounter{page}{-1}
\newpage

\section{Introduction}
Zero knowledge proof systems, first introduced by Goldwasser, Micali and Rackoff \cite{goldwasser1989knowledge}, have proven central to the study of complexity theory and cryptography. 
Abstractly, a zero knowledge proof is a form of interactive proof in which the verifier can efficiently simulate the honest prover on ``yes" instances. Therefore, the verifier learns nothing other than whether its input is a ``yes" or ``no" instance.

In this work, we study \emph{statistical} zero knowledge proofs systems.
Here, ``efficiently simulate" means that the verifier can, by itself, sample from a distribution which is statistically close to the distribution of the transcript of its interaction with the honest prover\footnote{\emph{Computational} zero-knowledge, in which the zero-knowledge condition is that the verifier can sample from a distribution that is \emph{computationally indistinguishable} from the transcript, 
has also been the subject of intense study. In this work we focus exclusively on statistical zero knowledge.}.
The resulting class of decision problems that have statistical zero knowledge proofs is denoted $\SZK$. 
One can similarly define variants of this class, such as non-interactive statistical zero knowledge (where the proof system is non-interactive, denoted $\NISZK$), or perfect zero knowledge (where the verifier can exactly simulate the honest prover, denoted $\PZK$).

Many problems, some of which are not necessarily in $\NP$, have been shown to admit $\SZK$ protocols. These include Graph Non-isomorphism, as well as problems believed to be hard on average, such as Quadratic Residuosity (which is equivalent to the discrete logarithm problem), and the Approximate Shortest Vector and Closest Vector problems in lattices \cite{goldreich1991proofs,goldwasser1989knowledge,goldreich1998limits,peikert2008noninteractive}. 
Although $\SZK$ contains problems believed to be hard, it lies very low in the polynomial hierarchy 
(below $\AM \cap \coAM$), and cannot contain $\NP$-complete
problems unless the polynomial hierarchy collapses \cite{Fortnow87, AH91, bhz87}. 
Owing in part to its unusual property of containing problems believed to be hard but not $\NP$-complete,
$\SZK$ has been the subject of intense interest among complexity theorists and cryptographers.

Despite its importance, many basic questions about the hardness of $\SZK$
and  its variants remain open.
Our results in this work can be understood as grouped into three classes, detailed in each
of the next three subsections. However, we prove these results
via a unified set of techniques.

\subsection{Group 1: Evidence for the Hardness of $\SZK$}
\label{s:results1}
\medskip
\noindent \textbf{Motivation.} Several cryptosystems have been based on the believed hardness of problems in $\SZK$, most notably Quadratic Residuosity and the Approximate Shortest Vector and Closest Vector problems mentioned above.
If one could solve $\SZK$-hard problems efficiently, it would break  these cryptosystems.
Hence, a natural task is to show lower bounds demonstrating that problems in $\SZK$ cannot be solved easily. 
For example, one might want to show that quantum computers or other, more powerful models of computation cannot solve $\SZK$-hard problems efficiently. 
This would provide evidence for the belief that problems in $\SZK$ are computationally hard.

Of course, proving such results unconditionally is very difficult, because $\SZK$ is contained in $\AM \cap \coAM$ \cite{Fortnow87, AH91}, so even proving lower bounds against classical algorithms solving $\SZK$-hard problems would require separating $\mathsf{P}$ from $\NP$.\footnote{Since $\SZK \subseteq \AM \cap \coAM \subseteq \PH$, 
if $\mathsf{P} \neq \SZK$, then $\mathsf{P} \neq \PH$, which in particular implies $\mathsf{P} \neq \NP.$} 
Therefore, a more reasonable goal has been to create oracles relative to which $\SZK$ is not contained in other complexity classes; one can then unconditionally prove that ``black-box" algorithms from other complexity classes cannot break $\SZK$.

\medskip
\noindent \textbf{Additional Context.}  While much progress has been made in this direction (see Section \ref{sec:priorhardness} for details), 
the problem of giving an oracle separation between $\SZK$ and $\PP$ has been open since it was posed by Watrous in 2002 \cite{AaronsonPersonal} and additionally mentioned as an open problem in \cite{aaronson2012impossibility}.
Here, $\PP$ is the set of decision problems decidable in polynomial time by randomized algorithms with unbounded error.
Since a $\PP$ algorithm can flip polynomially many coins in its decision process, the gap between the acceptance probabilities of yes and no instances can be exponentially small. $\PP$ is a very powerful complexity class -- it contains $\NP$ and $\coNP$ (since it is trivially closed under complement) as well as $\BPPpath$. Furthermore, by Toda's theorem \cite{Toda91}, $\mathsf{P}^{\PP}$ contains the entire polynomial hierarchy.
Additionally Aaronson showed $\PP=\PostBQP$, the set of problems decidable by quantum algorithms equipped with postselection (the ability to discard all runs of an experiment which do not achieve an exponentially unlikely outcome). 
As a result, it is difficult to prove lower bounds against $\PP$. 


\medskip
\noindent \textbf{Our Results.} We answer Watrous' question by giving an oracle separating $\SZK$ from $\PP$. 
In fact, we prove something significantly stronger: our oracle construction separates $\NISZK$ from $\UPP$.\footnote{$\UPP$ is traditionally defined as an oracle complexity class, in which machines must output the correct answer with probability strictly greater than $1/2$,
and are charged for oracle queries but not for computation time. In this model, the gap between $1/2$ and the probability of outputting the correct answer can be \emph{arbitrarily} (in particular, superexponentially) small.}

	\begin{theo}
		\label{cor:niszk-pp} \label{mainfcor}
		There exists an oracle $\Oracle$ such that $\NISZK^{\Oracle} \not\subset \UPP^{\Oracle}$.
	\end{theo}

\subsection{Group 2: Limitations on the Power of Perfect Zero Knowledge}
\label{sec:pzkintro}
\medskip
\noindent \textbf{Motivation.}  
Much progress has been made on understanding the relationship between
natural variants of $\SZK$ \cite{okamoto1996relationships,goldreich1999can,Fischlin2002,Malka2015,LovettZhang16}.
For example, it is known that $\SZK=\co\SZK$ \cite{okamoto1996relationships}, and if $\NISZK=\co\NISZK$ then $\SZK=\NISZK=\co\NISZK$ \cite{goldreich1999can}. 
Additionally Lovett and Zhang~\cite{LovettZhang16} recently gave an oracle separation between $\NISZK$ and $\co\NISZK$ as well as $\SZK$ and $\NISZK$.
However, many questions remain open, especially regarding the power
of \emph{perfect} zero-knowledge proof systems. 


Many important $\SZK$ protocols, such as the ones for Graph Non-Isomorphism and 
Quadratic Nonresiduosity, are in fact $\PZK$ protocols. This illustrates the power of perfect zero knowledge. 
In this work, we are primarily concerned with studying the \emph{limitations} of perfect
zero knowledge. We are particularly interested in four questions:
Does $\SZK = \PZK$? What about their non-interactive variants, $\NISZK$ and $\NIPZK$? 
Is $\PZK$ closed under complement, the way that $\SZK$ is? What about $\NIPZK$?
Answering any of these questions in the negative would require showing $\mathsf{P} \neq \NP$,\footnote{$\mathsf{P} = \NP$ implies $\mathsf{P} = \PH$, and therefore $\SZK = \PTIME$.}
so it is natural to try to exhibit oracles relative to which $\SZK \neq \PZK$, $\NISZK \neq \NIPZK$, $\PZK \neq \coPZK$,
and $\NIPZK \neq \coNIPZK$.

\medskip \noindent \textbf{Additional Context.} In 1991, Aiello and H{\aa}stad \cite{AHPZKBPP} gave evidence that $\PZK$ contains
hard problems by creating an oracle relative to which $\PZK$
is not contained in $\BPP$. On the other hand, they also
gave an oracle that they \emph{conjectured} separates $\SZK$ from $\PZK$ (but were unable to prove this). 
Exhibiting such an oracle requires a technique that can tell the difference 
between zero simulation error ($\PZK$) and 
simulation to inverse exponential error ($\SZK$), and prior to our work, no such technique was known.
{The question of whether $\SZK=\PZK$ has been asked by Goldwasser \cite{GoldwasserVideo} as well.}
The analogous question for the non-interactive classes $\NISZK$ and $\NIPZK$ is also
well motivated, and was explicitly asked in recent work of Lovett and Zhang \cite{LovettZhang16}.

Determining whether variants of $\SZK$ satisfy the same closure properties as $\SZK$ is
natural as well: indeed, a main result of Lovett and Zhang \cite{LovettZhang16} 
is an oracle relative to which $\NISZK\neq\co\NISZK$. 

\medskip \noindent \textbf{Our Results.} We give oracles separating $\SZK$ from $\PZK$, $\NISZK$ from $\NIPZK$,
$\PZK$ from $\coPZK$, and $\NIPZK$ from $\coNIPZK$. 
The first two results answer the aforementioned questions raised by Aiello and H{\aa}stad {\cite{AHPZKBPP}}
(though our oracle is different from the candidate proposed by Aiello and H{\aa}stad),
and Lovett and Zhang {\cite{LovettZhang16}}. Along the way, we show that $\PZK$ is contained in $\PP$ in a relativizing manner --
this is in sharp contrast to $\SZK$ (see Theorem \ref{mainfcor}). 

\begin{theo}
For any oracle $\Oracle$, $\PZK^{\Oracle} \subseteq \PP^{\Oracle}$. In addition,
there exist oracles $\Oracle_1$ and $\Oracle_2$ such that $\SZK^{\Oracle_1} \not\subseteq \PZK^{\Oracle_1}$,
$\NISZK^{\Oracle_1} \not\subseteq \NIPZK^{\Oracle_1}$, $\PZK^{\Oracle_2} \not\subseteq \coPZK^{\Oracle_2}$, and 
$\NIPZK^{\Oracle_2} \not\subseteq \coNIPZK^{\Oracle_2}$.
 \label{thm:pzklimit}
\end{theo}

A summary of known relationships between complexity classes 
in the vicinity of $\SZK$, including the new results established in this work, is provided in Figure \ref{fig:complexdiagram}.
\begin{figure}
	\label{fig:inc-graph}
	\centering
\hspace*{-7em} 
	\begin{tikzpicture}[->,>=stealth',shorten >=1pt,auto,
	semithick,scale = 2.0]
	\tikzstyle{every state}=[draw=none,text=black,rectangle]
	\tikzstyle{contain old}=[thick]
	\tikzstyle{contain new}=[thick,color = red]
	\tikzstyle{equivalent old}=[thick,<->]
	\tikzstyle{separation old}=[thick,dashed]
	\tikzstyle{separation new}=[thick,dashed,color = red]
	\node [state] (SZK) at (-3, 2.66) {$\SZK$};
	\node [state] (NISZK) at (-3, 0.66) {$\NISZK$};
	\node [state] (HVPZK) at (-1, 2.66) {$\HVPZK$};
	\node [state] (NIPZK) at (-1, -0) {$\NIPZK$};
	\node [state] (coNIPZK) at (0.5, -0) {$\coNIPZK$};
	\node [state] (PP) at (-1, 4) {$\PP$};
	\node [state] (UPP) at (-1, 6) {$\UPP$};
	\node [state] (coHVPZK) at (0.5, 2.66) {$\coHVPZK$};
	\node [state] (PZK) at (-1, 1.33) {$\PZK$};
	\node [state] (coPZK) at (0.5, 1.33) {$\coPZK$};
	\node [state] (AMcapcoAM) at (-3, 4) {$\AM \cap \coAM$};
	\node [state] (AM) at (-3, 5) {$\AM$};
	\node [state] (PostBQP) at (0.5, 4) {$\PostBQP$};
	\node [state] (PH) at (-3, 6) {$\PH$};
	
	\draw[contain old] (NISZK) to [out=120,in=240] (SZK);
	\draw[contain old] (NIPZK) -- (PZK);
	\draw[contain old] (PZK) -- (HVPZK);
	\draw[contain old] (HVPZK) to [out=195, in=345] (SZK);
	\draw[separation new] (NISZK) to [out=140, in=160] (UPP);
	\draw[contain old] (PP) -- (UPP);
	\draw[contain old] (PostBQP) -- (UPP);
	\draw[contain old] (AM) -- (PH);
	\draw[contain new] (HVPZK) -- (PP);
	\draw[contain old] (SZK) -- (AMcapcoAM);
	\draw[contain old] (AMcapcoAM) -- (AM);
	\draw[separation new] (HVPZK) to [out=30, in=150] (coHVPZK);
	\draw[separation new] (coHVPZK) to [out=210, in=330] (HVPZK);
	\draw[separation new] (PZK) to [out=30, in=150] (coPZK);
	\draw[separation new] (coPZK) to [out=210, in=330] (PZK);
	\draw[separation new] (NIPZK) to [out=30, in=150] (coNIPZK);
	\draw[separation new] (coNIPZK) to [out=210, in=330] (NIPZK);
	\draw[equivalent old] (PP) -- (PostBQP);
	\draw[separation new] (NISZK) -- (NIPZK);
	\draw[separation new] (SZK) to [out=15, in=165] (HVPZK);
	\draw[separation new] (SZK) -- (PP);
	\draw[separation new] (SZK) -- (PZK);
	\draw[separation old] (AMcapcoAM) -- (PP);
	\draw[separation old] (SZK) to [out=-60,in=60] (NISZK);
	\draw[separation old] (PH) to [out=30,in=150] (UPP);
	\end{tikzpicture}
	
	\newcommand{\class}{\mathcal{C}}

	\caption{$\class_1 \to \class_2$ indicates $\class_1$ is contained in $\class_2$ respect to {\em every} oracle, and $\class_1 \dashrightarrow \class_2$ denotes that there is an oracle $\Oracle$ such that $\class_1^{\Oracle} \not\subset \class_2^{\Oracle}$. {\color{red}Red} indicates new results. Certain non-inclusions that are depicted are subsumed by other non-inclusions (e.g., $\NISZK$ not in $\UPP$ subsumes $\SZK$ not in $\PP$). We include some redundant arrows to facilitate comparison of our results to prior work.}
	\label{fig:complexdiagram}
\end{figure}


\subsection{Group 3: Communication Complexity}
\noindent \textbf{Motivation and Context.}
Paturi and Simon \cite{paturisimon} introduced the model of \emph{unbounded error communication complexity}, captured by the communication complexity class $\UPPcc$.\footnote{As is standard, given a query model $\mathsf{C}^{\mathsf{dt}}$ (or a communication model $\mathsf{C}^{\cc}$), we define a corresponding complexity class,
also denoted $\mathsf{C}^{\mathsf{dt}}$ (or $\mathsf{C}^{\cc}$), consisting of all problems that have polylogarithmic cost protocols in the model.} In this model, two parties with inputs $(x, y)$ execute a randomized communication protocol, and are only required to output $f(x, y)$ with probability strictly better than random guessing. Unbounded error communication protocols are extremely powerful, owing to this weak success criterion. 
In fact, $\UPPcc$ represents the frontier of our understanding of communication complexity: it is the most powerful communication model against which we know how to prove lower bounds. We direct the interested reader to \cite{landscape} for a thorough overview of communication complexity classes and their known relationships.
		
\medskip \noindent \emph{What Lies Beyond the Frontier?}
In an Arthur-Merlin game, a computationally-unbounded prover (Merlin) attempts to convince a
computationally-bounded verifier (Arthur) of the value of a given Boolean function on a given
input. The communication analogue of Arthur-Merlin games is captured by the communication
complexity class $\AMcc$. 
 
Many works have pointed to $\AMcc$ as one of the simplest communication models against which we do not know how to prove superlogarithmic lower bounds. Works attempting to address this goal include \cite{goositcs, landscape, CCMTV, Lokam01, LS09, PSS14, KP14b, KP14a, klauck11}. 
In fact, there are even simpler communication models against which we do not know how to prove lower bounds: it is known that $\NISZKcc \subseteq \SZKcc \subseteq \AMcc \cap \coAMcc \subseteq \Stwocc$, and we currently cannot prove lower bounds even against $\NISZKcc$.

Despite our inability to prove lower bounds against these classes, prior to our work it was possible that 
$\AMcc$ is actually contained in $\UPPcc$ (which, as described above, is a class against which we \emph{can} prove lower bounds). The prior works that had come closest to ruling this out were as follows.
	
	\begin{itemize}\setlength\itemsep{0.7em}
	\item $\AMcc \cap \coAMcc \not\subseteq \PPcc$. This was established (using a partial function) by Klauck \cite{klauck11}, who proved it by combining Vereschagin's analogous query complexity separation with Sherstov's pattern matrix method \cite{sherstov2011pattern}.
	\item $\Stwocc \not\subseteq \UPPcc$. This result was proved (using a total function) by Razborov and Sherstov \cite{razborovsherstov}. 
\end{itemize}

Based on this state of affairs, G{\"{o}}{\"{o}}s et al. \cite{landscape} explicitly posed the problem of showing that $\AMcc \cap \coAMcc \not\subseteq \UPPcc$.

\medskip
\noindent \textbf{Our Results.} In this work, we do even better than showing that $\AMcc \not \subseteq \UPPcc$. By ``lifting'' our oracle separation of $\NISZK$ and $\UPP$ to the communication setting, we show (using a partial function) that $\NISZKcc \not \subseteq \UPPcc$. Hence, if $\UPPcc$ is taken to represent the frontier of our understanding of communication complexity, our result implies that $\NISZKcc$ (and hence $\AMcc$) is truly beyond the frontier. This also answers the question of 
G{\"{o}}{\"{o}}s et al. \cite{landscape}.

\begin{theo} \label{thm:commintro}
There is a (promise) problem in $\NISZKcc$ that is not in $\UPPcc$.
\end{theo}

\subsection{Other Consequences of Our Results}
In addition to the above oracle and communication separations, our results have a number of applications in other areas of theoretical computer science. For example, our results have implications regarding the power of complexity classes capturing the power of quantum computing with ``more powerful" modified versions of quantum mechanics \cite{ABFL16, aaronsondqp},
imply limitations on the Polarization Lemma of Sahai and Vadhan \cite{sahai2003complete}, yield novel lower bounds for certain forms of property testing algorithms, and imply upper bounds for \emph{streaming interactive proofs} \cite{cty,  CCMTV}. These results are described in detail in Section \ref{sec:consequences}.

\subsection{Overview of Our Techniques}
\subsubsection{Oracle Separation of $\NISZK$ and $\UPP$ (Proof Overview for Theorem \ref{mainfcor})} 
To describe our methods, it is helpful to introduce the notions of approximate degree and threshold degree, both of which are measures of Boolean function complexity that capture
the difficulty of point-wise approximation by low-degree polynomials. The $\eps$-approximate degree of a function $f \colon \{0, 1\}^n \rightarrow \{0, 1\}$, denoted $\dega_{\eps}(f)$, is the least degree of a real polynomial that point-wise approximates $f$ to error $\eps$. The threshold degree of $f$, denoted $\degsh(f)$, 
is the least degree of a real polynomial that agrees in sign with $f$ at all points. It is easy to see that threshold degree is equivalent to the limit of the approximate degree as the error parameter $\eps$ approaches 1/2 from below. 

A recent and growing line of work has addressed a variety of open problems in complexity theory by establishing various forms of hardness amplification for approximate degree. Roughly speaking, these results show how to take a function $f$ which is hard to approximate by degree $d$ polynomials to error $\eps=1/3$, and turn $f$ into a related function $F$ that is hard to approximate by degree $d$
polynomials even when $\eps$ is very close to 1/2. In most of these works, $F$ is obtained from $f$ by block-composing $f$ with a ``hardness-amplifying function'' $g$. We denote such a block-composition by $g(f)$.

The technical core of our result lies in establishing a new form of hardness amplification for approximate degree. Specifically, let $g$ be the partial function $\Gapmajority_n \colon \{0,1\}^n \rightarrow \{0,1\}$ (throughout this introduction, whenever necessary, we use subscripts after function names to clarify the number of variables on which the function is defined). Here $\Gapmajority$ is the gapped majority function, defined,  for some $1\geq \delta>0.5$, to be 1 if $\geq\delta$ fraction of its inputs are 1, to be 0 if $\geq\delta$ fraction of its inputs are 0, and to be undefined otherwise (in this introduction, we will ignore the precise choice of $\delta$ that we use in our formal results).\footnote{We clarify that if $f$ is a partial function then $\Gapmajority_n(f)$ is technically not a composition of functions, since for some inputs $x=(x_1, \dots, x_n)$ on which $\Gapmajority_n(f)$ is defined, there may be values of $i$ for which $x_i$ {is} outside of the domain of $f$. See Section \ref{sec:gapmajdef} for further discussion of this point.}

\begin{theo}(Informal) \label{thm:great}
Let $f \colon \{0,1\}^M \rightarrow \{0,1\}$. Suppose that $\dega_{1/3}(f) \geq d$. Define $F \colon \{0,1\}^{n \cdot M} \rightarrow \{0,1\}$  via $F=\Gapmajority_n(f).$ Then $\degsh(F) = \Omega(\min(d, n))$.
\end{theo}

In our main application of Theorem \ref{thm:great}, we apply the theorem to a well-known (partial) function $f=\Collision_M$ called the Collision problem. This function is known to have approximate degree $\tilde{\Omega}(M^{1/3})$, so Theorem \ref{thm:great} implies that $F:=\Gapmajority_{M^{1/3}}( \Collision_M)$ has threshold degree $\tilde{\Omega}(M^{1/3})$. Standard results then imply that the $\UPP$ \emph{query complexity} of $F$ is $\tilde{\Omega}(M^{1/3})$ as well. That is, $F \not\in \UPPdt$. 

\begin{cor}[Informal]
Let $m=M^{4/3}$, and define $F \colon \{0,1\}^m \rightarrow \{0,1\}$ via $F:=\Gapmajority_{M^{1/3}}( \Collision_M)$. Then $\UPPdt(F) = \tilde{\Omega}(m^{1/4})$. 
\end{cor}

Moreover, as we show later, $\Gapmajority_{M^{1/3}}(\Collision_M)$ is in $\NISZKdt$. Hence, we obtain a separation between $\NISZKdt$ and $\UPPdt$. The desired oracle separating $\NISZK$ from $\UPP$ follows via
standard methods.

\medskip \noindent \textbf{Comparison of Theorem \ref{thm:great} to Prior Work.}
The hardness amplification result from prior work that is most closely related to Theorem \ref{thm:great} is due to Sherstov \cite{sherstov2014breaking}. Sherstov's result makes use of a notion known as (positive) one-sided approximate degree \cite{sherstov2014breaking, bun2015hardness}. Positive one-sided approximate degree is a measure that is intermediate between approximate degree and threshold degree -- the positive one-sided approximate degree of $f$, denoted $\odegap_{\eps}(f)$, is always at most as large as the approximate degree of $f$ but can be much smaller, and it is always at least as large as the threshold degree of $f$  but can be much larger (see Section \ref{sec:test} for a formal definition of positive one-sided approximate degree).\footnote{The notion of positive one-sided approximate degree treats inputs in $f^{-1}(1)$ and $f^{-1}(0)$ asymmetrically. There is an analogous notion called negative one-sided approximate degree that reverses the roles of $f^{-1}(1)$ and $f^{-1}(0)$ \cite{thalericalp, varun}. Our use of the positive vs. negative terminology follows prior work \cite{thalericalp, varun} -- other prior works \cite{sherstov2014breaking, bun2015hardness}
only used negative one-sided approximate degree, and referred to this complexity measure without
qualification as one-sided approximate degree.
In this paper, we exclusively use the notion of positive one-sided approximate degree.} 

\begin{theo}[Sherstov] \label{thm:sherstov}
Let $f \colon \{0,1\}^M \rightarrow \{0,1\}$. Suppose that $\odegap_{1/3}(f) \geq d$. Define $F \colon \{0,1\}^{n \cdot M} \rightarrow \{0,1\}$  via $F=\AND_n(f).$ Then $\degsh(F) = \Omega(\min(d, n))$.\footnote{Sherstov stated his result for $\OR_n(f)$ under the assumption that $f$ has large \emph{negative} one-sided approximate degree. Our statement of Theorem \ref{thm:sherstov} is the equivalent result under the assumption that $f$ has large positive one-sided approximate degree.} 
\end{theo}

There are two differences between Theorems \ref{thm:great} and \ref{thm:sherstov}. The first is that the hardness-amplifier in Theorem \ref{thm:great} is $\Gapmajority$, while in Theorem \ref{thm:sherstov} it is $\AND$. $\Gapmajority$ is a ``simpler'' function than $\AND$ in the following sense: block-composing $f$ with $\Gapmajority$ preserves membership in complexity classes such as $\NISZKdt$ and $\SZKdt$; this is not the case for $\AND$, as $\AND$ itself is not in $\SZKdt$. This property is essential for us to obtain threshold degree lower
bounds even for functions that are in $\NISZKdt$.


The second difference is that Theorem \ref{thm:great} holds under the assumption that $\dega_{1/3}(f) \geq d$, while Theorem \ref{thm:sherstov} makes the stronger assumption that $\odegap_{\eps}(f) \geq d$. 
While we do not exploit this second difference in our applications, ours is the first form of hardness amplification that works for approximate degree rather than one-sided approximate degree. This property has already been exploited in subsequent work \cite{btlatest}. 

\medskip \noindent \textbf{Proof Sketch for Theorem \ref{thm:great}.} 
A \emph{dual polynomial} is a dual solution to an appropriate
linear program capturing the threshold degree of any function. 
Specifically, for a (partial) function $f$ defined on a subset of $\{0, 1\}^n$, a dual polynomial witnessing the fact that $\dega_{\eps}(f) \geq d$
is a function $\psi \colon \{0,1\}^n \rightarrow \R$ that satisfies the following three properties. 

\begin{itemize}
\item[(a)] $\psi$ is uncorrelated with all polynomials $p$ of total degree at most $d$. That is, for any $p \colon \{0,1\}^n \rightarrow \R$ such that $\deg(p) \leq d$, it holds that $\sum_{x \in \{0,1\}^n} \psi(x) \cdot p(x) = 0$. We refer to this property by saying that $\psi$ has \emph{pure high degree} $d$.
\item[(b)] $\psi$ has $\ell_1$ norm equal to 1, i.e., $\sum_{x \in \{0,1\}^n} |\psi(x)|=1$.
\item[(c)] $\psi$ has correlation at least $\eps$ with $f$. That is, if $D$ denotes the domain on which $f$ is defined, then $\sum_{x \in D} \psi(x) \cdot f(x) - \sum_{x \in \{0,1\}^n \setminus D} |\psi(x)| > \eps$.
\end{itemize}

It is not hard to see that a dual witness for the fact that $\degsh(f)\geq d$ is a function $\psi$ satisfying Properties (a) and (b) above,
that additionally is \emph{perfectly} correlated with $f$. That is, $\psi$ additionally satisfies
\begin{equation} \label{eq:orthdist}\sum_{x \in D} \psi(x) \cdot f(x) - \sum_{x \in \{0,1\}^n \setminus D} |\psi(x)| = 1.\end{equation} In this case, $\psi \cdot f$ is non-negative,
 and is referred to as an \emph{orthogonalizing distribution} for $f$.
 
 \label{sec:proofoverview}
 
 We prove Theorem \ref{thm:great} by constructing an explicit orthogonalizing distribution for $\Gapmajority_n(f)$.
Specifically, we show how to take a dual polynomial witnessing the fact that $\dega_{1/3}(f) \geq d$,
and turn it into an orthogonalizing distribution witnessing the fact that $\degsh(F) = \Omega(\min(d, n))$.
 
 Our construction of an orthogonalizing distribution for 
$\Gapmajority_n( f )$ is inspired by and reminiscent of Sherstov's construction of an orthogonalizing distribution for $\AND_n( f)$ \cite{sherstov2014breaking}, which in turn builds on a dual polynomial for $\AND_n(f)$ constructed by Bun and Thaler \cite{bun2015hardness}. In more detail,
Bun and Thaler constructed a dual polynomial $\psi_{BT}$ of pure high degree $d$ that had correlation
$1-2^{-n}$ with $\AND_n( f)$. Sherstov's dual witness
was defined as $\psi_{BT} + \psi_{corr}$, where $\psi_{corr}$ is an \emph{error-correction term} that also has pure high degree $\Omega(d)$. The purpose of $\psi_{corr}$
is to
``zero-out'' $\psi_{BT}$ at all points where $\psi_{BT}$ differs in sign from $f$,
without affecting the sign of $\psi_{BT}$ on any other inputs. 

Naively, one might hope that $\psi_{BT} + \psi_{corr}$ is also a dual witness to the fact 
that $\degsh(\Gapmajority_n( f))$ is large. Unfortunately, this is not the case, as it does
not satisfy Equation \eqref{eq:orthdist} with respect to $\Gapmajority_n ( f)$.
It is helpful to think of this failure as stemming from two issues. 
First, $\psi_{BT} + \psi_{corr}$ places non-zero weight
on many inputs on which $\Gapmajority_n ( f)$ is undefined (i.e., on inputs
for which fewer than $\delta n$ copies of $f$ evalauate to 1 and fewer than $\delta n$ copies of $f$ evaluate to 0). 
Second, there are inputs on which $\Gapmajority_n(f)$
is defined, yet $\AND_d(f)$ does not agree with $\Gapmajority_n(f)$.

To address both of these issues, we add a \emph{different} error-correction term $\psi_{corr}'$ of pure high degree $\tilde{\Omega}(\min(n, d))$ to $\psi_{BT}$. Our correction term does not just zero out the value of $\psi_{BT}$ on inputs on which it disagrees in sign with $\AND_n( f)$, but also zeros it out on inputs for which $\Gapmajority_n( f)$ is undefined, and 
on inputs on which $\AND_n( f)$ does not agree with $\Gapmajority_n ( f)$. 

Moreover, we show that adding $\psi_{corr}'$ does not affect the sign of $\psi_{BT}$ on other inputs -- achieving this requires some new ideas in both the definition $\psi_{corr}'$ and its analysis. 
Putting everything together, we obtain a dual witness  $\psi_{BT} + \psi'_{corr}$ showing that $\degsh(\Gapmajority_n ( f)) = \Omega(\min(n, d))$.

\subsubsection{Limitations on the Power of Perfect Zero Knowledge (Proof Overview For Theorem \ref{thm:pzklimit})}
We begin the proof of Theorem \ref{thm:pzklimit}
by showing that $\HVPZK$ (\emph{honest verifier} perfect zero knowledge) is contained in $\PP$ in a relativizing manner (see Section \ref{sec:oracleseps}).
Since the inclusions $\PP \subseteq \UPP$, $\NIPZK \subseteq \HVPZK$, $\PZK \subseteq \HVPZK$, and $\NISZK \subseteq \SZK$ hold with respect to any oracle, 
this means that our oracle separating $\NISZK$ from $\UPP$ (Theorem \ref{mainfcor}) also separates $\SZK$ from $\PZK$
and $\NISZK$ from $\NIPZK$.

We then turn to showing that $\PZK$ and $\NIPZK$ are not closed under complement with respect to some oracle. Since the proofs
are similar, we focus on the case of $\PZK$ in this overview.

Since both $\PZK$ and $\co\PZK$ are contained in $\PP$ with respect to any oracle, our oracle separation of $\NISZK$ from $\PP$ (Theorem \ref{mainfcor}) does not imply an oracle relative to which $\PZK\neq\co\PZK$.
{Instead, to obtain this result we prove a new amplification theorem for one-sided approximate degree.}
{Using similar techniques as Theorem \ref{thm:great}, we show that if  $f$ has high positive one-sided approximate degree, then block-composing $f$ with the gapped AND function yields a function with high threshold degree. }
Here $\Gapand$ is partial function that outputs $1$ if all inputs are $1$, outputs $0$ if at least a $\delta$ fraction of inputs are $0$, and is undefined otherwise.
\begin{theo}(Informal) \label{thm:gapandinformal}
Let $f \colon \{0,1\}^M \rightarrow \{0,1\}$. Suppose that $\odegap_{1/3}(f) \geq d$. Then $\degsh(\Gapand_n(f)) \newline= \Omega(\min(d, n))$.
\end{theo}

We then show that (a) $\PZKdt$ is closed under composition with $\Gapand$ and (b) there is a function $f$ in $\PZKdt$ whose \emph{complement} $\bar{f}$ has high-one sided positive one-sided approximate degree.  
If $\PZKdt$ were closed under complement, then $\bar{f}$ would be in $\PZKdt$. 
By amplifying the hardness of $\bar{f}$ using Theorem \ref{thm:gapandinformal}, we obtain 
 a problem that is still in $\PZKdt$ (this holds by property (a)) yet outside of $\PPdt$ (this holds by property (b), together with Theorem \ref{thm:gapandinformal}). This is easily seen to contradict the fact $\PZK$ is in $\PP$ relative to all oracles. Hence, $\bar{f}$ is a function in $\coPZKdt$ that is not in $\PZKdt$, and standard techniques translate
 this fact into an oracle separating $\coPZK$ from $\PZK$. 
We provide details of these results in Section \ref{sec:oracleseps}. 

\subsubsection{Lifting to Communication Complexity: Proof Overview For Theorem \ref{thm:commintro}}
 To extend our separation between $\NISZK$ and $\UPP$ to the world of communication complexity,
we build on recently developed methods of Bun and Thaler \cite{bt16}, who themselves used and generalized the breakthrough work
of Razborov and Sherstov \cite{razborovsherstov}. 
Razborov and Sherstov showed that if $F$ has high threshold degree and this is witnessed by an orthogonalizing distribution that satisfies an additional smoothness condition, then $F$ can be transformed into a related function 
$F'$ that has high $\UPPcc$ complexity (specifically, $F'$ is obtained from $F$ via the \emph{pattern matrix method} introduced in \cite{sherstov2011pattern}). So in order to turn $\Gapmajority( \Collision)$ into a function with high $\UPPcc$ complexity, it is enough to give a \emph{smooth} orthogonalizing distribution for $F$. 

Bun and Thaler \cite{bt16} showed how to take the dual witness Sherstov constructed for $\OR (f)$ in the proof of Theorem \ref{thm:sherstov} and smooth it out, assuming the inner function $f$ satisfies some modest additional conditions. Fortunately, a variant of $\Collision$ called the \textsf{Permutation  Testing Problem} ($\PTP$ for short) satisfies these additional conditions, and
since our construction of an orthogonalizing distribution for $\Gapmajority (\PTP)$  is reminiscent of Sherstov's orthogonalizing distribution for $\OR(f)$,
we are able to modify the methods of Bun and Thaler to smooth out our dual witness for $\Gapmajority( \PTP)$. Although there are many technical details to work through, adopting the methodology of Bun and Thaler to our setting does not require substantially new ideas, and we do not consider it to be a major technical contribution of this work. Nonetheless, it does require the careful management of various subtleties arising from our use of promise problems as opposed to total Boolean functions, and our final communication lower bound inherits many of the advantages of our
Theorem \ref{thm:great} relative to prior work (such as applying to functions
with high approximate degree rather than high one-sided approximate degree).

\subsection{Other Works Giving Evidence for the Hardness of $\SZK$}
\label{sec:priorhardness}
As mentioned in Section \ref{sec:pzkintro}, Aiello and H{\aa}stad showed that $\PZK$ (and also $\SZK$) is not contained in $\BPP$ relative to some oracle \cite{AHPZKBPP}.
Agrawal et al. later used similar techniques to show that $\SZK$ is not contained in the class $\SRE$ (which can be viewed as a natural generalization of $\BPP$) relative to some oracle \cite{AIKP15}.
Aaronson \cite{aaronson2002quantum} gave an oracle relative to which $\SZK$ is not contained in $\BQP$ -- and therefore quantum computers cannot break $\SZK$-hard cryptosystems in a black-box manner. 
Building on that work, Aaronson \cite{aaronson2012impossibility} later gave oracle separations against the class $\QMA$ (a quantum analogue of $\NP)$ and the class $\AOPP$ (a class intermediate between $\QMA$ and $\PP$).
Therefore even quantum proofs cannot certify $\SZK$ in a black-box manner\footnote{{Note, however, that oracle separations do not necessarily imply the analogous separations in the ``real world" -- see \cite{Barak01} and \cite{CCGHHRR94} for instances in which the situation in the presence of oracles is far from the situation in the real world.}}.

Until recently, the lower bound most closely related to our oracle separation of $\NISZK$ and $\UPP$ (cf. Theorem \ref{mainfcor}) was Vereschagin's result from 1995, which gave an oracle relative to which $\AM \cap \coAM$ is not contained in $\PP$ \cite{Vereshchagin95}. 
Our result is an improvement on Vereschagin's because the inclusions
$\NISZK \subseteq \SZK \subseteq \AM \cap \coAM$ can be proved in a relativizing manner (cf. Figure \ref{fig:complexdiagram}).
It also generalizes Aaronson's oracle separation between $\SZK$ and $\AOPP$ \cite{aaronson2012impossibility}. 

Vereschagin  \cite{Vereshchagin95} also reports that Beigel claimed a simple proof of the existence of a function $f$ that is in the query complexity class $\AMdt$, but is not in the query complexity class $\UPPdt$. Our result improves 
on Beigel's in two regards. First, since $\NISZKdt \subseteq \AMdt$, separating $\NISZKdt$ from $\UPPdt$ is more difficult than separating $\AMdt$ from $\UPPdt$. Second, Beigel 
only claimed a superlogarithmic lower bound on the $\UPPdt$ query complexity of $f$, while we give a polynomial lower bound.

Theorem \ref{mainfcor} also improves on very recent work of Chen \cite{lijiepszk, lijiepszkqszk}, which gave a query separation between the classes $\mathsf{P}^\SZK$ and $\PP$.

\section{Technical Preliminaries}
\label{sec:prelim}

\subsection{Complexity Classes}
\label{sec:comp-classes}

\paragraph{Notation.} In an interactive proof $(P,V)$ where $P$ is the prover and $V$ is the verifier, we denote by $(P,V)(x)$ the random variable corresponding to the transcript of the protocol on input $x$. For distributions $D_0$ and $D_1$, $||D_0 - D_1||$ denotes the Total Variational Distance between them $\left(||D_0- D_1|| = \frac{1}{2}|D_0 - D_1|_1 \right)$.

\begin{defi}[Honest Verifier Statistical Zero Knowledge]
    \label{def:szkclass}
    A promise problem $L = (L_Y, L_N)$ is in $\HVSZK$ if there exists a tuple of Turing machines $(P, V, S)$, where the \emph{verifier} $V$ and \emph{simulator} $S$ run in probabilistic polynomial time, satisfying the following:
    \begin{itemize}
      \item $(P,V)$ is an interactive proof for $L$ with negligible completeness and soundness errors.
      \item For any $x\in L_Y$,
        \begin{align*}
          \|S(x) - (P,V)(x)\| \leq \negl(|x|)
        \end{align*}
    \end{itemize}
\end{defi}

\label{sec:niszkdef}
\begin{defi}[Non-Interactive SZK]
    \label{def:niszkclass}
    A promise problem $L = (L_Y, L_N)$ is in $\NISZK$ if there exist a tuple of Turing machines $(P, V, S)$, where the \emph{verifier} $V$ and \emph{simulator} $S$ run in probabilistic polynomial time, satisfying the following:
    \begin{itemize}
      \item $(P,V)$ is an interactive proof for $L$ with negligible completeness and soundness errors, with the following additional conditions:
        \begin{enumerate}
          \item $P$ and $V$ both have access to a long enough common random string.
          \item The interactive proof consists of a single message from $P$ to $V$.
        \end{enumerate}
      \item For any $x\in L_Y$,
        \begin{align*}
          \|S(x) - (P,V)(x)\| \leq \negl(|x|)
        \end{align*}
    \end{itemize}
\end{defi}

The definitions of these classes in the presence of an oracle are the same, except that $P$, $V$, and $S$ all have access to the oracle.

It is easy to see that $\NISZK$ is contained in $\HVSZK$, even in the presence of oracles. The class $\SZK$ is defined to be almost the same as $\HVSZK$, except for the stipulation that for any verifier (even one that deviates from the prescribed protocol), there is a simulator that simulates the prover's interaction with that verifier. It was shown in \cite{GSV98} that $\SZK$ is equal to $\HVSZK$, and their proof continues to hold in the presence of any oracle. For this reason, $\NISZK$ is also contained in $\SZK$ in the presence of any oracle, and we shall be implicitly making use of this fact at several points where we state corollaries for $\SZK$ instead of $\HVSZK$.

\subsection{Approximate Degree, Threshold Degree, and Their Dual Characterizations}
\label{sec:test}
	We first recall the definitions of approximate degree, positive one-sided approximate degree, and threshold degree for partial functions.
	
	\begin{defi} \label{maindef}
		Let $D \subseteq \{0, 1\}^M$, and let $f$ be a function mapping $D$ to $\{0,1\}$.
		\begin{itemize}
		\item The \emph{approximate degree} of $f$ with approximation constant $0\le\epsilon<1/2$, denoted $\dega_\epsilon(f)$, is the least degree of a real polynomial $p \colon \{0, 1\}^M \to \R$ such that $|p(x)-f(x)| \le \epsilon$ when $x \in D$, and $|p(x)| \le 1+\epsilon$ for all $x\not\in D$. We refer to such a $p$ as an \emph{approximating polynomial} for $f$. We use $\dega(f)$ to denote $\dega_{1/3}(f)$.
		
		\item The \emph{threshold degree} of $f$, denoted $\degsh(f)$, is the least degree of a real polynomial $p$ such that $p(x) > 0 $ when $f(x) = 1$, and $p(x) < 0$ when $f(x) = 0$.
		
		\item The \emph{postive one-sided approximate degree} of $f$ with approximation constant $0\le\epsilon<1/2$, denoted $\odegap_\epsilon(f)$, is the least degree of a real polynomial $p$ such that $|p(x) - 1| \le \epsilon$ for all $x \in f^{-1}(1)$, and $p(x) \le \epsilon$ when $x \in f^{-1}(0)$. We refer to such a $p$ as a \emph{positive one-sided approximating polynomial} for $f$. We use $\odegap(f)$ to denote $\odegap_{1/3}(f)$.		
		\end{itemize}
	\end{defi}
	
\medskip \noindent \textbf{Remark}.
We highlight the following subtlety in Definition \ref{maindef}: an approximating polynomial for a partial function $f$ is required to be bounded in absolute value even outside of the domain $D$ on which $f$ is defined, yet this is not required of a one-sided approximating polynomial for $f$. The reason we choose to require an approximating polynomial to be bounded outside of $D$ is to ensure that the $\Collision$ function (defined later in Section \ref{sec:collision}) has large approximate degree. 
	
	\medskip
	
	There are clean dual characterizations for each of the three quantities defined in Definition \ref{maindef}. We state these characterizations without proof, and direct the interested reader to \cite{sherstov2015power,sherstov2014breaking,bun2015dual} for details.
	
	For a function $\psi \colon \{0, 1\}^M \to \R$, define the $\ell_1$ norm of $\psi$ by $\|\psi\|_1 = \sum_{x \in \{0,1\}^M} |\psi(x)|$. 
	If the support of a function $\psi \colon \{0,1\}^M \to \R$ is (a subset of) a set $D \subseteq \{0, 1\}^M$, we will write $\psi \colon D \to \R$. 	
	For
	functions $f, \psi \colon D \to \R$, denote their inner product by $\langle f, \psi \rangle := \sum_{x \in D} f(x) \psi(x)$. 
	We say that a function $\psi \colon \{0, 1\}^M \to \R$ has \emph{pure high degree $d$}
	if $\psi$ is uncorrelated with any polynomial $p \colon \{0, 1\}^M \to \R$ of total degree at most $d$, i.e., if
	$\langle \psi, p \rangle = 0$. 
	
	\begin{theo}\label{theo:dual-dega}
		
		Let $f:D \to \{0,1\}$ with $D \subseteq \{0,1\}^M$ be a partial function and $\epsilon$ be a real number in $[0,1/2)$. $\dega_\epsilon(f) > d$ if and only if there is a real function $\psi: \{0,1\}^M \to \R$ such that:
		\begin{enumerate}
			\item (Pure high degree): $\psi$ has pure high degree of $d$.
			\item (Unit $\ell_1$-norm): $\|\psi\|_{1} =1$. 
			\item (Correlation): $\sum_{x \in D} \psi(x) f(x) - \sum_{x\not\in D} |\psi(x)| > \epsilon$.
		\end{enumerate}
	\end{theo}

	\begin{theo}\label{theo:dual-degsh}
		
		Let $f:D \to \{0,1\}$ with $D \subseteq \{0,1\}^M$ be a partial function. $\degsh(f) > d$ if and only if there is a real function $\psi: D \to \R$ such that:
		\begin{enumerate}
			\item (Pure high degree): $\psi$ has pure high degree of $d$.
			\item (Sign Agreement): $\psi(x) \ge 0$ when $f(x) = 1$, and $\psi(x) \le 0$ when $f(x) = 0$.
			\item (Non-triviality): $\|\psi\|_1 > 0$.
		\end{enumerate}
	\end{theo}
	
	\begin{theo}\label{theo:dual-odega}
		
		Let $f:D \to \{0,1\}$ with $D \subseteq \{0,1\}^M$ be a partial function and $\epsilon$ be a constant in $[0,1/2)$. $\odegap_\epsilon(f) > d$ if and only if there is a real function $\psi: D \to \R$ such that:
		\begin{enumerate}
			\item (Pure high degree): $\psi$ has pure high degree of $d$.
			\item (Unit $\ell_1$-norm): $\|\psi\|_{1} =1$. 
			\item (Correlation): $\langle \psi, f \rangle > \epsilon$.
			\item (Negative Sign Agreement): $\psi(x) \le 0$  whenever $f(x) = 0$.
		\end{enumerate}
	\end{theo}

\subsection{$\PPdt$ and $\UPPdt$}

Now we define the two natural analogues of $\PP$ complexity in the query model.

\begin{defi}  Let $f:D \to \{0,1\}$ with $D \subseteq \{0,1\}^M$ be a partial function. Let $\mathcal{T}$ be a randomized decision tree which computes $f$ with a probability better than $1/2$. Let $\alpha$ be the maximum real number such that
	
	$$
	\min_{x \in D} \Pr[\text{$\mathcal{T}$ outputs $f(x)$ on input $x$}] \ge \frac{1}{2} + \alpha. 
	$$
	
	Then we define the $\PP$ query cost of $\mathcal{T}$ for $f$ to be $\PPdt(\mathcal{T};f) = C(\mathcal{T};f) + \log_2(1/\alpha)$, where $C(\mathcal{T};f)$ denotes the maximum number of queries $\mathcal{T}$ incurs on an input in the worst case. We define $\UPPdt(\mathcal{T};f) = C(\mathcal{T};f)$. Observe that $\UPPdt(\mathcal{T};f)$ is the same as $\PPdt(\mathcal{T};f)$, except that the advantage $\alpha$ of the randomized decision tree over random guessing is not incorporated into $\UPPdt(\mathcal{T};f)$. We define $\PPdt(f)$ (respectively, $\UPPdt$) as the minimum of $\PPdt(\mathcal{T};f)$ (respectively, $\UPPdt(\mathcal{T};f)$) over all $\mathcal{T}$ that computes $f$ with a probability better than $1/2$.
\end{defi}

$\PPdt$ is closely related to approximate degree with error very close to $1/2$. We have the following well-known relationship between them.

\begin{lemma}\label{lm:dega-to-ppdt}
	Let $f:D \to \{0,1\}$ with $D \subseteq \{0,1\}^M$ be a partial function. Suppose $\dega_{1/2-2^{-d}}(f) > d$ for some positive integer $d$. Then $\PPdt(f) > d/2$.
\end{lemma}

Meanwhile, $\UPPdt$ is exactly characterized by threshold degree.

\begin{lemma}\label{lm:tdeg-to-uppdt}
	Let $f:D \to \{0,1\}$ with $D \subseteq \{0,1\}^M$ be a partial function. Then $\UPPdt(f)  = \degsh(f)$.
\end{lemma}

\subsection{Gap Majority and Gap AND}
\label{sec:gapmajdef}
In this subsection we introduce a transformation of partial functions which will be used in this paper.

\newcommand{\nyes}{n_{\mathsf{Yes}}}
\newcommand{\nno}{n_{\mathsf{No}}}

\begin{defi}\label{defi:gap-majority}
	Let $f:D \to \{0,1\}$ with $D \subseteq \{0,1\}^M$ be a partial function and $n$ be a positive integer, $0.5< \epsilon \le 1$ be a real number. We define the gap majority version of $f$, denoted by $\GapMaj{f}{n}{\epsilon}$, as follows:
	
	Given an input $x=(x_1,x_2,\dotsc,x_n) \in \{0,1\}^{M\cdot n}$, we define $\nyes(x):= \sum_{i=1}^{n} \indicator_{x_i \in D \wedge f(x_i) = 1}$ and  \newline $\nno(x) := \sum_{i=1}^{n} \indicator_{x_i \in D \wedge f(x_i) = 0}$. 
	Then
	$$
	\GapMaj{f}{n}{\epsilon}(x) = \begin{cases}
	1 \quad& \text{when $\nyes(x) \ge \epsilon\cdot n$}\\
	0 \quad& \text{when $\nno(x) \ge \epsilon\cdot n$}\\
	\text{undefined} \quad& \text{otherwise}\\
	\end{cases}
	$$
	
	Note that even on inputs $x$ for which $\GapMaj{f}{n}{\epsilon}(x)$ is defined, there may be some values of $i$ for which $x_i$ is not in $D$. 
	For brevity, we will occasionally write $\Gapmajority(f)$ when $n$ and $\eps$ are clear from context.
\end{defi}

We also define the $\Gapand$ function. This is a partial function that agrees with the total function $\AND$ wherever it is defined.

\begin{defi}
	Let $n$ be a positive integer, $0 < \epsilon < 1$ be a constant. We define the Gapped $\AND$ function, $\Gapand_{n,\epsilon} : D \to \{0,1\} $ with $D \subseteq \{0,1\}^n$, as the function that outputs $1$ if all inputs are $1$; outputs $0$ if at least $\epsilon \cdot n$ inputs are $0$; and is undefined otherwise.
	
	For a partial function $f:D \to \{0,1\}$ with $D \subseteq \{0,1\}^M$, we define
	 $\Gapand_{n, \eps}(f)$ to be a true block-composition of partial functions, i.e., 
	 $\Gapand_{n, \eps}(f)(x_1, \dots, x_n)=\Gapand_{n, \eps}(f(x_1), \dots, f(x_n))$ whenever the right hand side of the equality is defined, and $\Gapand_{n, \eps}(f)$ is undefined otherwise.
	 \end{defi}
	

\begin{rem}
Note that $\Gapmajority_{n, \eps}(f)$ is not technically a block-composition of partial functions, since $\Gapmajority_{n, \eps}(f)(x_1, \dots, x_n)$ is defined even on some inputs for which some $f(x_i)$ is not defined.
\end{rem}

\subsection{Problems}
	\label{sec:collision}
	
We now recall the \PCollision\ problem. This problem interprets its input as a function 
$f$ mapping $[n]$ to $[n]$, and the goal is to decide whether the input is a permutation or is $2$-to-$1$, promised that one of them is the case. We need a slightly generalized version, which asks to distinguish between permutations and $k$-to-$1$ functions.

\begin{defi}[\PCollision\ problem]
	Fix an integer $k \geq 2$, and assume for simplicity that $n$ is a power of $2$. The partial function $\Collision_n^k$ is defined on a subset of $\{0, 1\}^{n \log n}$. It interprets its input as specifying
	a function $f \colon [n] \to [n]$ in the natural way, and evaluates to $1$ if $f$ is a permutation, 0 if $f$ is a $k$-to-$1$ function, and is undefined otherwise. 
When $k$ and $n$ are clear from context, we write $\Collision$ for brevity.
\end{defi}

This problem admits a simple $\SZK$ protocol in which the verifier makes only $\polylog(n)$ queries to the input. Specifically, the verifier executes the following sub-protocol $\polylog(n)$ times: the verifier chooses a random $i \in [n]$, makes a single query to learn $f(i)$, sends $f(i)$ to the prover, and rejects if the prover fails to respond with $i$. It is easy to see that the sub-protocol has perfect completeness, constant soundness error, and is perfect zero knowledge. Because the sub-protocol is repeated $\polylog(n)$ times, the total soundness error is negligible.

In 2002, Aaronson \cite{aaronson2002quantum} proved the first non-constant lower
bound for the $\Collision_n^2$ problem: namely, any bounded-error quantum algorithm to solve
it needs $\Omega(n^{1/5})  $\ queries to $f$. Aaronson and
Shi \cite{aaronson2004quantum} subsequently improved the lower bound to $\Omega(
n^{1/3})$, for functions $f:\left[n\right]  \rightarrow\left[
3n/2\right]  $; then Ambainis \cite{ambainis2005polynomial} and Kutin \cite{kutin2005quantum}
proved the optimal $\Omega(n^{1/3})  $\ lower bound for functions
$f:\left[n\right]  \rightarrow\left[n\right]  $.

We need a version of the lower bound that makes explicit the dependence on $k$ and $\eps$.

\begin{theo}[Implicit in Kutin~\cite{kutin2005quantum}]\label{theo:kutin-lowb}
	$\dega_{\epsilon}(\Collision_n^k) = \Omega(\sqrt[3]{(1/2-\epsilon)\cdot n/k})$ for any $0<\epsilon<1/2$ and $k | n$.
\end{theo}

See also \cite{bun2015dual} for a direct constructive proof (using Theorem~\ref{theo:dual-dega}) for the above theorem in the case that $k=2$.
	

    We will also utilize the \textsf{Permutation Testing Problem}, or $\PTP$ for short. This problem, which is closely related to the \PCollision\ problem, was defined in \cite{aaronson2012impossibility}, which also (implicitly) proved a bound on its one-sided approximate degree.
    
    \begin{defi}[$\PTP$]
    	\label{defi:ptp}
    	Given a function $f: [n] \rightarrow [n]$ (represented as a string in $\{0,1\}^{n\log n}$),
    	\begin{enumerate}
    		\item $\PTP_n(f) = 1$ if $f$ is a permutation.
    		\item $\PTP_n(f) = 0$ if $f(i)$ differs from every permutation on at least $n/8$ values of $i$.
		\item $\PTP_n(f)$ is undefined otherwise.
    	\end{enumerate}
    \end{defi}
    
    \begin{theo}[Implicit in \cite{aaronson2012impossibility}]
    	\label{theo:ptp}
    	For any $0 < \epsilon < 1/6$,
    	\begin{align*}
    	\odegap_\epsilon(\overline{\PTP}_n) = \Omega(n^{1/3})
    	\end{align*}
    \end{theo}    
    
    The $\SZK$ protocol described for $\overline{\Collision}$ works unmodified for $\PTP$ as well.
    \label{sec:ptpprotocol}
 
\section{Hardness Amplification For Approximate Degree}
\label{sec:composition}

In this section we prove a novel hardness amplification theorem. Specifically,
we show that for any function $f$ with high approximate degree, composing $f$
with $\Gapmajority$ yields a function with high threshold degree, and hence the resulting function is hard
for any $\UPP$ algorithm in the query model. Similarly, we show
that if $f$ has high positive one-sided approximate degree, then composing $f$
with $\Gapand$ yields a function with high threshold degree.

 Note that this hardness amplification theorem is tight, in the sense that if $f$ has low approximate degree, then composing $f$ with $\Gapmajority$ yields a function that has low $\UPP$ query complexity, and the same holds for composing $f$ with $\Gapand$ if $f$ has low positive one-sided approximate degree. See Appendix \ref{sec:char} for details. 
 
\subsection{Notation}

For a partial function $f$, an integer $n$ and a real $\epsilon \in (1/2,1]$, we denote $\GapMaj{f}{n}{\epsilon}$ by $F$ for convenience, where $n$ and $\epsilon$ will always be clear in the context. We also use $x = (x_1,x_2,\dotsc,x_n)$ to denote an input to $F$, where $x_i$ represents the input to the $i$th copy of $f$.

The following simple lemma establishes some basic properties of dual witnesses exhibiting the fact that  $\dega_{\epsilon}(f) > d$  or $\odegap_{\epsilon}(f) > d$. 

\newcommand{\vstart}{\vspace{-1mm}}
\newcommand{\vmid}{\vspace{-8mm}}

\newcommand{\twoprops}[4]{
	\vstart
	\begin{flalign}\label{#2} 
	\qquad\bullet\qquad\text{#1}&&
	\end{flalign}
	\vmid
	\begin{flalign}\label{#4} 
	\qquad\bullet\qquad\text{#3}&&
	\end{flalign}
}

\newcommand{\twopropsname}[3]{
	\twoprops{#2}{#1prop1}{#3}{#1prop2}
}

\newcommand{\threeprops}[6]{
	\vstart
	\begin{flalign}\label{#2} 
	\qquad\bullet\qquad\text{#1}&&
	\end{flalign}
	\vmid
	\begin{flalign}\label{#4} 
	\qquad\bullet\qquad\text{#3}&&
	\end{flalign}
	\vmid
	\begin{flalign}\label{#6} 
	\qquad\bullet\qquad\text{#5}&&
	\end{flalign}
}

\newcommand{\threepropsname}[4]{
	\threeprops{#2}{#1prop1}{#3}{#1prop2}{#4}{#1prop3}
}

\newcommand{\fourprops}[8]{
	\vstart
	\begin{flalign}\label{#2} 
	\qquad\bullet\qquad\text{#1}&&
	\end{flalign}
	\vmid
	\begin{flalign}\label{#4} 
	\qquad\bullet\qquad\text{#3}&&
	\end{flalign}
	\vmid
	\begin{flalign}\label{#6} 
	\qquad\bullet\qquad\text{#5}&&
	\end{flalign}
	\vmid
	\begin{flalign}\label{#8} 
	\qquad\bullet\qquad\text{#7}&&
	\end{flalign}
}

\newcommand{\fourpropsname}[5]{
	\fourprops{#2}{#1prop1}{#3}{#1prop2}{#4}{#1prop3}{#5}{#1prop4}
}

\newcommand{\fivepropsname}[6]{
	\vstart
	\begin{flalign}\label{#1prop1} 
	\qquad\bullet\qquad\text{#2}&&
	\end{flalign}
	\vmid
	\begin{flalign}\label{#1prop2} 
	\qquad\bullet\qquad\text{#3}&&
	\end{flalign}
	\vmid
	\begin{flalign}\label{#1prop3} 
	\qquad\bullet\qquad\text{#4}&&
	\end{flalign}
	\vmid
	\begin{flalign}\label{#1prop4} 
	\qquad\bullet\qquad\text{#5}&&
	\end{flalign}
	\vmid
	\begin{flalign}\label{#1prop5} 
	\qquad\bullet\qquad\text{#6}&&
	\end{flalign}
}

\begin{lemma}\label{lm:mus} \label{lem:simple}
	Let $f \colon D \to \{0,1\}$ with $D \subseteq \{0,1\}^M$ be a partial function, $\epsilon$ be a real in $[0,1/2)$, and $d$ be an integer such that $\dega_{\epsilon}(f) > d$.
	
	Let $\mu \colon \{0,1\}^M \to \R$ be a dual witness to the fact $\dega_{\epsilon}(f) > d$ as per Theorem~\ref{theo:dual-dega}. If $f$ satisfies the stronger condition that $\odegap_{\epsilon}(f) > d$, let $\mu$ to be a dual witness to the fact that $\odegap_{\epsilon}(f) > d$ as per Theorem~\ref{theo:dual-odega}.
	
	We further define $\mu_+(x) := \max\{0,\mu(x)\}$ and $\mu_-(x) := -\min\{0,\mu(x)\}$ to be two non-negative real functions on $\{0,1\}^M$, and $\mu_-^i$ and $\mu_+^i$ be the restrictions of $\mu_-$ and $\mu_+$ on $f^{-1}(i)$ respectively for $i \in \{0,1\}$. Then the following holds:
		\threepropsname{lm41}{$\mu_+$ and $\mu_-$ have disjoint supports.}{$\langle \mu_+, p \rangle = \langle \mu_-, p \rangle$ for any polynomial $p$ of degree at most $d$. Hence, $\|\mu_+\|_1 = \|\mu_-\|_1 = \frac{1}{2}$.}{$\|\mu_+^1\|_1 > \epsilon$ and $\|\mu_-^0\|_1 > \epsilon$. 
			If $\odegap_\epsilon(f) > d$, then $\|\mu_+^1\|_1 = 1/2$.}
\end{lemma}

The lemma follows directly from Theorem~\ref{theo:dual-dega}. We provide a proof in Appendix~\ref{app:missing-lift} for completeness.

\subsection{Warm Up : A $\PP$ Lower Bound}

As a warmup, we establish a simpler hardness amplification theorem for $\PPdt$.

\newcommand{\udistr}{\mathcal{U}}

\begin{theo}\label{theo:lift}
	Let $f:D \to \{0,1\}$ with $D \subseteq \{0,1\}^M$ be a partial function, $n,d$ be two positive integers, and $1/2 <\epsilon<1$ and $0<\epsilon_2<1/2$ be two constants such that $2\epsilon_2>\epsilon$. Suppose $\dega_{\epsilon_2}(f) > d$. Then
	$$
	\PPdt(\GapMaj{f}{n}{\epsilon}) > \Omega\left\{\min\left(d,(2\epsilon_2 - \epsilon)^{2} \cdot n\right) \right\}.
	$$
\end{theo}
\begin{proof}
	For $i \in \{0, 1\}$ let $\mu_+,\mu_-,\mu_+^i,\mu_-^i$ be functions whose existence is guaranteed by Lemma~\ref{lm:mus}, combined with the assumption that $\dega_{\epsilon_2}(f) > d$.
	
	In light of Lemma~\ref{lm:dega-to-ppdt}, it suffices to {show that $\dega_{1/2 - 2^{-T}}(\GapMaj{f}{n}{\epsilon}) > T$, for\\
	 $T = \Omega\left\{\min \left(d,(2\epsilon_2 - \epsilon)^{2} \cdot n \right) \right\}$}. We prove this by constructing a dual witness to this fact, as per Theorem \ref{theo:dual-dega}.
	
	We first define the following two non-negative functions on $\{0,1\}^{n \cdot M}$:
	$$
	\psi^+(x) := \prod_{i=1}^{n} \mu_+(x_i) \quad \text{and} \quad \psi^-(x) := \prod_{i=1}^{n} \mu_-(x_i).
	$$
	
	Our dual witness $\psi$ is simply their linear combination:
	$$
	\psi := 2^{n-1} \cdot (\psi^+ - \psi^-).
	$$
	
	We remark that $\psi$ is precisely the function denoted by $\psi_{BT}$ alluded to in Section \ref{sec:proofoverview}.
	Now we verify that $\psi$ is the dual witness we want. 
	
	\medskip \noindent \textbf{Proving the $\psi$ has unit $\ell_1$-norm.} 
	Since $\mu_+$ and $\mu_-$ have disjoint supports by Condition~(\ref{lm41prop1}) of Lemma~\ref{lm:mus}, so does $\psi^+$ and $\psi^-$. Therefore $\|\psi\|_1 = 2^{n-1} \cdot (2^{-n} + 2^{-n}) = 1$.
	
	\medskip \noindent \textbf{Proving the $\psi$ has pure high degree $d$.} 
	Let $p \colon \{0,1\}^{n \cdot M} \to \R$ be any monomial of degree at most $d$, 
	and let $p_i \colon \{0, 1\}^M \rightarrow \R$ be such that $p(x_1, \dots, x_n) = \prod_{i=1}^n p_i(x_i)$. Then it holds that
	
	$$
	\langle \psi^+, p \rangle = \prod_{i=1}^{n} \langle \mu_+, p_i \rangle = \prod_{i=1}^{n} \langle \mu_-, p_i \rangle = \langle \psi^-, p \rangle,
	$$ 
	where  the second equality holds by Condition~(\ref{lm41prop2}) of Lemma \ref{lem:simple}.
	
	As a polynomial is a sum of monomials, by linearity, it follows that $\langle \psi, p \rangle = \langle \psi^+, p \rangle - \langle \psi^-, p \rangle = 0$ for any polynomial $p$ with degree at most $d$.
	
	\medskip \noindent \textbf{Proving that $\psi$ has high correlation with $F$.} 
    Define $\distr_0 := 2 \cdot \mu_-$ and $\distr_1 := 2 \cdot \mu_+$. Note $\mu_+$ and $\mu_-$  are non-negative functions with norm $1/2$, so $\distr_0$ and $\distr_1$ can be thought as distributions on $\{0,1\}^M$. We further define distributions $\udistr_i$ on $\{0,1\}^{n \cdot M}$ for $i \in \{0,1\}$ as $\udistr_i := \distr_i^{\otimes n}$. Observe that $\udistr_0 = 2^{n} \cdot \psi^-$ and $\udistr_1 = 2^{n} \cdot \psi^{+}$ as functions.
	
	\newcommand{\GapEps}{\Delta}
	\newcommand{\temp}{k}
	
	Then by Condition~(\ref{lm41prop3}) of Lemma \ref{lm:mus}, we have $\Pr_{x\sim \distr_1}[f(x) = 1] = 2 \cdot \|\mu_+^1\|_1 >2\epsilon_2>\epsilon$, and $\Pr_{x\sim \distr_0}[f(x) = 0] = 2 \cdot \|\mu_-^0\|_1 >2\epsilon_2>\epsilon$.
	
	Let $D_F$ denote the domain of $F$. By the definition of $F = \GapMaj{f}{n}{\epsilon}$ and a simple Chernoff bound, we have 
	\begin{equation} \label{eqlabel}
 	2^{n} \cdot \sum_{x \in D_F} \psi^+(x) \cdot F(x) = \Pr_{x \sim \udistr_1}[F(x) = 1] \ge 1 - 2^{-c_1 \GapEps^2 \cdot n },
	\end{equation}
	
	where $c_1$ is a universal constant and $\GapEps := 2\epsilon_2 - \epsilon$. For brevity, let $k$ denote $c_1 \GapEps^2 \cdot n$.
	
	Since $2^n \cdot \|\psi^+\|_1=1$, inequality \eqref{eqlabel} further implies that
	$$
	2^{n} \cdot \sum_{x \notin D_F} \psi^+(x) \le 2^{-k}.
	$$
	
	Similarly, we have
	$$
	\Pr_{x \sim \udistr_0}[F(x) = 0] \ge 1 - 2^{-\temp},
	$$
	
	which implies that
	$$
	2^{n} \cdot \sum_{x \notin D_F} \psi^-(x) \le 2^{-k}.
	$$
	
	Putting everything together, we can calculate the correlation between $F$ and $\psi$ as follows:
	
	\begin{align*}
	    &\sum_{x \in D_F} F(x) \psi(x) - \sum_{x \notin D_F} |\psi(x)|\\
	\ge& 2^{n-1} \cdot \sum_{x \in D_F} \psi^{+}(x) F(x) - 2^{n-1} \cdot \left( \sum_{x \notin D_F} \psi^-(x) + \sum_{x \notin D_F} \psi^+(x) \right)\\
	\ge& 1/2 - 2^{-k-1} - 2^{-k} \\
	>& 1/2 - 2^{-k+1}.
	\end{align*}
	
	Setting $T = \min(d,k-1)$, then we can see that $\psi$ is a dual witness for {$\dega_{1-2^{-T}}(\GapMaj{f}{n}{\epsilon}) > T$}. Clearly $T = \Omega\left\{\min\left(d,(2\epsilon_2 - \epsilon)^{2} \cdot n\right) \right\}$, which completes the proof.
\end{proof}

\subsection{The $\UPP$ Lower Bound}

\newcommand{\mugood}{\varphi_{\circ}}
\newcommand{\mubad}{\varphi_{\times}}

\newcommand{\corr}{\psi_{corr}}

\newcommand{\Acnt}{n_A}

The dual witness $\psi \sim \psi^{+} - \psi^{-}$ constructed in the previous subsection is not a dual witness for the high threshold degree of $F=\Gapmajority_n(f)$ for two reasons: it puts weight on some points outside of the domain of $F$, and it does not satisfy the sign-agreement condition of Theorem \ref{theo:dual-degsh}.

In order to obtain a valid dual witness for threshold degree, we add two error correction terms $\corr^+$ and $\corr^-$ to $\psi$. The purpose of the error correction terms is to zero out the erroneous values, while simultaneously maintaining the high pure degree property and avoiding changing the sign of $\psi$ on inputs at which it does not agree in sign with $F$. We achieve this through an error correction lemma that may be of independent interest.

\begin{lemma}[Error Correction Lemma]\label{lm:error-correction}
	Let $A$ be a subset of $\{0,1\}^M$, and $\varphi$ be a function on $\{0,1\}^M$. Let $\mugood$ and $\mubad$ be the restrictions of $\varphi$ on $A$ and $\{0,1\}^M \setminus A$ respectively. 
That is, $\mugood(x_i)=\varphi(x_i)$ if $x_i \in A$ and $\mugood(x_i)=0$ otherwise, and similarly 
$\mubad(x_i)=\varphi(x_i)$ if $x_i \not\in A$ and $\mubad(x_i)=0$ otherwise.
Define $\psi: \{0,1\}^{n \cdot M} \to \{0,1\}$ as $\psi(x_1,x_2,\dotsc,x_n) := \prod_{i=1}^n \varphi(x_i)$, and $\Acnt(x) := \sum_{i=1}^{n} \indicator_{x_i \in A}$.

	Suppose $\alpha = \|\mubad\|_1/\|\mugood\|_1 < 1/40$, and let $0.5< \epsilon < 1$ be a real number and $n$ be a sufficient large integer. Then there exists a function $\corr \colon \{0,1\}^{n \cdot M} \to \R$ such that:
	\threepropsname{lm43}{$\corr(x) = \psi(x)$, when $\Acnt(x) \le \epsilon \cdot n$.}{$|\corr(x)| \le \psi(x)/2$, when $\Acnt(x) > \epsilon \cdot n$.}{$\corr$ has pure high degree of at least $\left(1-\left(1 + 10 \alpha\right)  \cdot \epsilon \right) \cdot n - 4$.}
	
\end{lemma}

We defer the proof of Lemma \ref{lm:error-correction} to Subsection \ref{sec:errorcorrection}. Here, we show that it implies the desired hardness amplification results.


\begin{theo}\label{theo:lift-UPP} \label{mainfthm}
	Let $f:D \to \{0,1\}$ with $D \subseteq \{0,1\}^M$ be a partial function, $n$ be a sufficiently large integer, $d$ be an integer, and $1/2 <\epsilon<1$ and $ 0.49 < \epsilon_2 < 1/2$ be two constants.
Let $a = \frac{2\epsilon_2}{1-2\epsilon_2}$. Then the following holds. 
	\begin{align*} \text{If } \dega_{\epsilon_2}(f) > d, \text{ then }
	\degsh(\GapMaj{f}{n}{\epsilon}) &> \min\left(d,\left(1 - \left(1+\frac{10}{a} \right) \cdot \epsilon \right) \cdot n - 4\right).\\
	\text{If } \odegap_{\epsilon_2}(f) > d, \text{ then } \degsh(\Gapand_{n,\epsilon}(f)) &> \min\left(d,\left(1 - \left(1+\frac{10}{a} \right) \cdot \epsilon \right) \cdot n - 4\right).
	\end{align*}
\end{theo}

\begin{proof}
We prove both claims in the theorem by exhibiting a single dual solution that witnesses both.

As in the proof of Theorem \ref{theo:lift}, for $i \in \{0, 1\}$, let $\mu_+,\mu_-,\mu_+^i,\mu_-^i$ denote the functions whose existence is guaranteed by Lemma~\ref{lm:mus}, combined with the assumption
that either $\dega_\epsilon(f) > d$ or $\odegap_\epsilon(f) > d$. Also as in the proof of Theorem \ref{theo:lift}, define
the following two non-negative functions on $\{0,1\}^{n \cdot M}$:
$$
\psi^+(x) := \prod_{i=1}^{n} \mu_+(x_i) \quad \text{and} \quad \psi^-(x) := \prod_{i=1}^{n} \mu_-(x_i).
$$
	
	Given an input $x = (x_1,x_2,\dotsc,x_n)$, let $\nyes(x):= \sum_{i=1}^{n} \indicator_{f(x_i) = 1}$ and $\nno(x) := \sum_{i=1}^{n} \indicator_{f(x_i) = 0}$ as in Definition~\ref{defi:gap-majority}. Now apply Lemma~\ref{lm:error-correction} with the following parameters.
	
	\begin{itemize}
		\item Set $A = f^{-1}(1)$, $\varphi = \mu_{+}$. Then for $\alpha$ as defined in Lemma \ref{lm:error-correction}, we have $\alpha = \frac{\|\mu_+\|_1 -\|\mu_+^1\|_1}{\|\mu_+^1\|_1} \le \frac{1 - 2\epsilon_2}{2 \epsilon_2} = a^{-1}$ by Conditions~(\ref{lm41prop2}) and (\ref{lm41prop3}) of Lemma~\ref{lm:mus}. Note that $a^{-1}< 1/40$ by the assumption that $0.49 < \eps_2$. 
		Hence, by Lemma~\ref{lm:error-correction}, there exists a function $\corr^+ : \{0,1\}^{n \cdot M} \to \R$ such that:
		\threepropsname{psiplus}{$\corr^+(x) = \psi^+(x)$, for all $x$ such that $\nyes(x) \le \epsilon \cdot n$}{$|\corr^+(x)| \le \psi^+(x)/2$, for all $x$ such that $\nyes(x) > \epsilon \cdot n$}{$\corr^+$ has pure high degree at least $\left(1-\left(1 + \frac{10}{a}\right)  \cdot \epsilon \right) \cdot n - 4$}
		\item Similarly, set $A = f^{-1}(0)$, $\varphi = \mu_{-}$. Again by Lemma~\ref{lm:error-correction}, there exists a function $\corr^- : \{0,1\}^{n \cdot M} \to \R$ such that:
		\threepropsname{psiminus}
		{$\corr^-(x) = \psi^-(x)$, for all $x$ such that $\nno(x) \le \epsilon \cdot n$}{$|\corr^-(x)| \le \psi^-(x)/2$, for all $x$ such that $\nno(x) > \epsilon \cdot n$}{$\corr^-$ has pure high degree of at least $\left(1-\left(1 + \frac{10}{a}\right)  \cdot \epsilon \right) \cdot n - 4$}
	\end{itemize}
	
	For convenience, let $N = \left(1-\left(1 + \frac{10}{a}\right)  \cdot \epsilon \right) \cdot n - 4$. We are ready to construct the dual witness $\psi$ that establishes the claimed threshold degree lower bounds. Define $\psi \colon \{0,1\}^{n \cdot M} \to \R$ by
	
	$$
	\psi := (\psi^+ - \corr^+) - (\psi^- - \corr^-).
	$$
	
	We first establish two properties of $\psi$.
	\twopropsname{psi}{When $\nyes(x) \ge \epsilon \cdot n$, $\psi(x) = \psi^+(x) - \corr^{+}(x) \ge \psi^{+}(x)/2 \ge 0$}
	{When $\nno(x) \ge \epsilon \cdot n$, $\psi(x) = -(\psi^-(x) - \corr^{-}(x)) \le -\psi^{-}(x)/2 \le 0$}
	\medskip
	\noindent \textbf{Verifying Condition~(\ref{psiprop1}) and (\ref{psiprop2}).} To establish that Condition~(\ref{psiprop1}) holds, observe that since $\nyes(x) \ge \epsilon \cdot n$, and $\eps > 1/2$ by assumption, it follows that $\nno(x) \le (1-\epsilon) \cdot n \le \epsilon \cdot n$. This implies that $\psi^-(x) = \corr^-(x)$ by Condition~(\ref{psiminusprop1}) and $|\corr^+(x)| \le \psi^+(x)/2$ by Condition~(\ref{psiplusprop2}). Then $\psi(x) = \psi^+(x) - \corr^+(x) \ge \psi^{+}(x)/2 \ge 0$, where the last inequality follows from the fact that $\psi^{+}$ is non-negative.
	
	Similarly, for Condition~(\ref{psiprop2}), as $\nno(x) \ge \epsilon \cdot n$, it follows that $\nyes(x) \le (1-\epsilon) \cdot n \le \epsilon \cdot n$. This implies that $\psi^+(x) = \corr^+(x)$ by Condition~(\ref{psiplusprop1}) and $|\corr^{-}(x)| \le \psi^-(x)/2$ by Condition~(\ref{psiminusprop2}). Note $\psi^-$ is also non-negative. Hence $\psi(x) = -(\psi^-(x) - \corr^-(x)) \le -(\psi^-(x)/2) \le 0$.
	 
	 \medskip
	We now verify that $\psi$ is a dual witness for $\degsh(F) > \min\left(d, N\right)$ (recall that $F$ denotes $\Gapmajority(f)$). 
	
	\medskip
\noindent \textbf{Analyzing the pure high degree of $\psi$.} Write $ \psi := \psi^+ - \psi^- -\corr^+ + \corr^-$. We already established that $\psi^+ - \psi^-$ has pure high degree $d$ in the proof of Theorem \ref{theo:lift}, and both $\corr^+$ and $\corr^+$ have pure high degree at least $N$ (cf. Conditions~(\ref{psiplusprop3}) and (\ref{psiminusprop3})). By linearity, $\psi$ itself has pure high degree at least $\min\left(d, N\right)$.
	
	\medskip
\noindent \textbf{Showing that the support of $\psi$ is a subset of the inputs on which $F$ is defined.} 	Let $x$ be an input outside of the domain of $F$. Then by the definition of $\Gapmajority$, 
it must be the case that both $\nyes(x)$ and $\nno(x)$ are strictly less than $\epsilon \cdot n$. This means that $\psi^+(x) = \corr^+(x)$ and $\psi^-(x) = \corr^-(x)$  by Conditions~(\ref{psiplusprop1}) and (\ref{psiminusprop1}), and hence $\psi(x) = 0$. Therefore, the support of $\psi$ is a subset of the domain of $F$. 
	
		\medskip
\noindent \textbf{Showing that $\psi$ agrees in sign with $F$.} When $F(x) = 1$, by the definition of $\Gapmajority$, we have $\nyes(x) \ge \epsilon \cdot n$. Then $\psi(x) \ge 0$ follows directly from Condition~(\ref{psiprop1}). Similarly, when $F(x) = 0$, we have $\nno(x) \ge \epsilon \cdot n$ and $\psi(x) \le 0$ by Condition~(\ref{psiprop2}). Therefore, $\psi$ agrees in sign with $F$.

		\medskip
\noindent \textbf{Showing that $\psi$ is non-trivial.} Pick an input $x_0$ to $f$ such that $\mu_+^1(x_0) > 0$, and let $x=(x_0,x_0,\dotsc,x_0)$. Then we have $f(x_0) = 1$ and $\nyes(x) = n \ge \epsilon \cdot n$. 
Therefore, $\psi(x) = \psi^+(x) - \corr^+(x) \ge \psi^+(x)/2 = (\mu_+^1(x_0))^n/2 > 0$ by Condition~(\ref{psiprop1}). So $\psi$ is non-trivial.
	
\medskip Putting everything together and invoking Theorem~\ref{theo:dual-degsh} proves the first claim of Theorem~\ref{theo:lift-UPP}.
	
\medskip
\noindent \textbf{Showing $\psi$ is also a dual witness for $\Gapand_{n,\epsilon}(f)$.}	Now we show that, when $\odegap_{\epsilon_2}(f) > d$, the same function $\psi$ is also a dual witness for 
	$\degsh(\Gapand_{n,\epsilon}(f)) > \min\left(d, N\right)$. 
	
	We already proved that the pure high degree of $\psi$ is as claimed, and that it is non-trivial. So it remains to verify $\psi$ only puts weight in the domain of $\Gapand_{n,\epsilon}(f)$, and that $\psi$ agrees in sign with $\Gapand_{n,\epsilon}(f)$.
	
	By Condition~(\ref{lm41prop3}) of  Lemma~\ref{lm:mus}, we have $|\mu_+^1| = |\mu_+| = \frac{1}{2}$, which means $\mu_+$ only puts weight inputs in $f^{-1}(1)$. So $\psi^+$ only takes non-zero values when $\nyes(x) = n$. Also, note that when $\nno(x) \le \epsilon \cdot n$, we have $\psi^-(x) = \corr^-(x)$ by Condition~(\ref{psiminusprop1}). Therefore, $\psi$ only puts weight on inputs when $\nyes(x) = n$ or $\nno(x) > \epsilon \cdot n$. All such inputs are in the domain of $\Gapand_{n,\epsilon}(f)$. 
	
	Finally, we verify that $\psi$ agrees in sign with $\Gapand_{n,\epsilon}(f)$.
	When $\Gapand_{n,\epsilon}(f)(x) = 1$, we have $\nyes(x) = n \ge \epsilon \cdot n$, hence $\psi(x) \ge 0$ by Condition~(\ref{psiprop1}). When $\Gapand_{n,\epsilon}(f)(x) = 0$, we have $\nno(x) \ge \epsilon \cdot n$, so $\psi(x) \le 0$ follows immediately from Condition~(\ref{psiprop2}). Applying Theorem~\ref{theo:dual-degsh} again, this completes the proof for the second claim of Theorem~\ref{theo:lift-UPP}.

\end{proof}

\subsection{Proof of the Error Correction Lemma}
\label{sec:errorcorrection}
In this subsection we prove Lemma~\ref{lm:error-correction}. 
We need two lemmas. In the first, we construct a polynomial with certain properties.

\begin{lemma}\label{lm:helper2}
	Let $a \ge 40$, $n$ be a sufficiently large integer, and $\epsilon$ be a real such that $0.5 < \epsilon < 1$.
	Then there exists an (explicitly given) univariate polynomial $P \colon \R \to \R$ such that:
	
	\begin{itemize}
		\item $P(x) = (-a)^x$ for $x \in \{0,\dotsc,\epsilon \cdot n \}$.
		\item $|P(x)| \le a^x / 2$ for $x \in \{\epsilon \cdot n +1, \dotsc, n \}$.
		\item $P$ has degree of at most $\left(1 + \frac{10}{a}\right) \cdot \epsilon \cdot n + 3$.
	\end{itemize}
\end{lemma}

We prove Lemma \ref{lm:helper2} by defining $P$ via interpolation through carefully chosen values -- we defer this proof to Appendix~\ref{app:missing-lift}.

\begin{lemma}\label{lm:poly}
  Let $x=(x_1, \dots, x_n) \in \{0, 1\}^{M \cdot n}$, and let $\Acnt(x) := \sum_{i=1}^{n} \indicator_{x_i \in A}$. Let $\psi$ and $\alpha$ be as defined in the statement of Lemma \ref{lm:error-correction}. For any univariate polynomial $P \colon \R \to \R$ of degree at most $d$, the following function on $\{0,1\}^{n \cdot M}$ has pure high degree $n-d-1$:
	$$
	\psi_{P}(x) := \alpha^{\Acnt(x)} \cdot \psi(x) \cdot P(\Acnt(x)).
	$$
\end{lemma}

Before proving Lemma~\ref{lm:poly}, we show that it and Lemma \ref{lm:helper2} together imply Lemma~\ref{lm:error-correction}. 

\begin{proofof}{Lemma~\ref{lm:error-correction}}
	We first deal with the special case that $\alpha = 0$. In this case, we note $\psi(x) > 0$ only if $\Acnt(x) = n$. So it suffices to let $\corr$ be the constant function $\mathbf{0}$.
	
	From now on, we assume $\alpha > 0$. Let $a = \alpha^{-1}$ and $N =\left(1 + 10 \alpha \right) \cdot \epsilon \cdot n + 3 = \left(1 + \frac{10}{a}\right) \cdot \epsilon \cdot n + 3.$ 

\medskip
\noindent \textbf{Construction of $\corr$.}	
Applying Lemma~\ref{lm:helper2}, we obtain a polynomial $P(x)$ such that:
	\threepropsname{Px}
		{$P(x) = (-a)^x$ for $x \in \{0,\dotsc,\epsilon \cdot n \}$.}
		{$|P(x)| \le a^x / 2$ for $x \in \{\epsilon \cdot n +1, \dotsc, n \}$.}
		{$P$ has a degree of at most $\left(1 + \frac{10}{a}\right) \cdot \epsilon \cdot n + 3 = N$.}
	
	We define $\corr := \psi_{P}$,  where $\psi_P$ is as defined in Lemma \ref{lm:poly}.

\medskip
\noindent \textbf{Verification that $\corr$ satisfies the properties claimed in Lemma \ref{lm:error-correction}.} 
Given an input $x=(x_1,\dotsc,x_n)$, let $m=\Acnt(x)$ for brevity. We now verify Condition~(\ref{lm43prop1}), (\ref{lm43prop2}) and (\ref{lm43prop3}) of Lemma~\ref{lm:error-correction} in order.
When $m \le \epsilon \cdot n$, we have $P(m) = (-a)^{m}$ by Condition~(\ref{Pxprop1}), hence $$\corr(x) = (-\alpha)^m \cdot (-a)^m \cdot \psi(x) = \psi(x),$$ so Condition~(\ref{lm43prop1}) holds. When $m > \epsilon \cdot n$, as $|P(m)| \le a^{m}/2$ by Condition~(\ref{Pxprop2}), we have $|\corr(x)| \le \alpha^m \cdot a^{m}/2 \cdot \psi(x) = \psi(x)/2$, which establishes Condition~(\ref{lm43prop2}). Finally, since $P$ is of degree at most $N$ as guaranteed by Condition~(\ref{Pxprop3}), Lemma \ref{lm:poly} implies that $\corr$ has pure high degree of at least $n - N - 1 = \left(1 - \left(1 + \frac{10}{a} \right) \cdot \epsilon \right) \cdot n - 4$. So Condition~(\ref{lm43prop3}) is verified, and this completes the proof.
\end{proofof}

Finally, we prove Lemma~\ref{lm:poly}.
\begin{proofof}{Lemma~\ref{lm:poly}}
We begin by constructing some useful auxiliary functions. 

\medskip \noindent \textbf{Definition and analysis of auxiliary functions $\psi_k \colon \{0,1\}^{n \cdot M} \to \R$.}
 Recall the definitions of $\mugood$ and $\mubad$ from the statement of Lemma \ref{lm:error-correction}: $\mugood(x_i)=\varphi(x_i)$ if $x_i \in A$ and $\mugood(x_i)=0$ otherwise, and 
$\mubad(x_i)=\varphi(x_i)$ if $x_i \not\in A$ and $\mubad(x_i)=0$ otherwise.
For each integer $k \in \{0,\dotsc,n \}$, we define
	
	\begin{equation} \label{psik}
	\psi_k(x) := \sum_{S \subseteq [n], |S| = k} \left(\prod_{i \in S} \varphi(x_i) \cdot \prod_{i \notin S} (\mubad - \alpha \cdot \mugood)(x_i)\right).
	\end{equation}
	
	We claim that: 
	\begin{equation} \label{eqf0} \psi_k \text{ has pure high degree at least }n-k-1.\end{equation} 
	
	To establish this, it suffices to show that for every $|S|=k$, the following function
	\begin{equation}\label{psiS}
	\psi_{S}(x) := \prod_{i \in S} \varphi(x_i) \cdot \prod_{i \notin S} (\mubad - \alpha \cdot \mugood)(x_i)
	\end{equation}
	has pure high degree at least $n-k-1$, as $\psi_k$ is simply a sum of $\psi_{S}$'s with $|S| = k$.
	
	\newcommand{\istar}{i^{\star}}

	Let $[n] := \{1, \dots, n\}$, $p : \{0,1\}^{n \cdot M} \to \R$ be any monomial of degree at most $n-k-1$, and let $p_i : \{0,1\}^{M} \to \R$ be such that $p(x_1,\dotsc,x_n) = \prod_{i=1}^{n} p_i(x_i)$. Then $\sum_{i=1}^{n} \deg(p_i) = \deg(p) \le n-k-1$, which means there are at least $k+1$ $p_i$'s which have degree zero, i.e., are constant functions. Since $k+1 + |[n] \setminus S| = n + 1 > n$, there must exist an index $\istar$ such that $p_{\istar}$ is a constant function $p_{\istar}(x) \equiv c$, and $\istar \notin S$. Now, since $\mubad$ and $\mugood$ have disjoint supports and $\alpha = \|\mubad\|_1 / \|\mugood\|_1$ by definition, we have $\langle \mubad - \alpha \cdot \mugood, p_{\istar} \rangle = (\|\mubad\|_1 - \alpha \cdot \|\mugood\|_1) \cdot c = 0$. Therefore, by Equation~\eqref{psiS},
	
	$$
	\langle \psi_{S},p \rangle = 
	\langle \mubad - \alpha \cdot \mugood, p_{\istar} \rangle \cdot \prod_{i \in S} \langle \varphi, p_i \rangle \cdot 
	\prod_{i \notin S, i \ne \istar} \langle \mubad - \alpha \cdot \mugood, p_i \rangle = 0.
	$$
	
	As a polynomial is a sum of monomials, by linearity, we conclude that $\psi_S$ has pure high degree $n-k-1$ for all $|S| = k$, and so does $\psi_k$.
	
	\medskip
	Now, fix an input $x = (x_1,x_2,\dotsc,x_n)$. Write $m = \Acnt(x)$ (note that $m$ is actually a function of $x$). We are going to re-express $\psi_k(x)$ in a convenient form, as follows.
	We assume that the first $m$ inputs $x_1,x_2,\dotsc,x_m$ satisfy $x_i \in A$ -- this is without loss of generality by symmetry.
	
	Fix a set $S \subseteq [n]$ with $|S|=k$. 
	We claim that 
	$$\psi_S(x) = (-\alpha)^{m-|S \cap [m]|} \cdot \prod_{i=1}^n \varphi(x_i).$$ To see this, 
	first consider $i \in S \cap [m]$. Then there 
	is a factor $\varphi(x_i)$ appearing in $\psi_S(x)$. 
	Now fix an $i \in \left([m] \setminus S\right)$. Then there is a factor of 
	$(\mubad - \alpha \cdot \mugood)(x_i) = (-\alpha) \cdot \varphi(x_i)$ appearing in $\psi_S(x)$.
	Finally, fix any $i \not \in [m]$. If $i \in S$, then there is a factor $\varphi(x_i)$ appearing in $\psi_S(x)$, and if $i \not \in S$, the factor appearing in 
	$\psi_S(x)$ is $(\mubad(x_i) - \alpha \cdot \mugood)(x_i) = \mubad(x_i) = \varphi(x_i)$.
	
	Hence, we may write
	\begin{align}
	\psi_k(x) & \notag = \sum_{S \subseteq [n], |S| = k}  (-\alpha)^{|S \cap [m]|} \cdot \prod_{i =1}^n \varphi(x_i)\\
	&\notag =
	\prod_{i=1}^{n} \varphi(x_i) \cdot \left(\sum_{j=0}^{k} \binom{m}{j} \cdot \binom{n-m}{k-j} \cdot (-\alpha)^{m-j} \right)\\ & =
	 (-\alpha)^{m} \cdot \psi(x) \cdot \left(\sum_{j=0}^{k} \binom{m}{j} \cdot \binom{n-m}{k-j} \cdot (-\alpha)^{-j}\right).\label{lastline}\end{align}

\noindent \textbf{Definition and analysis of auxiliary univariate polynomials $P_k \colon \R \to \R$.}
	Let $$P_k(m) := \sum_{i=0}^{k} \binom{m}{i} \cdot \binom{n-m}{k-i} \cdot (-\alpha)^{-i}.$$ Then $P_k(m)$ is polynomial in $m$ of degree at most $k$, and by Equation \eqref{lastline}, 
	\begin{equation} \label{eqff} \psi_{k}(x) = (-\alpha)^m \cdot \psi(x) \cdot P_k(m). \end{equation}
	
	Expanding the binomial coefficients, we have
	\begin{equation}\label{anotherform}
	P_k(m) = \sum_{i=0}^{k} \frac{\prod_{j=0}^{i-1} (m-j)}{i!} \cdot \frac{\prod_{j=0}^{k-i-1} (n-m-j)}{(k-i)!} \cdot (-\alpha)^{-i}.
	\end{equation}
	
	Observe that the coefficient of $m^{k}$ in Equation \eqref{anotherform} is
	\begin{align*}
	&\sum_{i=0}^{k}  \frac{1}{i! \cdot (k-i)!} \cdot (-\alpha)^{-i} \cdot (-1)^{k-i}\\
	=&(-1)^k \cdot \frac{1}{k!} \cdot \sum_{i=0}^k \alpha^{-i} \cdot \binom{k}{i}\\
	=&(-1)^k \cdot \frac{1}{k!} \cdot \left( 1 + \alpha^{-1} \right)^k \ne 0, 
	\end{align*}
	
	where the second equality follows from the equation $(a+b)^k =\sum_{i=0}^{k} a^i b^{k-i} \cdot \binom{k}{i}$, and the last inequality follows because $\alpha \geq 0$ by definition.
	
	So $P_k(m)$ is a polynomial of degree exactly $k$. Therefore, the set $\{P_k(m)\}_{k=0}^{d}$ generates all the polynomials in $m$ with degree at most $d$.

\medskip
\noindent \textbf{Verification that $\psi_P$ has the pure high degree claimed in Lemma \ref{lm:poly}.}
 Let $P \colon \R \to \R$ be a polynomial of degree at most $d$. By the previous paragraph, we  may write $P(m)$ as 
 \begin{equation} \label{eqf2} P(m) = \sum_{k=0}^{d} \beta_k \cdot P_k(m),\end{equation}
  for some real numbers $\beta_0,\dotsc,\beta_d$.
	
	Then we have
	\begin{align*}
	\psi_P(x) =& \alpha^{m(x)} \cdot \psi(x) \cdot P(m(x))\\ 
	=& \alpha^{m(x)} \cdot \psi(x) \cdot \sum_{k=0}^{d} \left(\beta_k \cdot P_k(m(x)) \right)\\
	=& \sum_{k=0}^{d} \beta_k \cdot \psi_k(x).\\
	\end{align*}
	Here, the first equality holds by definition of $\psi_P$, the second by Equation \eqref{eqf2}, and the third by Equation \eqref{eqff}.
	
	Each $\psi_k$ appearing in the above sum has pure high degree at least $n-d-1$ by Property \eqref{eqf0}. Hence, by linearity, $\psi_P$ has pure high degree of $n-d-1$. This completes the proof.
\end{proofof}
\section{$\NISZK^\Oracle \not\subset \UPP^{\Oracle}$}
\label{sec:niszk-pp}

	\newcommand{\GCol}{\mathsf{GCol}}
	In this section we construct an oracle $\Oracle$ such that $\NISZK^\Oracle \not\subset \UPP^{\Oracle}$. We will use the function $\GCol_n := \GapMaj{\Collision_{n^{3/4}}^{3\log n}}{n^{1/4}}{1-\frac{1}{3\log n}}$ to attain the desired oracle separation.
	
	We first show that its complement $\overline{\GCol_n}$ is easy for $\NISZK$ by providing a reduction from it to the statistical distance from uniform (SDU) problem. SDU is complete for $\NISZK$ and so has an $\NISZK$ protocol~\cite{goldreich1999can}. We first introduce the problem \SDU.
	
	\begin{defi}[\textsf{Statistical Distance from Uniform} (\SDU)~\cite{goldreich1999can}]\label{defi:SDU}
		The promise problem \textsf{Statistical Distance from Uniform}, denoted $\SDU = (\SDU_{\mathsf{YES}},\SDU_{\mathsf{NO}})$, consisted of
		
		\begin{align*}
		\SDU_{\mathsf{YES}} &= \{X : \|X-U\| < 1/n \}\\
		\SDU_{\mathsf{NO}} &= \{X : \|X-U\| > 1 - 1/n\}
		\end{align*}
		
		where $X$ is a distribution encoded as a circuit outputting $n$ bits, and $U$ is the uniform distribution on $n$ bits, and $\|X-U\|$ denotes the statistical distance between $X$ and $U$.
	\end{defi}
	
	\begin{theo}\label{theo:GCOL-EASY}
		There is a $\polylog(n)$-time $\NISZK$ protocol for $\overline{\GCol_n}$.
	\end{theo}
	\begin{proof}
		
		For simplicity, we assume $n$ is a power of $2$. We prove this theorem by showing a reduction from $\overline{\GCol_n}$ to an instance of $\SDU$ with distributions on $\log n$ bits.
		
		Now, let $m=n^{1/4}$, $k = n^{3/4}$ and $x=(f_1,f_2,\dotsc,f_m)$ be an input to $\GCol_n$, where each $f_i$ is interpreted as a function from $[k] \to [k]$. We construct the distribution $\distr(x)$ as follows: to generate a sample from $\distr(x)$, we pick a pair $(i,j) \in [m] \times [k]$ at uniformly random, and output the sample $(i,f_i(j))$. Clearly $\distr(x)$ is $\polylog(n)$-time preparable.
		
		Now we show this is a valid reduction. Let $\udistr$ be the uniform distribution on $[m] \times [k]$ and $\udistr_k$ be the uniform distribution on $[k]$. For a function $f : [k] \to [k]$, let $\distr_f$ be the distribution obtained by outputting $f(i)$ for an index $i \sim \udistr_k$. Then we can see $\distr(x) = \frac{1}{m} \sum_{i=1}^{m} \{i\} \times \distr_{f_i}$.
		
		When $\overline{\GCol_n}(x) = 1$, we have 
		$$
		\|\distr(x) - \udistr\| = \frac{1}{m} \sum_{i=1}^{m} \| \udistr_k - \distr_{f_i} \| \le \frac{1}{3\log n} < \frac{1}{\log n}.
		$$
		
		Here, the first inequality holds because at least a $1 - \frac{1}{3\log n}$ fraction of $f_i$'s are permutations, which implies that $\|\udistr_k - \distr_{f_i}\| = 0$.
		
		When $\overline{\GCol_n}(x) = 0$, we have
		$$
		\|\distr(x) - \udistr\| = \frac{1}{m} \sum_{i=1}^{m} \| \udistr_k - \distr_{f_i} \| \ge \left(1-\frac{1}{3\log n}\right) \cdot \left(1 - \frac{1}{3\log n}\right) > 1 - \frac{1}{\log n}.
		$$
		
		Here, the first inequality holds because at least a $1 - \frac{1}{3\log n}$ fraction of $f_i$'s are $3\log n$-to-$1$, which implies that $\|\udistr_k - \distr_{f_i}\| = 1 - \frac{1}{3\log n}$.

		Putting everything together, we have shown $\distr(x)$ that is a valid reduction to $\SDU$. This completes the proof.
		
	\end{proof}
	
	Then by a straightforward application of Theorem~\ref{theo:lift}, we can show $\GCol_n$ is hard for any $\UPP$ algorithm.
	
	\begin{theo}
		\label{theo:sep}
		$\UPPdt(\GCol_n) = \Omega(n^{1/4}/\log n)$.
	\end{theo}
	\begin{proof}
		Observe that $\dega_{1/2 - \frac{1}{50\log n}}(\Collision_{n^{3/4}}^{3\log n}) = \Omega(n^{1/4} / \log^{2/3} n)$ by Theorem~\ref{theo:kutin-lowb}. Applying Theorem~\ref{theo:lift-UPP} with $a = \frac{1 - 2 \cdot \frac{1}{50\log n}}{2 \cdot \frac{1}{50\log n}} = 25 \log n -1$ (recall $a = \frac{2\epsilon_2}{1-2\epsilon_2}$ in Theorem~\ref{theo:lift-UPP}), we have that
		\begin{align*}
		&\UPPdt\left(\GapMaj{\Collision_{n^{3/4}}^{3\log n}}{n^{1/4}}{1-\frac{1}{3\log n}}\right) \\
		\ge&\degsh\left(\GapMaj{\Collision_{n^{3/4}}^{3\log n}}{n^{1/4}}{1-\frac{1}{3\log n}}\right) \tag{by Lemma~\ref{lm:tdeg-to-uppdt}} \\
		\ge&\min\left\{ \left(1 - \left(1+\frac{10}{a}\right)\cdot \left(1-\frac{1}{3\log n}\right) \right) \cdot n^{1/4} - 4, \dega_{1/2 - \frac{1}{50\log n}}(\Collision_{n^{3/4}}^{3\log n}) \right\} \\
		\ge&\Omega(n^{1/4} /\log n). 
		\end{align*}
	\end{proof}
	
	
	Now Theorem \ref{mainfcor} from the introduction follows from standard diagonalization methods and the observation that $\UPP$ is closed under complement.
	



\section{Limitations on Perfect Zero Knowledge Proofs (Proof of Theorem \ref{thm:pzklimit})}
\label{sec:oracleseps}

In this section, we study the limitations of perfect zero knowledge in the presence of oracles.

\begin{defi}[Honest Verifier Perfect Zero Knowledge]
    \label{def:pzkclass}
    A promise problem $L = (L_Y, L_N)$ is in $\HVPZK$ if there exist a tuple of Turing machines $(P, V, S)$, where the \emph{verifier} $V$ and \emph{simulator} $S$ run in probabilistic polynomial time, satisfying the following:
    \begin{itemize}
      \item $(P,V)$ is an interactive proof for $L$ with negligible completeness and soundness errors.
      \item $S(x)$ is allowed to fail (by outputting $\bot$), but with probability at most 1/2.
      \item For any $x\in L$, let $\hat{S}(x)$ denote the distribution of $S(x)$ conditioned on it not failing. Then,
        \begin{align*}
          \|\hat{S}(x) - (P,V)(x)\| = 0
        \end{align*}
    \end{itemize}
\end{defi}

The class $\PZK$ is defined similarly but, as in the case of $\SZK$, with the additional stipulation that for any verifier that deviates from the protocol, there is a simulator that simulates the prover's interaction with that verifier. It is easy to see that $\PZK \subseteq \HVPZK$ in the presence of any oracle. (The definitions of these classes in the presence of an oracle are the same, except that $P$, $V$, and $S$ all have access to the oracle.)

Note that the probability of failure of the simulator can be made negligible (in $|x|$) by repeating it a polynomial number of times and taking its output to be that from the first time that it succeeds. We use this implicitly in the rest of our development. The variant where the simulator is not allowed to fail is called Super-Perfect Zero Knowledge by Goldreich and Teichner~\cite{GT14}. There (and also elsewhere), this definition is considered to be ``oversimplified'' as such proof systems are not known for problems outside $\BPP$. However, in the setting of honest verifiers with small but non-zero completeness error, the class thus defined turns out to be equal to $\HVPZK$ \cite{GT14}. While sometimes these classes are defined with the requirement of perfect completeness in the zero knowledge proofs, note that defining them as above only makes our results stronger -- requiring perfect completeness can only make $\HVPZK$ smaller, and our oracle separation between $\PZK$ and $\coPZK$ continues to hold when both these classes are defined with perfect completeness.

\subsection{A Preliminary Lemma}

We will need the following lemma in the proof of the theorems that follow.

\begin{lemma}[]
    \label{lem:pp-l2}
    There is an oracle Turing Machine $M_2$ that is such that when given sample access to two distributions $p$ and $q$, $M_2$ uses two samples and,
    \begin{align*}
      \Pr[M_2^{p,q}\ \text{accepts}] = \frac{1}{2} + \frac{\ltwo{p-q}^2}{8}
    \end{align*}
\end{lemma}

\begin{proof}
    $M_2^{p,q}$ behaves as follows:
    \begin{enumerate}
        \setlength\itemsep{0em}\vspace{-2mm}
      \item With probability $\frac{1}{4}$, sample $y_1, y_2$ from $p$.
        \begin{itemize}
          \item If $y_1 = y_2$, accept with probability 1.
          \item Else, accept with probability $\frac{1}{2}$.
        \end{itemize}
      \item With probability $\frac{1}{4}$, do the same with samples from $q$.
      \item With probability $\frac{1}{2}$, sample $y_1$ from $p$ and $y_2$ from $q$.
        \begin{itemize}
          \item If $y_1 = y_2$, reject with probability 1.
          \item Else, accept with probability $\frac{1}{2}$.
        \end{itemize}
    \end{enumerate}
    \begin{align*}
      \Pr[M_2^{p,q}\ \text{accepts}] &= \frac{1}{4}\left[ (1-\ltwo{p}^2)\frac{1}{2} + \ltwo{p}^2 \right] + \frac{1}{4}\left[ (1-\ltwo{q}^2)\frac{1}{2} + \ltwo{q}^2 \right] + \frac{1}{2}\left[ (1-\ip{p}{q})\frac{1}{2} + \ip{p}{q}\cdot 0 \right]\\
      &= \frac{1}{2} + \frac{\ltwo{p-q}^2}{8}
    \end{align*}
\end{proof}

\subsection{Showing $\HVPZK \subseteq \PP$ Relative to Any Oracle}
The first step 
in our proof of Theorem \ref{thm:pzklimit} is to show that $\HVPZK$ is contained in $\PP$ in a relativizing manner. 

\begin{theo}[]
    \label{lem:pzk-pp}
      $\HVPZK \subseteq \PP$.
    Further, this is true in the presence of any oracle.
\end{theo}

\begin{proof}
    Let $L$ be a language with an $\HVPZK$ proof system $(P,V,S)$. We will show how to decide membership in $L$ in $\PP$. Fix any input length $n$, and let the number of messages in the proof system for any input of this length be $m$, and the length of each message be $\ell$ (these are without loss of generality). Also suppose that the first message is always sent by the verifier. Let the number of random bits used by $V$ on an input of length $n$ be $v$, and the number of random bits used by $S$ be $s$.

    For any $x\in \bset^n$, we write the output of the simulator $S$ on input $x$ using randomness $r$ as $S(x;r) = (R_V(x;r), T_1(x;r), \dots, T_m(x;r))$, where $R_V$ is the simulated randomness of the verifier, and $T_i$ is the simulated $i$th message in the protocol. Let $S_i$ denote $S$ truncated at $T_i$. Denote by $V_S$ the verifier simulated by $S$, and by $P_S$ the simulated prover, both conditioned on the simulator not failing.

    \begin{claim}
        \label{claim:valid-sim}
        An input $x$ is in $L$ if and only if the following three conditions are satisfied:
        \begin{enumerate}
          \item $V_S$ on input $x$ behaves like the actual verifier $V$. This involves the following:
            \begin{itemize}
              \item $R_V(x)$, conditioned on not being $\bot$, is distributed uniformly over $\bset^v$.
              \item For any non-failing transcript $(r_V, t_1, \dots, t_m)$ output by $S(x)$, the verifier's responses in $(t_1, \dots, t_m)$ are consistent with what $V$ would have sent when using $r_V$ as randomness.
            \end{itemize}
          \item $P_S$ on input $x$ is a valid prover.
            \begin{itemize}
              \item This means that the distribution of the prover's simulated messages $(T_{2i}(x))$ should depend only on the messages in the transcript so far $(T_1(x), \dots, T_{2i-1}(x))$, and should be independent of the verifier's simulated randomness $(R_V(x))$.
            \end{itemize}
          \item $S(x)$ is an accepting transcript with probability at least ${3}/{4}$.
        \end{enumerate}
    \end{claim}

    For any $x\in L$, the transcript of the actual protocol satisfies the above properties, and so does the simulation, since it is perfect conditioned on not failing.
    
    The other direction follows on noting that if all three conditions are satisfied for some $x$, then $P_S$ is a prover strategy that convinces the actual verifier $V$ that $x\in L$. By the soundness of the $(P,V)$ proof system, this can only happen if $x$ is indeed in $L$.

    So to decide the membership of $x$ in $L$ in $\PP$, it is sufficient to be able to decide each of the above three properties of $S(x)$ in $\PP$ (since $\PP$ is closed under conjunction \cite{BRS91}). Of these, property (3) is easily seen to be decidable in $\BPP$, and hence in $\PP$.

    Lemma~\ref{lem:pp-l2} says, in particular, that testing whether two polynomial-time-samplable distributions $p$ and $q$ are identical can be done in $\PP$. Let $U_S(x)$ be the distribution sampled by first running $S(x)$, outputting $\bot$ if it fails, and a uniform sample from $\bset^v$ if it doesn't. The first check on $V_S$ is the same as checking whether $R_V(x)$ is identical to $U_S(x)$. The other check required on $V_S$ is a $\co\NP$ statement, and hence can be done in $\PP$.

    Let $M_2$ be the TM from Lemma~\ref{lem:pp-l2}. To check that $P_S$ is a valid prover, consider the TM -- call it $M_P$ -- that works as follows on input $x$.
    \begin{enumerate}
      \item Select $i \in \left[-1,\frac{m}{2}\right]$ at random.
      \item If $i = 0$, run $M_2$ on the distributions $R_V(x)$ and $U_S(x)$.
      \item If $i = -1$, check the consistency of transcripts produced by $S(x)$ with the simulated randomness.
        \begin{itemize}
          \item This is done by selecting $r_S \in \bset^s$, and running $S(x;r_S)$ to get $(r_V, t_1, \dots, t_m)$.
          \item If this transcript is failing or consistent, accept with probability $1/2$, else with probability $1$.
        \end{itemize}
      \item Else, select at random $t_1, \dots, t_{2i} \in \bset^\ell$, $r_V^1, r_V^2 \in \bset^v$, and $r_S^1, r_S^2 \in \bset^s$.
      \item If $S(x;r_S^1)$ does not have $(r_V^1, t_1, \dots, t_{2i-1})$ as a prefix or $S(x;r_S^2)$ does not have $(r_V^2, t_1, \dots, t_{2i-1})$ as a prefix, accept with probability ${1}/{2}$.
      \item Let $p$ be the distribution over $\bset$ such that $p(1) = \Pr[S_{2i}(x) = (r_V^1, t_1, \dots, t_{2i})]$, and $q$ be the same but with $r_V^2$ instead of $r_V^1$.
      \item Run $M_2$ on the distributions $p$ and $q$.
    \end{enumerate}

    \begin{claim}
        \label{claim:in-pp}
        $M_P(x)$ accepts with probability at most $\frac{1}{2}$ if and only if $V_S$ is a valid verifier and $P_S$ is a valid prover on input $x$.
    \end{claim}

    Suppose $V_S$ is a valid verifier and $P_S$ is a valid prover on input $x$. If $M_P$ selects $i=0$ or $i=-1$, then it accepts with probability $\frac{1}{2}$ because $V_S$ is a valid verifier. 

    If $i\notin \{-1,0\}$, and $M_P$ picks $r_V^1, r_V^2, t_1, \dots, t_{2i}$. If this fails the check in step 5, then $M_P$ again accepts with probability $1/2$. If this does not happen and $r_V^1, r_V^2, t_1, \dots, t_{2i-1}$ are in the support of $S_{2i-1}(x)$,
    \begin{align*}
      \Pr[S_{2i}(x) = (r_V^1, &t_1, \dots, t_{2i})]\\ &= \Pr[S_{2i-1}(x) = (r_V^1, t_1, \dots, t_{2i-1})] \Pr[T_{2i}(x) = t_{2i}\ |\ S_{2i-1}(x) = (r_V^1, t_1, \dots, t_{2i-1})] \\
                                                   &= \Pr[S_{2i-1}(x) = (r_V^1, t_1, \dots, t_{2i-1})] \Pr[T_{2i}(x) = t_{2i}\ |\ S_{2i-1}(x) = (r_V^2, t_1, \dots, t_{2i-1})]
    \end{align*}
    where the second equality is because $P_S$ is a valid prover, so its responses do not depend on the simulated randomness of the verifier. We can write the first term in the product above as:
    \begin{align*}
      \Pr[S_{2i-1}(x) = (&r_V^1, t_1, \dots, t_{2i-1})] \\ &= \Pr[S_{2i-2}(x) = (r_V^1, t_1, \dots, t_{2i-2})]\Pr[T_{2i-1}(x) = t_{2i-1}\ |\ S_{2i-2}(x) = (r_V^1, t_1, \dots, t_{2i-2})]\\
                                                       &=  \Pr[S_{2i-2}(x) = (r_V^1, t_1, \dots, t_{2i-2})]
    \end{align*}
    where the second equality is because $V_S$ is a valid verifier and is deterministic once $R_V$ is fixed, and step 5 was there precisely to check that this probability is non-zero.

    Now starting from the fact that $\Pr[S_0(x) = r_V^1] = \Pr[S_0(x) = r_V^2]$, and using the above relationships, we can inductively prove that $\Pr[S_{2i}(x) = (r_V^1, t_1, \dots, t_{2i})] = \Pr[S_{2i}(x) = (r_V^2, t_1, \dots, t_{2i})]$. This implies that the call to $M_2$ in step 7 of $M_P$ accepts with probability $1/2$, as the distributions $p$ and $q$ there are identical. So in all cases, $M_P$ accepts with probability $1/2$.

    To prove the converse, we start by noting that each branch of $M_P$ always accepts with probability $1/2$ or more. So even if one of the branches accepts with probability strictly more than $1/2$, the acceptance probability of $M_P$ as a whole will be strictly more than $1/2$.

    Now suppose $V_S$ is not a valid verifier. Then $M_P$ would accept with probability strictly more than $1/2$ because either $i=0$ or $i=-1$ would accept with probability more than $1/2$.

    The remaining case is where $V_S$ is a valid verifier but $P_S$ is not a valid prover. This means that at some point the distribution of $P_S$'s responses depended on the simulated verifier's randomness. Specifically, there must exist an $i\in [m/2]$ and $r_V^1, r_V^2, t_1, \dots, t_{2i}$ such that $(\{r_V^1, r_V^2\}, t_1, \dots, t_{2i-1})$ are in the support of $S_{2i-1}(x)$ and:
    \begin{align*}
      \Pr[T_{2i}(x) = t_{2i}\ |\ S_{2i-1}(x) = (r_V^1, t_1, \dots, t_{2i-1})] \neq \Pr[T_{2i}(x) = t_{2i}\ |\ S_{2i-1}(x) = (r_V^2, t_1, \dots, t_{2i-1})]
    \end{align*}
    For this $r_V^1$ and $r_V^2$, let $i_0$ be the least $i$ such that there exist $t_1, \dots, t_{2i_0}$ where such an inequality holds. $i_0$ being the smallest such $i$ implies, by the same induction arguments above and the validity of $V_S$ as a verifier, that:
    \begin{align*}
      \Pr[S_{2i_0-1}(x) = (r_V^1, t_1, \dots, t_{2i_0-1})] = \Pr[S_{2i_0-1}(x) = (r_V^2, t_1, \dots, t_{2i_0-1})]     
    \end{align*}
    Putting the above two relations together, we get:
    \begin{align*}
      \Pr[S_{2i_0}(x) = (r_V^1, t_1, \dots, t_{2i_0})] \neq \Pr[S_{2i_0}(x) = (r_V^2, t_1, \dots, t_{2i_0})]     
    \end{align*}
    So when $M_P$ chooses $i = i_0$ and these values of $r_V^1, r_V^2$ and $t_1, \dots, t_{2i_0}$, it will accept with probability strictly greater than $1/2$, and so it will do so overall as well. This proves Claim~\ref{claim:in-pp}.

    Due to the fact that $\PP$ is closed under complement and Claim~\ref{claim:in-pp}, we have now established that the conditions in Claim~\ref{claim:valid-sim} can be checked in $\PP$. And so by Claim~\ref{claim:valid-sim}, $L$ can be decided in $\PP$. It is also easy to see that this proof still works relative to any oracle, as it only makes black-box use of $S$. 
\end{proof}

The following theorem follows immediately from Corollary~\ref{cor:niszk-pp} and Lemma~\ref{lem:pzk-pp}.

\begin{theo}[]
    \label{thm:niszk-pzk}
    There is an oracle $\Oracle$ such that $\NISZK^\Oracle \not\subseteq \HVPZK^\Oracle$. Consequently, $\SZK^\Oracle \not\subseteq \PZK^\Oracle$ and $\NISZK^\Oracle \not\subseteq \NIPZK^\Oracle$.
\end{theo}


\subsection{A Relativized Separation of $\PZK$ and $\coPZK$}
\begin{theo}[]
    \label{thm:pzk-copzk}
    There is an oracle $\Oracle$ such that $\PZK^\Oracle \neq \coPZK^\Oracle$.
\end{theo}
\begin{proof}
In order to prove Theorem~\ref{thm:pzk-copzk}, we first show that $\HVPZK$ is closed under ``composition'' with $\mathsf{GapAND}$.

\begin{lemma}
    \label{lem:pzk-closure}
    Let $f:D \to \{0,1\}$ with $D \in \{0,1\}^M$ be a partial function and $n$ be a positive integer, $1/2 < \epsilon < 1$ be a constant. If $f$ has a $\polylog(M)$-time $\HVPZK$ protocol, then $\GapAND{f}{n}{\epsilon}$ has a $\polylog(nM)$-time $\HVPZK$ protocol.
\end{lemma}

\begin{proof}
    For convenience, denote $\GapAND{f}{n}{\epsilon}$ by $g$. Given an $\HVPZK$ protocol $(P,V,S)$ for $f$, we will construct an $\HVPZK$ protocol $(P',V',S')$ for $g$. Given an input $x = (x_1, \dots, x_n)$ for $g$, $V'$ selects, say, $\log^2(n)$ values of $i \in [n]$, and $P'$ and $V'$ run the interactive protocol $(P,V)$ on each of the corresponding $x_i$'s independently. $V'$ accepts if and only if $(P,V)$ accepts on all these $x_i$'s. Completeness and soundness follows easily from standard arguments and the definition of $g$.

    On a similar input, the simulator $S'$ simply selects the same number of $i$'s, and runs $S$ on the corresponding $x_i$'s. Since in a YES instance all of the $x_i$'s are such that $f(x_i) = 1$, $S$ simulates the transcripts of $(P,V)$ on all of these exactly, except with a negligible probability when it fails on one or more of these. Hence $S'$ simulates $(P',V')$ exactly as well, again failing only with a negligible probability.
\end{proof}

We will need the following implication of the constructions in~\cite{DGOW95}.

\begin{lemma}[Implied by \cite{DGOW95}]
    \label{lem:nipzk-pzk}
    Any partial Boolean function that has a $\polylog(n)$-time $\NIPZK$ protocol also has a $\polylog(n)$-time $\PZK$ protocol.
\end{lemma}

We will use the function $\PTP_n$ (cf. Definition~\ref{defi:ptp}) to establish our separation. The following are immediate consequences of Theorems~\ref{theo:ptp} and \ref{theo:lift-UPP} and Lemma~\ref{lm:dega-to-ppdt}.
\begin{cor}
    \label{cor:ptp}
      $\PPdt(\GapAND{\overline{\PTP}_n}{n}{7/8}) = \Omega(n^{1/3})$
\end{cor}

\begin{lemma}
    \label{lem:ptp-pzk}
    $\PTP_n$ has a $\polylog(n)$-time $\PZK$ protocol.
\end{lemma}

\begin{proof}
    We will show this by presenting a $\polylog(n)$-time $\NIPZK$ protocol for $\PTP_n$ and invoking Lemma~\ref{lem:nipzk-pzk}. The protocol is very similar to the one described in Section \ref{sec:ptpprotocol}. Given a function $f:[n]\rightarrow [n]$ as input, an $r\in [n]$ is chosen at random using the common random string. $P$ is then supposed to send an $x$ to $V$ such that $f(x) = r$. $V$ accepts if this is true. Completeness, soundness and perfect zero-knowledge are all easily argued using the definition of $\PTP_n$.
\end{proof}

Now we have everything we need to prove Theorem~\ref{thm:pzk-copzk}. Suppose $\PZK^\Oracle = \coPZK^\Oracle$ with respect to all oracles $\Oracle$. This implies that any language that is in $\polylog(n)$-time $\PZK$ is also in $\polylog(n)$-time $\coPZK$, and vice versa -- if this were not true for some language, then we would be able to use that language to construct an oracle that separates the two classes by diagonalization. In particular, this hypothesis and Lemma~\ref{lem:ptp-pzk} imply that $\overline{\PTP}_n$ has a $\polylog(n)$-time $\PZK$ (and hence $\HVPZK$) protocol. Then, by Lemma~\ref{lem:pzk-closure}, $\GapAND{\overline{\PTP}_n}{n}{7/8}$ has a $\polylog(n)$-time $\HVPZK$ protocol.

    This fact, along with the lower bound in Corollary~\ref{cor:ptp}, can be used to construct an oracle separating $\HVPZK$ from $\PP$ by standard diagonalization. But by Lemma~\ref{lem:pzk-pp}, such an oracle cannot exist. So there has to be some oracle separating $\PZK$ and $\coPZK$, completing the proof of Theorem \ref{thm:pzk-copzk}.
\end{proof}

An argument identical to the proof of Theorem~\ref{thm:pzk-copzk} (without the need to invoke Lemma \ref{lem:nipzk-pzk})
shows that the same oracle separates $\NIPZK$ and $\coNIPZK$, as well as $\HVPZK$ and $\coHVPZK$. 

\begin{theo}
\label{thm:nipzk} The oracle $\Oracle$ witnessing Theorem \ref{thm:pzk-copzk} also satisfies $\NIPZK^\Oracle \neq \coNIPZK^{\Oracle}$, as well as $\HVPZK^\Oracle \neq \coHVPZK^{\Oracle}$.
\end{theo}


Combining Theorems \ref{lem:pzk-pp}, \ref{thm:niszk-pzk}, \ref{thm:pzk-copzk}, and \ref{thm:nipzk} yields Theorem \ref{thm:pzklimit} from Section \ref{sec:pzkintro}.

\label{sec:pzk}


\section{Communication Separation Between $\NISZK$ and $\UPP$}\label{sec:lowerbound-cc}
\label{sec:cc}
Based on the framework of Razborov and Sherstov~\cite{razborovsherstov}, and Bun and Thaler~\cite{bt16}, we are able to generalize Section \ref{sec:niszk-pp}'s separation between the query complexity classes $\NISZKdt$ and $\UPPdt$ to communication complexity. That is, we prove the following theorem (the communication complexity classes $\NISZKcc$ and $\UPPcc$ are formally defined in Appendix \ref{ccappendix}).
 
\begin{theo} \label{cctheorem} \label{THMCC}
$\NISZKcc \not\subset \UPPcc.$
\end{theo}

In light of Razborov and Sherstov's framework, proving the above theorem boils down to identifying some $f$ in $\NISZKdt$ such that $\degsh(\Gapmajority(f))$ is large, and moreover there is a dual witness to this fact that satisfies an additional \emph{smoothness} condition.
Unfortunately, the dual witness for $\Gapmajority(f)$ constructed in Theorem~\ref{theo:lift-UPP} is not smooth.  

Bun and Thaler described methods for ``smoothing out'' certain dual witnesses. However, their
methods were specifically described in the context of functions of the form $\OR(f)$, while
we must consider functions of the form $\Gapmajority(f)$. 
Nonetheless, we are able to apply Bun and Thaler's methodology to smooth out the 
dual witness for $\Gapmajority(f)$ that we constructed in Theorem~\ref{theo:lift-UPP}, for a non-trivial class of functions $f$. We thereby obtain the claimed separation of Theorem \ref{THMCC}.
While the proof of this result largely follows the same lines of Bun and Thaler's, 
there are many details and a few subtleties to work through. We present the proof in Appendix~\ref{sec:prelimcc} for completeness.

\section{Additional Consequences}
\label{sec:consequences}
\subsection{Consequences In Structural Complexity}
Theorem \ref{mainfcor}'s oracle separation between $\NISZK$ (or $\SZK$) and $\PP$ also answers a number of open questions in structural complexity; for instance it trivially implies oracle separations between $\SZK$ and $\BPPpath$, as well as separations between $\CQP$ \& $\DQP$ and $\PP$. The latter classes, defined by  Aaronson \cite{aaronsondqp}, and Aaronson, Bouland, Fitzsimons and Lee \cite{ABFL16} are complexity classes capturing the power of quantum computing with ``more powerful" modified versions of quantum mechanics. The authors of \cite{aaronsondqp,ABFL16} showed these classes are in $\EXP$ and $\BPP^{\#\mathsf{P}}$, respectively, and ask if one could improve the upper bounds on these classes to e.g. $\PP$ (which is an upper bound on $\BQP$) as an open problem. Since these classes contain $\SZK$, our result implies that one cannot place their classes in $\PP$ using relativizing techniques. This partially explains the difficulty in substantially improving their upper bounds.
This was part of our original motivation for studying this problem.

\subsection{Consequences for Polarization}
\label{sec:polarization}

A polarization algorithm is an algorithm that is given black-box sampling access to two distributions, and outputs two new distributions that are either extremely close in total variation distance (if they were initially somewhat close) or extremely far in total variation distance (if they were originally somewhat far). In this section we describe how our oracle separation between $\SZK$ and $\PP$ implies lower bounds on polarization algorithms. In particular we show black-box polarization algorithms are limited in how close they can push the statisical difference to $0$ or $1$ relative to the number of bits in the output distribution.

The concept of polarization first arose in work of Sahai and Vadhan \cite{sahai2003complete}. In their work, Sahai and Vadhan showed that the statistical difference problem is complete for the class $\SZK$. The statistical distance problem is formulated as follows: Let $P_b(x)$ be poly-sized classical circuits. Let $D_b$ be the distribution on $\{0,1\}^n$ induced by inputting a uniformly random input $x$ to $P_b(x)$. 
The statistical difference problem is, given circuits $P_0$ and $P_1$, determine if either $||D_0-D_1|| \leq 1/3$ or if $||D_0-D_1|| \geq 2/3$, promised one is the case. 
Here $||D_0-D_1||$ indicates the total variation distance between the distributions $D_0$ and $D_1$.

In their paper, Sahai and Vadhan also showed a remarkable property of the statistical difference problem -- namely that the constants 1/3 and 2/3 in the Statistical Difference problem can be amplified to be exponentially close to 0 and 1 \cite{sahai2003complete}. This property is not immediately obvious, because it cannot be obtained by simply repeatedly sampling from $D_0$ and $D_1$.
Nevertheless, they showed the following: given black-box distributions $D_0$ and $D_1$, and a number $k$ expressed in unary, then in polynomial time one can sample from distributions $D_0'$ and $D_1'$ (using polynomially many samples from $D_0$ and $D_1$) such that, if $||D_0-D_1||\leq 1/3$, then $||D_0'-D_1'||\leq \epsilon$ and if $||D_0-D_1|| \geq 2/3$, then $||D_0'-D_1'||\geq 1-\epsilon$, where $\epsilon=2^{-k}$. 
Hence without loss of generality, one can assume that the distributions in the statistical different problem are exponentially close to 0 or 1; their transformation ``polarizes" the distributions to be either very close or very far from one another. 
This is known as the Polarization Lemma, and is a key part of the proof that Statistical Difference is $\SZK$-complete\footnote{In statistical zero-knowledge proof systems, the verifier must be able to simulate the honest prover to negligibly small (1/superpoly) total variation distance. The ability to polarize distributions allows the statistical difference problem to have this property.}.

Given this fundamental result, it is natural to ask whether or not one can improve the parameters of the Polarization Lemma. For instance, Sahai and Vadhan noted in their paper that their algorithm could only polarize distributions under the promise $||D_0-D_1||>\alpha$ or $||D_0-D_1||<\beta$ in the case that $\alpha^2 > \beta$.
So their algorithm can polarize $\alpha=2/3$ and $\beta=1/3$, but not $\alpha=5/9$ and $\beta=4/9$. A natural question is whether or not this limitation could be removed. 
Holenstein and Renner answered this question in the negative for certain types of black-box polarization \cite{HR05}.
In particular, they showed that any form of black-box polarization which works by drawing strings $b,c \in \{0,1\}^\ell$, and then sets 
$D_0' = D_{b_1}\otimes ...\otimes D_{b_\ell}$ and $D_1' = D_{c_1}\otimes ...\otimes D_{c_\ell}$ cannot polarize in the case where $\alpha^2 <\beta$.
As Sahai and Vadhan's polarization algorithm took this form, this was strong evidence that this limitation was fundamental.
Note, however, that it remains open to show that polarization cannot occur when $\alpha^2<\beta$ using \emph{arbitrary} black-box algorithms. For instance, one could feed the random outputs of of $D_0$ back into the circuit for $D_1$ in order to help polarize the distributions. 
While it is not clear how these sorts of operations could help one polarize, it is difficult to rule out the possibility that such operations might lead to a stronger polarization algorithm.

In this paper we consider different parameters of the Polarization Lemma - namely how small can the security parameter $\epsilon$ be relative to the size of the range of the output distributions. 
For example, if one is given distributions $D_0$ and $D_1$ over $n$-bit strings with total variation distance $>2/3$ or $<1/3$, then can one create distributions $D_0'$ and $D_1'$ over $n'$-bit string such that the total variation distance is $\leq \epsilon$ or $\geq 1-\epsilon$ where $\epsilon = 2^{-n'}$, or $2^{-n'^2}$? 
At first it might appear the answer to the above question is trivially yes - because one can simply set $k=-n^2$ (or $k=n^c$ for any constant c) and run the Polarization Lemma.
However this does not work because the Polarization Lemma increases the size of the domains of the distributions as it polarizes; in other words $n'$ is some polynomial function of $n$ and $k$. 
By tweaking the parameters of the Polarization Lemma slightly \cite{sahai2003complete}, one can polarize distributions on $n$ bits to distributions on $n'=\poly(n)$ bits which are polarized to roughly $\epsilon\approx 2^{-\sqrt{n'}}$. However, it seems difficult to do better than $\epsilon=2^{-\sqrt{n'}}$ using the proof techniques of Sahai and Vadhan \cite{sahai2003complete}. This is because their proof alternates between two lemmas, one which total variation distance towards 1 in the case the distributions are far apart, and another which pushes the total variation distance towards zero in the case the distributions are close. In order to make the distributions $2^{-k}$-close or $1-2^{-k}$-far, one must apply both lemmas, each of which increases the number of bits output by the distribution by a factor of $k$. Hence using Sahai and Vadhan's Lemma with $k=n^c$, the best one can achieve are distributions on $n'=n^{2c+1}$ bits which are either $2^{-n^c}$-close or $(1-2^{-n^c})$-far. For large constant $c$ this gives $\epsilon \approx 2^{-\sqrt{n'}}$. It seems difficult to improve their lemma further using the techniques of their paper. 

A natural question is therefore: what is the smallest value of $\epsilon$ that one can achieve relative to the size of the output distributions $n'$? 
In this work, we show that if $\epsilon$ can be made very small relative to $n'$, then that would place $\SZK^\mathcal{O} \subseteq \PP^\mathcal{O}$ (and even $\SZK^\mathcal{O} \subseteq \BPPPATH^\mathcal{O}$) for all oracles $\mathcal{O}$.
Therefore, as a corollary of our main result, $\epsilon$ cannot be made very small by any poly-time black-box polarization algorithm.
More specifically, we achieve a lower bound of $\epsilon>2^{-n'/2 -1}$ for any poly-time polarization algorithm. 
%

More specifically, we prove two theorems showing that a stronger version of polarization places $\SZK$ in $\PP$ relative to all oracles. Therefore, a stronger polarization algorithm cannot exist as a corollary of Theorem \ref{mainfcor}.

\begin{theo} Suppose that there is an algorithm running in poly$(n)$ time, which given black box distributions $D_0$, $D_1$ on strings of length $n$ which obey either $|D_0-D_1|<1/3$ or $|D_0-D_1|>2/3$,  produces two output distributions $D_0'$ and $D_1'$ on strings of length $n'=\text{poly}(n)$ such that either $|D_0'-D_1'|<\epsilon$ (in the first case) or $|D_0'-D_1'|>1-\epsilon$ (in the second case) where $\epsilon \leq 2^{-n'/2 -1}$.  Then $\SZK^{\mathcal{O}}\subseteq \PP^{\mathcal{O}}$ for all oracles $\mathcal{O}$. \label{thm:szkppandpolarization}
\end{theo}

\begin{theo} Suppose that there is an algorithm running in poly$(n)$ time, which given black box distributions $D_0$, $D_1$ on strings of length $n$ which obey either $|D_0-D_1|<1/3$ or $|D_0-D_1|>2/3$,  produces two output distributions $D_0'$ and $D_1'$ on strings of length $n'=\text{poly}(n)$ such that either $|D_0'-D_1'|<\epsilon$ (in the first case) or $|D_0'-D_1'|>1-\epsilon$ (in the second case) where $\epsilon \leq 2^{-2n'/3 -1}$.  Then $\SZK^{\mathcal{O}}\subseteq \left( \BPPPATH\right)^{\mathcal{O}}$ for all oracles $\mathcal{O}$. \label{thm:szkbpppathandpolarization}
\end{theo}

Therefore as a corollary of Theorem \ref{mainfcor}, there do not exist poly-time polarization algorithms achieving $\epsilon = 2^{-n'/2-1}$. In fact one could have achieved such a lower bound even if one had merely given an oracle separation between $\SZK$ and $\BPPPATH$.
It remains open to close the gap between our lower bound of $\epsilon = 2^{-n'/2-1}$ and the upper bound of $\epsilon = 2^{-n'^{1/2 + \delta}}$ for any $\delta>0$ given by Sahai and Vadhan \cite{sahai2003complete}. 

The proof of Theorem \ref{thm:szkppandpolarization} is relatively straightforward. Suppose one can polarize to $\epsilon'\ll2^{-n'/2}$. Then the output distributions now have a promise on the $\ell_2$ distance between the output distributions - in particular the $\ell_2$ distance between them is more or less than some (exponentially small) threshold. It is easy to decide this problem $\PP$ - this is because the $\ell_2$ distance square is a degree-two polynomial in the output probabilities. To see this, say you're trying to determine if $S=\sum_{x\in\{0,1\}^n} (D_0'(x)-D_1'(x))^2 $
is more or less than some threshold $t$, 
consider the following algorithm: pick at random x, pick a random number 1,2,3 or 4. If the number is 1 (respectively 4) sample two samples from $D_0'$ (respectively $D_1'$) and accept if they both give output x, otherwise output accept/reject using a 50-50 coin flip. If the number is 2 or 3 sample one sample from $D_0'$ and $D_1'$ and reject iff they collide, otherwise output a 50-50 coin flip. The probabiltiy this machine accepts is $1/2 + S/2 $ - which is more more or less than a known threshold $(1+t)/2$. Therefore by correcting the bias of the machine with an initial coin flip, this is a $\PP$ algorithm to decide the problem.
In short, deciding thresholds for the $\ell_2$ norm is easy for $\PP$ because it is a low-degree polynomial, while deciding thresholds for the $\ell_1$ norm is hard for $\PP$ because the $\ell_1$ norm is not a low degree polynomial.

On the other hand, the proof of Theorem \ref{thm:szkbpppathandpolarization} is involved -- it works by examining the algorithms of Aaronson \cite{aaronsondqp} and Aaronson, Bouland, Fitzsimons and Lee \cite{ABFL16} showing that certain modified versions of quantum mechanics can be used to solve $\SZK$-hard problems in polynomial time. 
These algorithms are not based on postselection (otherwise they would place $\SZK\subseteq \PostBQP=\PP$ for all oracles, a contradiction with our main result). 
However, it turns out that if one has a very strong polarization lemma, then one can turn them into postselected quantum algorithms (and even postselected classical algorithms) for statistical difference. 
Interestingly, this was part of our original motivation for this work. 
We include this proof in Appendix \ref{polarizationbpppath} for the interested reader. 
\subsection{Consequences for Property Testing}
\medskip 
\noindent \textbf{Lower Bounds for Property Testers That Barely Do Better Than Random Guessing.} For any $\NISZK$-hard property testing problem $P$, our query complexity lower bounds immediately imply that any property testing algorithm for $P$ that outputs the correct answer with probability strictly greater than $1/2$ requires $n^{\Omega(1)}$ queries. 
For concreteness, we highlight the result we obtain for the $\NISZK$-complete
problem of entropy approximation. Specifically, given a distribution $D$ over $n$ elements, a natural problem is to ask how many samples from $D$ are required to estimate the entropy of $D$ to additive error. In 2011, Valiant and Valiant  \cite{ValiantValiant2011} showed that to achieve any constant additive error less than $\log 2 /2$, it is both necessary and sufficient to take $\Theta(n/\log n)$ samples from $D$. However, their bounds assume that one wishes to estimate the entropy with high probability, say with probability $1- o(1/\text{poly}(n))$.
Quantitatively, our $\UPPdt$ query lower bounds imply the following.

\begin{cor} Any algorithm which decides if the entropy of $D$ (over domain size $n$) is $\leq k-1$ or $\geq k+1$ and succeeds with probability $>\frac{1}{2}$ requires $\Omega(\lb)$ samples from $D$.
\end{cor}

In other words, estimating the entropy of a distribution to additive error 2 requires $\tilde{\Omega}(n^{1/4})$ samples
, even if the algorithm is only required to have an arbitrarily small bias in deciding the answer correctly.

\subsection{Consequences for Delegating Computation}
\label{sec:sips}
In this section, we point out an easy implication of our results:
two-message streaming interactive proofs (SIPs) \cite{cty}
of logarithmic cost can compute functions outside of $\UPPcc$.

In a SIP,
a verifier with limited working memory makes a single streaming pass over an input $x$, and then
interacts with an untrusted prover, who evaluates a function $f$ of the input, and attempts to convince the verifier
of the value of $f(x)$. The protocol must satisfy standard notions of completeness and soundness. The cost of the protocol is the size of the verifier's working memory and the total length of the messages exchanged between the prover and verifier. We direct the interested reader to \cite{cty}
for the formal definition.

It follows from our analysis in Appendix \ref{ccappendix} that the $(O(n),n,F)$-pattern matrix of the function $F:=\GapMaj{\PTP_{n^{3/4}}}{n^{1/4}}{.499}$ specifies a communication problem that is outside of $\UPPcc$ (see Appendix \ref{sec:pmat} for a definition of pattern matrices).
For our purposes, the relevant properties of such pattern matrices are as follows. In the communication problem $F^{\cc}(x, y)$ corresponding to the pattern matrix of a function $F \colon \{0, 1\}^n \to \{0, 1\}$, Alice's input $x$ and Bob's input $y$ together specify a vector $u(x, y) \in \{0, 1\}^n$, and $F^{\cc}(x,y)$ is defined to equal $F(u(x, y))$. Moreover, each coordinate of $u(x, y)$ depends on $O(1)$ entries of $x$ and $y$ respectively. 

Observe that $F$ is computed by a simple two-message interactive proof in which the verifier makes $O(\log n)$ \emph{non-adaptive} queries to bits of the input $x$, uses $O(\log n)$ bits of working memory, and the total communication cost is also $O(\log n)$: the verifier
picks an instance of $\PTP_{n^{3/4}}$ at random, and runs the $\SZK$ protocol for $\PTP$
described in Section \ref{sec:ptpprotocol} on that instance (we do not need the zero-knowledge property of the $\SZK$ protocol here). Completeness and soundness follow from the definition of $\Gapmajority$ and completeness and soundness of the $\SZK$ protocol for $\PTP$. 

Consider a data stream consisting of Alice and Bob's inputs to $F^{\cc}$, and a SIP verifier who
wishes to compute $F^{\cc}(x, y)$. That is, the first part of the stream specifies $x$ and the second part specifies $y$. There is a simple two-message SIP for evaluating $F^{\cc}(x, y)$: 
the SIP verifier simulates the verifier in the above interactive proof for $F(u(x, y))$. Since 
the latter verifier only needs to know $O(\log n)$ bits of $u(x, y)$, and each bit of $u(x, y)$ depends
on $O(1)$ bits of $x$ and $y$, the SIP verifier can compute the bits of $u(x, y)$ that are necessary to run the simulation, using just $O(\log n)$ bits of memory and a single streaming pass over the input. 
We obtain the following theorem.

\begin{theo} \label{siptheo}
There is a communication problem $F^{\cc}(x, y) \colon \{0, 1\}^{n} \times \{0, 1\}^{n} \to \{0, 1\}$
such that $F^{\cc}$ is not in $\UPPcc$, yet given a data stream specifying $x$ followed by $y$,
there is a 2-message SIP of cost $O(\log n)$ computing $F^{\cc}(x, y)$. 
\end{theo}

Theorem \ref{siptheo} provides an explanation for why prior work attempting to understand the power of 2-message SIPs has succeeded only in proving lower bounds on special classes of such protocols \cite{CCMTV}. Indeed, taking $\UPPcc$ to represent the limit of our methods for proving lower bounds in communication complexity, the fact that 2-message SIPs and their analogous two-party communication model (called $\textbf{OIP}_{\mathbf{+}}^{\mathbf{[2]}}$ in \cite{CCMTV}) can compute functions outside of $\UPPcc$ means that proving superlogarithmic lower bounds for 2-message SIPs will require new lower bound techniques.

\subsubsection{Consequences for the Algebrization Barrier}

Arithmetization is a powerful technique in complexity theory that is used, for example, to prove $\mathsf{IP} = \PSPACE$~\cite{shamir1992ip,lund1992algebraic} and many other celebrated theorems. 
While arithmetization circumvents the \emph{relativization} barrier, Aaronson and Wigderson~\cite{aaronson2009algebrization} proposed a new \emph{algebrization} barrier. Roughly speaking, their results show that arithmetization alone will not suffice to resolve many open questions in complexity theory. Informally, 
one of their key results was the following.

\begin{theo}[Informal, implicit in Theorem 5.11 of \cite{aaronson2009algebrization}]
	For two complexity classes $\mathcal{C}$ and $\mathcal{D}$, if $\mathcal{C}^{\mathsf{cc}} \not\subset \mathcal{D}^{\mathsf{cc}}$, then arithmetization techniques alone cannot prove $\mathcal{C} \subseteq \mathcal{D}$.
\end{theo}

From the above theorem, our communication class separation $\NISZKcc \not\subset \UPPcc$ immediately implies the following informal corollary.

\begin{cor}[Informal]
	Arithmetization techniques alone cannot prove that $\NISZK \subseteq \PP$.
\end{cor}


\section{Open Problems}

Our works leaves a number of open related problems.
First, we have shown that the function $\Gapmajority(f)$ is hard for $\UPPdt$, for any function $f$ of high approximate degree, and that $\Gapand(f)$ is hard for $\UPPdt$, for any function of high positive one-sided approximate degree. Can one extend this work to characterize when $f\circ g$ is hard for $\UPPdt$, based on some properties of $f$ and $g$? We conjecture that the $\UPPdt$
complexity of $\Gapmajority(f)$ (respectively, $\Gapand(f)$) is characterized by the \emph{rational approximate degree} of $f$ (respectively, positive one-sided approximate degree of $f$).
Such a result would complement the characterization of the threshold degree 
of $\AND(f)$ in terms of positive one-sided rational approximate degree given in \cite{sherstov2014breaking}.

Additionally, we have shown a lower bound on certain parameters of the polarization lemma. Is there a polarization algorithm which matches our lower bound?

It would also be interesting to determine whether our lower bounds on property testing algorithms that output the correct answer with probability strictly greater than $1/2$ are quantitatively tight. For example, is there an algorithm that, given query access to a distribution $D$ (over domain size $n$) that is promised to have entropy $\leq k - 1$ or $\geq k + 1$, decides which is the case with probability greater than $1/2$, using $\tilde{O}(n^{1/4})$ samples from $D$?

Finally, the main open question highlighted by our work is to break through the $\UPP$
frontier in communication complexity. We formalize this question via the following challenge:
 prove any superlogarithmic lower 
bound for an explicit problem in a natural communication model that cannot be efficiently 
simulated by $\UPPcc$. Our work shows that any communication model capable of efficiently computing
the pattern matrix of $\Gapmajority(\PTP)$ is a candidate for achieving this goal. 
Thomas Watson has suggested the following as perhaps the simplest such candidate: consider the $\NISZKcc$ model, but restricted to be one-way, in the sense that neither
Merlin nor Bob can talk to Alice. This model effectively combines the key features of the $\NISZKcc$ and  $\textbf{OIP}_{\mathbf{+}}^{\mathbf{[2]}}$ (cf. \cite{CCMTV}) communication models. 
There is a logarithmic cost ``one-way $\NISZK$'' protocol for the pattern matrix of $\Gapmajority(\PTP)$, so this model cannot be efficiently simulated by $\UPPcc$. 
Curiously, despite the ability of this model to compute functions outside of $\UPPcc$, to the best of our knowledge it is possible that even the $\mathsf{INDEX}$ function requires polynomial cost in this model. 
Note that while Chakrabarti et al. \cite{CCMTV} gave an efficient $\textbf{OIP}_{\mathbf{+}}^{\mathbf{[2]}}$ communication protocol for $\mathsf{INDEX}$, their protocol is not zero-knowledge.

\section*{Acknowledgments}
We thank Scott Aaronson, Jayadev Acharya, Shalev Ben-David, Cl\'{e}ment Canonne, Oded Goldreich, Mika G{\"{o}}{\"{o}}s, Gautam Kamath, Robin Kothari, Tomoyuki Morimae, Harumichi Nishimura,  Ron Rothblum, Mike Saks, Salil Vadhan and Thomas Watson for helpful discussions.
Adam Bouland was supported in part by the NSF Graduate Research Fellowship under grant no. 1122374 and by the NSF Alan T. Waterman award under grant no. 1249349. Lijie Chen was supported in part by the National Basic Research Program of China Grant 2011CBA00300, 2011CBA00301, the National Natural Science Foundation of China Grant 61361136003.
Dhiraj Holden was supported in part by an Akamai Presidential fellowship, by NSF MACS - CNS-1413920, and by a
SIMONS Investigator award Agreement Dated 6-5-12.
Prashant Vasudevan was  supported in part by the Qatar Computing Research Institute under the QCRI-CSAIL partnership, and by the National Science Foundation Frontier grant CNS 1413920.

\bibliographystyle{alpha}
\bibliography{team} 

\newcommand{\etalchar}[1]{$^{#1}$}
\begin{thebibliography}{DGOW95}

\bibitem[Aar]{AaronsonPersonal}
Scott Aaronson.
\newblock Personal communication.

\bibitem[Aar02]{aaronson2002quantum}
Scott Aaronson.
\newblock Quantum lower bound for the collision problem.
\newblock In {\em Proceedings of the thiry-fourth annual ACM symposium on
  Theory of computing}, pages 635--642. ACM, 2002.

\bibitem[Aar05]{aaronsondqp}
Scott Aaronson.
\newblock Quantum computing and hidden variables.
\newblock {\em Physical Review A}, 71(3):032325, 2005.

\bibitem[Aar12]{aaronson2012impossibility}
Scott Aaronson.
\newblock Impossibility of succinct quantum proofs for collision-freeness.
\newblock {\em Quantum Information \& Computation}, 12(1-2):21--28, 2012.

\bibitem[ABFL16]{ABFL16}
Scott Aaronson, Adam Bouland, Joseph Fitzsimons, and Mitchell Lee.
\newblock The space ``just above" {BQP}.
\newblock In {\em Proceedings of the 2016 {ACM} Conference on Innovations in
  Theoretical Computer Science, Cambridge, MA, USA, January 14-16, 2016}, pages
  271--280, 2016.

\bibitem[AH91a]{AHPZKBPP}
William Aiello and Johan H{\aa}stad.
\newblock Relativized perfect zero knowledge is not {BPP}.
\newblock {\em Information and Computation}, 93:223--240, 1991.

\bibitem[AH91b]{AH91}
William Aiello and Johan H{\aa}stad.
\newblock Statistical zero-knowledge languages can be recognized in two rounds.
\newblock {\em J. Comput. Syst. Sci.}, 42(3):327--345, 1991.

\bibitem[AIKP15]{AIKP15}
Shweta Agrawal, Yuval Ishai, Dakshita Khurana, and Anat Paskin{-}Cherniavsky.
\newblock Statistical randomized encodings: {A} complexity theoretic view.
\newblock In {\em Automata, Languages, and Programming - 42nd International
  Colloquium, {ICALP} 2015, Kyoto, Japan, July 6-10, 2015, Proceedings, Part
  {I}}, pages 1--13, 2015.

\bibitem[Amb05]{ambainis2005polynomial}
Andris Ambainis.
\newblock Polynomial degree and lower bounds in quantum complexity: Collision
  and element distinctness with small range.
\newblock {\em Theory of Computing}, 1(1):37--46, 2005.

\bibitem[AS04]{aaronson2004quantum}
Scott Aaronson and Yaoyun Shi.
\newblock Quantum lower bounds for the collision and the element distinctness
  problems.
\newblock {\em Journal of the ACM (JACM)}, 51(4):595--605, 2004.

\bibitem[AW09]{aaronson2009algebrization}
Scott Aaronson and Avi Wigderson.
\newblock Algebrization: A new barrier in complexity theory.
\newblock {\em ACM Transactions on Computation Theory (TOCT)}, 1(1):2, 2009.

\bibitem[Bar01]{Barak01}
Boaz Barak.
\newblock How to go beyond the black-box simulation barrier.
\newblock In {\em 42nd Annual Symposium on Foundations of Computer Science,
  {FOCS} 2001, 14-17 October 2001, Las Vegas, Nevada, {USA}}, pages 106--115,
  2001.

\bibitem[BFS86]{babai1986complexity}
Laszlo Babai, Peter Frankl, and Janos Simon.
\newblock Complexity classes in communication complexity theory.
\newblock In {\em Foundations of Computer Science, 1986., 27th Annual Symposium
  on}, pages 337--347. IEEE, 1986.

\bibitem[BHZ87]{bhz87}
Ravi~B. Boppana, Johan H{\aa}stad, and Stathis Zachos.
\newblock Does co-np have short interactive proofs?
\newblock {\em Inf. Process. Lett.}, 25(2):127--132, 1987.

\bibitem[BRS91]{BRS91}
Richard Beigel, Nick Reingold, and Daniel~A. Spielman.
\newblock {PP} is closed under intersection (extended abstract).
\newblock In {\em Proceedings of the 23rd Annual {ACM} Symposium on Theory of
  Computing, May 5-8, 1991, New Orleans, Louisiana, {USA}}, pages 1--9, 1991.

\bibitem[BT15a]{bun2015dual}
Mark Bun and Justin Thaler.
\newblock Dual polynomials for collision and element distinctness.
\newblock {\em arXiv preprint arXiv:1503.07261}, 2015.

\bibitem[BT15b]{bun2015hardness}
Mark Bun and Justin Thaler.
\newblock Hardness amplification and the approximate degree of constant-depth
  circuits.
\newblock In {\em International Colloquium on Automata, Languages, and
  Programming}, pages 268--280. Springer, 2015.

\bibitem[BT16]{bt16}
Mark Bun and Justin Thaler.
\newblock Improved bounds on the sign-rank of {AC}$^0$.
\newblock {\em Electronic Colloquium on Computational Complexity {(ECCC)}},
  23:75, 2016.

\bibitem[BT17]{btlatest}
Mark Bun and Justin Thaler.
\newblock A nearly optimal lower bound on the approximate degree of
  ac\({}^{\mbox{0}}\).
\newblock {\em Electronic Colloquium on Computational Complexity {(ECCC)}},
  24:51, 2017.

\bibitem[CCG{\etalchar{+}}94]{CCGHHRR94}
Richard Chang, Benny Chor, Oded Goldreich, Juris Hartmanis, Johan H{\aa}stad,
  Desh Ranjan, and Pankaj Rohatgi.
\newblock The random oracle hypothesis is false.
\newblock {\em J. Comput. Syst. Sci.}, 49(1):24--39, 1994.

\bibitem[CCM{\etalchar{+}}15]{CCMTV}
Amit Chakrabarti, Graham Cormode, Andrew McGregor, Justin Thaler, and Suresh
  Venkatasubramanian.
\newblock Verifiable stream computation and arthur-merlin communication.
\newblock In David Zuckerman, editor, {\em 30th Conference on Computational
  Complexity, {CCC} 2015, June 17-19, 2015, Portland, Oregon, {USA}}, volume~33
  of {\em LIPIcs}, pages 217--243. Schloss Dagstuhl - Leibniz-Zentrum fuer
  Informatik, 2015.

\bibitem[Che16a]{lijiepszk}
Lijie Chen.
\newblock Adaptivity vs postselection.
\newblock {\em arXiv:1606.04016}, 2016.

\bibitem[Che16b]{lijiepszkqszk}
Lijie Chen.
\newblock A note on oracle separations for $\textsf{BQP}$.
\newblock {\em arXiv:1605.00619}, 2016.

\bibitem[CTY11]{cty}
Graham Cormode, Justin Thaler, and Ke~Yi.
\newblock Verifying computations with streaming interactive proofs.
\newblock {\em {PVLDB}}, 5(1):25--36, 2011.

\bibitem[DGJ{\etalchar{+}}10]{diakonikolas2010bounded}
Ilias Diakonikolas, Parikshit Gopalan, Ragesh Jaiswal, Rocco~A Servedio, and
  Emanuele Viola.
\newblock Bounded independence fools halfspaces.
\newblock {\em SIAM Journal on Computing}, 39(8):3441--3462, 2010.

\bibitem[DGOW95]{DGOW95}
Ivan Damg{\aa}rd, Oded Goldreich, Tatsuaki Okamoto, and Avi Wigderson.
\newblock Honest verifier vs dishonest verifier in public coin zero-knowledge
  proofs.
\newblock In {\em Advances in Cryptology - {CRYPTO} '95, 15th Annual
  International Cryptology Conference, Santa Barbara, California, USA, August
  27-31, 1995, Proceedings}, pages 325--338, 1995.

\bibitem[DR96]{dubhashi1996balls}
Devdatt~P Dubhashi and Desh Ranjan.
\newblock Balls and bins: A study in negative dependence.
\newblock {\em BRICS Report Series}, 3(25), 1996.

\bibitem[EFHK14]{KP14b}
Javier Esparza, Pierre Fraigniaud, Thore Husfeldt, and Elias Koutsoupias,
  editors.
\newblock {\em Automata, Languages, and Programming - 41st International
  Colloquium, {ICALP} 2014, Copenhagen, Denmark, July 8-11, 2014, Proceedings,
  Part {I}}, volume 8572 of {\em Lecture Notes in Computer Science}. Springer,
  2014.

\bibitem[Fis02]{Fischlin2002}
Marc Fischlin.
\newblock On the impossibility of constructing non-interactive
  statistically-secret protocols from any trapdoor one-way function.
\newblock In {\em Proceedings of the The Cryptographer's Track at the RSA
  Conference on Topics in Cryptology}, CT-RSA '02, pages 79--95, London, UK,
  UK, 2002. Springer-Verlag.

\bibitem[For87]{Fortnow87}
Lance Fortnow.
\newblock The complexity of perfect zero-knowledge (extended abstract).
\newblock In {\em Proceedings of the 19th Annual {ACM} Symposium on Theory of
  Computing, 1987, New York, New York, {USA}}, pages 204--209, 1987.

\bibitem[GG98]{goldreich1998limits}
Oded Goldreich and Shafi Goldwasser.
\newblock On the limits of non-approximability of lattice problems.
\newblock In {\em Proceedings of the thirtieth annual ACM symposium on Theory
  of computing}, pages 1--9. ACM, 1998.

\bibitem[GMR89]{goldwasser1989knowledge}
Shafi Goldwasser, Silvio Micali, and Charles Rackoff.
\newblock The knowledge complexity of interactive proof systems.
\newblock {\em SIAM Journal on computing}, 18(1):186--208, 1989.

\bibitem[GMW91]{goldreich1991proofs}
Oded Goldreich, Silvio Micali, and Avi Wigderson.
\newblock Proofs that yield nothing but their validity or all languages in {NP}
  have zero-knowledge proof systems.
\newblock {\em Journal of the ACM (JACM)}, 38(3):690--728, 1991.

\bibitem[Gol15]{GoldwasserVideo}
Shafi Goldwasser.
\newblock Zero knowledge probabilistic proof systems.
\newblock \url{https://www.youtube.com/watch?v=J4TkHuTmHsg#t=1h15m20s}, 2015.

\bibitem[GPW15a]{landscape}
Mika G{\"{o}}{\"{o}}s, Toniann Pitassi, and Thomas Watson.
\newblock The landscape of communication complexity classes.
\newblock {\em Electronic Colloquium on Computational Complexity {(ECCC)}},
  22:49, 2015.

\bibitem[GPW15b]{goositcs}
Mika G{\"{o}}{\"{o}}s, Toniann Pitassi, and Thomas Watson.
\newblock Zero-information protocols and unambiguity in arthur-merlin
  communication.
\newblock In Tim Roughgarden, editor, {\em Proceedings of the 2015 Conference
  on Innovations in Theoretical Computer Science, {ITCS} 2015, Rehovot, Israel,
  January 11-13, 2015}, pages 113--122. {ACM}, 2015.

\bibitem[GSV98]{GSV98}
Oded Goldreich, Amit Sahai, and Salil~P. Vadhan.
\newblock Honest-verifier statistical zero-knowledge equals general statistical
  zero-knowledge.
\newblock In {\em Proceedings of the Thirtieth Annual {ACM} Symposium on the
  Theory of Computing, Dallas, Texas, USA, May 23-26, 1998}, pages 399--408,
  1998.

\bibitem[GSV99]{goldreich1999can}
Oded Goldreich, Amit Sahai, and Salil Vadhan.
\newblock Can statistical zero knowledge be made non-interactive? or on the
  relationship of {SZK} and {NISZK}.
\newblock In {\em Advances in Cryptology—CRYPTO’99}, pages 467--484.
  Springer, 1999.

\bibitem[GT14]{GT14}
Oded Goldreich and Liav Teichner.
\newblock Super-perfect zero-knowledge proofs.
\newblock {\em Electronic Colloquium on Computational Complexity {(ECCC)}},
  21:97, 2014.

\bibitem[HR05]{HR05}
Thomas Holenstein and Renato Renner.
\newblock One-way secret-key agreement and applications to circuit polarization
  and immunization of public-key encryption.
\newblock In {\em Advances in Cryptology - {CRYPTO} 2005: 25th Annual
  International Cryptology Conference, Santa Barbara, California, USA, August
  14-18, 2005, Proceedings}, pages 478--493, 2005.

\bibitem[Kla11]{klauck11}
Hartmut Klauck.
\newblock On arthur merlin games in communication complexity.
\newblock In {\em Proceedings of the 26th Annual {IEEE} Conference on
  Computational Complexity, {CCC} 2011, San Jose, California, June 8-10, 2011},
  pages 189--199. {IEEE} Computer Society, 2011.

\bibitem[KP14]{KP14a}
Hartmut Klauck and Supartha Podder.
\newblock Two results about quantum messages.
\newblock In Erzs{\'{e}}bet Csuhaj{-}Varj{\'{u}}, Martin Dietzfelbinger, and
  Zolt{\'{a}}n {\'{E}}sik, editors, {\em Mathematical Foundations of Computer
  Science 2014 - 39th International Symposium, {MFCS} 2014, Budapest, Hungary,
  August 25-29, 2014. Proceedings, Part {II}}, volume 8635 of {\em Lecture
  Notes in Computer Science}, pages 445--456. Springer, 2014.

\bibitem[KT14]{varun}
Varun Kanade and Justin Thaler.
\newblock Distribution-independent reliable learning.
\newblock In Maria{-}Florina Balcan, Vitaly Feldman, and Csaba
  Szepesv{\'{a}}ri, editors, {\em Proceedings of The 27th Conference on
  Learning Theory, {COLT} 2014, Barcelona, Spain, June 13-15, 2014}, volume~35
  of {\em {JMLR} Workshop and Conference Proceedings}, pages 3--24. JMLR.org,
  2014.

\bibitem[Kut05]{kutin2005quantum}
Samuel Kutin.
\newblock Quantum lower bound for the collision problem with small range.
\newblock {\em Theory of Computing}, 1(1):29--36, 2005.

\bibitem[LFKN92]{lund1992algebraic}
Carsten Lund, Lance Fortnow, Howard Karloff, and Noam Nisan.
\newblock Algebraic methods for interactive proof systems.
\newblock {\em Journal of the ACM (JACM)}, 39(4):859--868, 1992.

\bibitem[Lok01]{Lokam01}
Satyanarayana~V. Lokam.
\newblock Spectral methods for matrix rigidity with applications to size-depth
  trade-offs and communication complexity.
\newblock {\em J. Comput. Syst. Sci.}, 63(3):449--473, 2001.

\bibitem[LS09]{LS09}
Nathan Linial and Adi Shraibman.
\newblock Learning complexity vs communication complexity.
\newblock {\em Combinatorics, Probability {\&} Computing}, 18(1-2):227--245,
  2009.

\bibitem[LZ16]{LovettZhang16}
Shachar Lovett and Jiapeng Zhang.
\newblock On the impossibility of entropy reversal, and its application to
  zero-knowledge proofs.
\newblock {\em ECCC TR16-118}, July 31 2016.

\bibitem[Mal15]{Malka2015}
Lior Malka.
\newblock How to achieve perfect simulation and a complete problem for
  non-interactive perfect zero-knowledge.
\newblock {\em Journal of Cryptology}, 28(3):533--550, 2015.

\bibitem[Oka96]{okamoto1996relationships}
Tatsuaki Okamoto.
\newblock On relationships between statistical zero-knowledge proofs.
\newblock In {\em Proceedings of the twenty-eighth annual ACM symposium on
  Theory of computing}, pages 649--658. ACM, 1996.

\bibitem[Pat92]{paturi1992degree}
Ramamohan Paturi.
\newblock On the degree of polynomials that approximate symmetric boolean
  functions (preliminary version).
\newblock In {\em Proceedings of the twenty-fourth annual ACM symposium on
  Theory of computing}, pages 468--474. ACM, 1992.

\bibitem[PS86]{paturisimon}
Ramamohan Paturi and Janos Simon.
\newblock Probabilistic communication complexity.
\newblock {\em J. Comput. Syst. Sci.}, 33(1):106--123, 1986.

\bibitem[PSS14]{PSS14}
Periklis~A. Papakonstantinou, Dominik Scheder, and Hao Song.
\newblock Overlays and limited memory communication.
\newblock In {\em {IEEE} 29th Conference on Computational Complexity, {CCC}
  2014, Vancouver, BC, Canada, June 11-13, 2014}, pages 298--308. {IEEE}
  Computer Society, 2014.

\bibitem[PV08]{peikert2008noninteractive}
Chris Peikert and Vinod Vaikuntanathan.
\newblock Noninteractive statistical zero-knowledge proofs for lattice
  problems.
\newblock In {\em Annual International Cryptology Conference}, pages 536--553.
  Springer, 2008.

\bibitem[RS10]{razborovsherstov}
Alexander~A. Razborov and Alexander~A. Sherstov.
\newblock The sign-rank of {AC}$^0$.
\newblock {\em {SIAM} J. Comput.}, 39(5):1833--1855, 2010.

\bibitem[Sha92]{shamir1992ip}
Adi Shamir.
\newblock {IP}={PSPACE}.
\newblock {\em Journal of the ACM (JACM)}, 39(4):869--877, 1992.

\bibitem[She11]{sherstov2011pattern}
Alexander~A Sherstov.
\newblock The pattern matrix method.
\newblock {\em SIAM Journal on Computing}, 40(6):1969--2000, 2011.

\bibitem[She14]{sherstov2014breaking}
Alexander~A Sherstov.
\newblock Breaking the {M}insky-{P}apert barrier for constant-depth circuits.
\newblock In {\em Proceedings of the 46th Annual ACM Symposium on Theory of
  Computing}, pages 223--232. ACM, 2014.

\bibitem[She15]{sherstov2015power}
Alexander~A Sherstov.
\newblock The power of asymmetry in constant-depth circuits.
\newblock In {\em Foundations of Computer Science (FOCS), 2015 IEEE 56th Annual
  Symposium on}, pages 431--450. IEEE, 2015.

\bibitem[SV03]{sahai2003complete}
Amit Sahai and Salil Vadhan.
\newblock A complete problem for statistical zero knowledge.
\newblock {\em Journal of the ACM (JACM)}, 50(2):196--249, 2003.

\bibitem[Tha14]{thalericalp}
Justin Thaler.
\newblock Lower bounds for the approximate degree of block-composed functions.
\newblock {\em Electronic Colloquium on Computational Complexity {(ECCC)}},
  21:150, 2014.

\bibitem[Tod91]{Toda91}
Seinosuke Toda.
\newblock {PP} is as hard as the polynomial-time hierarchy.
\newblock {\em {SIAM} J. Comput.}, 20(5):865--877, 1991.

\bibitem[Ver95]{Vereshchagin95}
Nikolai~K. Vereshchagin.
\newblock Lower bounds for perceptrons solving some separation problems and
  oracle separation of {AM} from {PP}.
\newblock In {\em Third Israel Symposium on Theory of Computing and Systems,
  {ISTCS} 1995, Tel Aviv, Israel, January 4-6, 1995, Proceedings}, pages
  46--51, 1995.

\bibitem[VV11]{ValiantValiant2011}
Gregory Valiant and Paul Valiant.
\newblock Estimating the unseen: an n/log(n)-sample estimator for entropy and
  support size, shown optimal via new {CLTs}.
\newblock In {\em Proceedings of the 43rd {ACM} Symposium on Theory of
  Computing, {STOC} 2011, San Jose, CA, USA, 6-8 June 2011}, pages 685--694,
  2011.

\end{thebibliography}

\appendix

\section{Upper Bounds on $\UPPdt(\Gapmajority(f))$ and $\UPPdt(\Gapand(f))$}
\label{sec:char}

In Section~\ref{sec:composition}, we proved that if $f$ has a high {\em approximate degree} ({\em positive one-sided approximate degree}), then $\Gapmajority(f)$ ($\Gapand(f)$) is hard for $\UPP$ algorithms. In this section we show that condition is also necessary: when $\odegap(f)$ is small, $\Gapand(f)$ has a lower $\UPP$ query complexity; and when $\dega(f)$ is small, $\Gapmajority(f)$ has a low $\UPP$ query complexity. Formally, we have:

\begin{theo}\label{theo:char}
	For a partial function $f$ with input length $n$, and a positive integer $m$.
	
	\begin{align*}
	\UPPdt(\GapMaj{f}{m}{2/3}) &= O(\dega(f)). \\ 
	\UPPdt(\GapAND{f}{m}{2/3})    &= O(\odegap(f)). 
	\end{align*}
	
	Recall that $\odegap(f) := \odegap_{1/3}(f)$, and $\dega(f) := \dega_{1/3}(f)$.
\end{theo}

The choice of the constants $2/3$ and $1/3$ is only for convenience. We need the following standard fact for approximate degree and positive one-sided approximate degree, (for example, see ~\cite[Claim 4.3]{diakonikolas2010bounded} and ~\cite[Fact~2.4]{sherstov2015power}).

\begin{claim}\label{claim:const-ok}
	For any constant $0 < \epsilon < 0.5$, and any partial function $f$, $\dega_{\epsilon}(f) = \Theta(\dega(f))$ and $\odegap_{\epsilon}(f) = \Theta(\odegap(f))$.
\end{claim}

Now we prove Theorem~\ref{theo:char}.

\begin{proofof}{Theorem~\ref{theo:char}}
	
	We first upper bound $\UPPdt(\GapMaj{f}{m}{2/3})$. By Lemma~\ref{lm:tdeg-to-uppdt}, it suffices to upper bound $\degsh(\GapMaj{f}{m}{2/3})$. By Claim~\ref{claim:const-ok}, we have $\dega_{1/20}(f) = O(\dega(f))$; let $p$ be the corresponding approximating polynomial for $f$. Now, let the input for $\GapMaj{f}{m}{2/3}$ be $x=(x_1,x_2,\dotsc,x_m)$, where each $x_i$ is the input to the $i$th copy of $f$. Let $q$ be a polynomial on $\{0,1\}^{nm}$ defined as $q(x) := \frac{1}{m} \cdot \sum_{i=1}^{m} p(x_i)^2 - 0.5$.
	
	Now, evidently $\deg(q)=2 \deg(p)$. From the definition of $\dega_{1/20}(f)$, we can see when $\GapMaj{f}{m}{2/3}(x) \newline= 1$, 
	$$
	q(x) \ge \frac{2}{3} \cdot 0.95^2 -0.5 > 0;
	$$ 
	and when $\Gapmajority(f)(x) = 0$, 
	$$
	q(x) \le \frac{2}{3} \cdot 0.05^2 + \frac{1}{3} \cdot 1.05^2 -0.5 < 0.
	$$ 
	
	Hence, by the definition of $\degsh$ (cf. Definition~\ref{maindef}), we conclude that $\degsh(\GapMaj{f}{m}{2/3}) = O(\dega(f))$, and this completes the proof for the first claim.
	
	Similarly, in order to upper bound $\UPPdt(\GapAND{f}{m}{2/3})$, it suffices to consider $\degsh(\GapAND{f}{m}{2/3})$. By Claim~\ref{claim:const-ok}, we have $\odegap_{1/20}(f) = O(\odegap(f))$; let $p$ be a corresponding positive one-sided approximating polynomial for $f$.  Let $q$ be a polynomial on $\{0,1\}^{nm}$ defined as $q(x) := \frac{1}{m} \cdot \sum_{i=1}^{m} p(x_i) - 0.5$.
	
	Clearly $\deg(q) = \deg(p)$. From the definition of $\odegap_{1/20}(f)$, we can see when $\GapAND{f}{m}{2/3}(x) = 1$, 
	$$
	q(x) \ge 0.95 -0.5 > 0;
	$$ 
	and when $\GapAND{f}{m}{2/3}(x) = 0$, 
	$$
	q(x) \le \frac{2}{3} \cdot 0.05 + \frac{1}{3} \cdot 1.05 -0.5 < 0.
	$$ 
	
	Therefore, $\degsh(\GapAND{f}{m}{2/3} = O(\odegap(f))$, and this completes the whole proof.
\end{proofof}

\section{Missing Proofs From Section~\ref{sec:composition}}
\label{app:missing-lift}

In this section we provide the missing proofs from Section~\ref{sec:composition}.
We begin with Lemma \ref{lm:helper2}, restating the lemma here for convenience.

\vspace{0.2cm}
\noindent
{\bf Lemma~\ref{lm:helper2} }
(restated)
{\em
	Let $a \ge 40$, $n$ be a sufficiently large integer and $\epsilon$ be a real such that $0.5 < \epsilon < 1$.
	Then there exists an (explicitly given) univariate polynomial $P \colon \R \to \R$ such that:
	
	\begin{itemize}
		\item $P(x) = (-a)^x$ for $x \in \{0,\dotsc,\epsilon \cdot n \}$.
		\item $|P(x)| \le a^x / 2$ for $x \in \{\epsilon \cdot n +1, \dotsc, n \}$.
		\item $P$ has degree of at most $\left(1 + \frac{10}{a}\right) \cdot \epsilon \cdot n + 3$.
	\end{itemize}
}

\begin{proofof}{Lemma~\ref{lm:helper2}}
We begin by constructing the polynomial $P \colon \R \to \R$ whose existence is claimed by the lemma. 

	\medskip
	\noindent \textbf{Construction of $P$.}
	Let $N = \left\lceil \left(1 + \frac{10}{a}\right) \cdot \epsilon \cdot n + 2 \right\rceil$. 
	We define $P$ through interpolation to be the unique polynomial of degree at most $N$ satisfying the following properties.
	
	\begin{itemize}
		\item $P(x) = (-a)^x$ for $x \in \{0,\dotsc,\epsilon \cdot n \}$.
		\item $P(x) = 0$ for $x \in \{\epsilon \cdot n + 1, \dotsc, N\}$.
	\end{itemize}
	
	\medskip
	\noindent \textbf{Analysis of $P$.} 
	Under the above definition, it is obvious that  the first and the last conditions in Lemma~\ref{lm:helper2} are satisfied by $P$. In the rest of the proof, we establish that $P$ also satisfies the second condition claimed by Lemma~\ref{lm:helper2}, i.e., 
	\begin{equation} \label{eq1} |P(x)| \le a^x / 2 \text{ for } x \in \{\epsilon \cdot n +1, \dotsc, n \}. \end{equation}
	
	When $\eps$ is a constant strictly between $1/2$ and $1$ and $a$ is a sufficiently large constant, Equation \eqref{eq1} is an easy consequence of standard bounds on the growth rate of low-degree polynomials defined through interpolation (cf. \cite[Lemma 3.1]{razborovsherstov}). However, our applications require us to consider $\eps \approx 1-1/3\log n = 1-o(1)$ and $a=\Theta(\log n)$. To handle this parameter regime, a more delicate analysis seems to be required.

	For each $i \in \{0,\dotsc,\epsilon \cdot n \}$, define the polynomial $e_i$ as
	
	\begin{equation}\label{eqe_i}
	e_i(x) := \prod_{j \in \{0,\dotsc,N\} \setminus \{i\}} \frac{x-j}{i-j}.
	\end{equation}
	
	Observe that
	\begin{equation}\label{eqob}
	\text{when $x \in \{0,\dotsc,N\}$, $e_i(x)$ is equivalent to $\indicator_{x = i}$.}
	\end{equation}
	
	Moreover, each $e_i(x)$ has degree at most $N$. Hence, we may write 
	\begin{equation} \label{eqPsum} P = \sum_{i \in \{0,\dotsc,\epsilon \cdot n\}} e_i \cdot (-a)^i.\end{equation}
	Indeed, the right hand side of Equation \eqref{eqPsum} is a polynomial of degree of at most $N$, and by Observation~\eqref{eqob}, the right hand side agrees with $P$ at all $N$ inputs in $\{0, 1, \dots, N\}$. It follows that the right hand side of Equation \eqref{eqPsum} and $P$ are equal as formal polynomials.
	
        Thus, for \emph{any} $x$, $P(x)$ can be expressed as follows.
	\begin{equation} \label{eqP}
	P(x) := \sum_{i \in \{0,\dotsc,\epsilon \cdot n\}} e_i(x) \cdot (-a)^i.
	\end{equation}

     For $x \in \{\epsilon \cdot n +1, \dotsc, N\}$, as $P(x) = 0$, Equation \eqref{eq1} trivially satisfied. So we  assume $x \in \{N+1,\dotsc, n\}$ from now on.
	Observe that for each $i \in \{0,\dotsc,\epsilon\cdot n\}$,
	\begin{equation}\label{eqob2}
	\prod_{j \in \{0,\dotsc,N\} \setminus \{i\}}  (x-j) = \prod_{j \in \{x-N,\dotsc,x\} \setminus \{x-i\}} j = \frac{x-N}{x-i} \prod_{j=x-N+1}^{x} j \le \prod_{j=x-N+1}^{x} j = x! / (x-N)!
	\end{equation}
	
	and
	\begin{equation}\label{eqob3}
	\prod_{j \in \{0,\dotsc,N\} \setminus \{i\}} |i-j| = \prod_{j=0}^{i-1} (i-j) \cdot \prod_{j=i+1}^{N} (j-i) = i! \cdot (N-i)!.
	\end{equation}
	
	Using Equation~\eqref{eqob3} and Inequality~\eqref{eqob2}, we can bound $|e_i(x)|$ by
	\begin{align}
	|e_i(x)| &\le \frac{x!/(x-N)!}{i! \cdot (N-i)!} \notag\\
	&=\frac{x!}{(x-N)! \cdot N!} \cdot \frac{N!}{i! \cdot (N-i)!}\notag\\
	&=\binom{x}{N} \cdot \binom{N}{i} \label{lastlast}.\\
	\end{align}
	
	Combining Expression \eqref{lastlast} with Equation~\eqref{eqP}, we can bound $|P(x)|$ by
	\begin{align}
	|P(x)| &\le  \sum_{i=0}^{\epsilon \cdot n} |e_i(x)| \cdot a^i \notag\\
	&\le \sum_{i=0}^{\epsilon \cdot n} \binom{x}{N} \cdot \binom{N}{i} \cdot a^i \notag\\
	&=  \binom{x}{N} \cdot  \sum_{i=0}^{\epsilon \cdot n} \binom{N}{i} \cdot a^i. \label{eqboundP}
	\end{align}
	
	Now, we are going to  bound $\sum_{i \in \{0,\dotsc,\epsilon \cdot n\}} \binom{N}{i} \cdot a^i$, as it is independent of the variable $x$. Note that for $i \in \{0,\dotsc,\epsilon \cdot n-1\}$, we have
	\begin{align*}
	&\left[\binom{N}{i+1} \cdot a^{i+1} \right]\Big/ \left[ \binom{N}{i} \cdot a^i \right]\\
	=& \frac{N-i}{i+1} \cdot a \\
	\ge& \frac{N - \epsilon \cdot n}{\epsilon \cdot n} \cdot a \tag{$i \le \epsilon \cdot n -1$}\\
	\ge& \frac{\left(1+\frac{10}{a}\right)\cdot \epsilon\cdot n - \epsilon \cdot n}{\epsilon \cdot n} \cdot a \tag{$N \ge \left(1+\frac{10}{a}\right) \cdot \epsilon \cdot n$}\\
	\ge& \frac{10}{a} \cdot a \ge 2.\\
	\end{align*}
	
	Hence,
	\begin{align}
	&\notag\sum_{i \in \{0,\dotsc,\epsilon \cdot n\}} \binom{N}{i} \cdot a^i\\
	\le&\notag \sum_{i \in \{0,\dotsc,\epsilon \cdot n\}} \binom{N}{\epsilon \cdot n} \cdot a^{\epsilon \cdot n} \cdot 2^{-\epsilon \cdot n +i}\\
	\le \label{eqCNepsn} & 2 \cdot \binom{N}{\epsilon \cdot n} \cdot a^{\epsilon \cdot n}.\\
	\end{align}
	Combining Expression \eqref{eqCNepsn} with Expression \eqref{eqboundP}, 
	it follows, in order to establish that Equation \eqref{eq1} holds, it suffices to show that
	\begin{equation}\label{eq2}
	2 \cdot \binom{x}{N}  \cdot \binom{N}{\epsilon \cdot n} \cdot a^{\epsilon \cdot n} \Big/ a^{x} \le 1/2 \text{ for } x \in \{N+1,\dotsc,n\}.
	\end{equation}

	\newcommand{\xstar}{x^{\star}}
	
	Note the left side of inequality~\eqref{eq2} is maximized if and only if the function
	\begin{equation}\label{eqf}
	f(x) := \binom{x}{N} \Big/ a^{x}
	\end{equation}
	is maximized. So now we are going to derive the value $\xstar \in \{N+1,\dotsc,n\}$ maximizing $f(\xstar)$.
	
	When $x \in \{N+1,\dotsc,n\}$, we have
	
	\begin{align}
	\notag f(x+1)/f(x)&=\left[ \binom{x+1}{N} \big/ a^{x+1} \right] \Big/ \left[\binom{x}{N} \big/ a^{x} \right]\\
	&=\frac{x+1}{a(x-N+1)} = \frac{1}{a} \cdot \left( 1 + \frac{N}{x-N+1} \right). \label{eqx}
	\end{align}
	
	By Equation~\eqref{eqx}, we can see that $f(x+1)/f(x)$ is a decreasing function in $x$. Therefore, $f(x)$ is maximized when $x$ is the smallest integer such that $
	f(x+1)/f(x) \le 1$, which is equivalent to
	
	\begin{align*}
	&1 + \frac{N}{x-N+1} \le a.\\
	\implies& \frac{N}{x-N+1} \le a-1\\
	\implies& (a-1) \cdot x \ge a\cdot N -(a-1).\\
	\end{align*}
	
	Therefore, the maximizer of $f(x)$ is 
	\begin{equation}
	\xstar =  \left\lceil \frac{a \cdot N}{a-1} - 1 \right\rceil. \label{eqxstar}
	\end{equation}
	
	Now, it suffices to verify that inequality~\eqref{eq2} holds when $x=\xstar$, i.e., 
	\begin{equation} \label{eqCxstarN}
	2 \cdot \binom{\xstar}{N}  \cdot \binom{N}{\epsilon \cdot n} \cdot a^{\epsilon \cdot n} \le a^{\xstar}/2.
	\end{equation}
	
	\noindent \textbf{Establishing Inequality \eqref{eqCxstarN}.} 
	We claim that 
	
	\begin{equation}
	\binom{\xstar}{N} \le (2e \cdot a)^{\xstar - N}. \label{eqclaim1}
	\end{equation}
	It is easy to see that $\xstar \ge N$ by Equation~\eqref{eqxstar}. When $\xstar = N$, we have $\binom{\xstar}{N} = 1 \le \left( 2 e \cdot a \right)^{\xstar - N}$. And when $\xstar > N$, we have $\xstar - N = \left\lceil \frac{N}{a-1} -1 \right\rceil \ge 1$, which in turn means $\frac{N}{a-1} > 1$.
	
	If $\frac{N}{a-1} > 2$, we have $\left\lceil \frac{N}{a-1} -1 \right\rceil \ge \frac{N}{a-1} - 1 \ge \frac{N}{2(a-1)}$. Otherwise, $\frac{N}{a-1} \in (1,2]$, and we also have $\left\lceil \frac{N}{a-1} -1 \right\rceil = 1 \ge \frac{N}{2(a-1)}$. Putting them together, we can see that when $\xstar > N$,
	\begin{equation} \label{theabove}
	\frac{\xstar}{\xstar - N} = \left\lceil \frac{a \cdot N}{a-1} - 1 \right\rceil \Big/ \left\lceil \frac{N}{a-1} - 1 \right\rceil \le \frac{a \cdot N}{a-1} \Big/ \frac{N}{2(a-1)} \le 2a.
	\end{equation}
	
	Combining Inequality \eqref{theabove} with the inequality $\binom{n}{m} \le \left( \frac{en}{m} \right)^m$, we have
	
	$$
	\binom{\xstar}{N}
	=\binom{\xstar}{\xstar-N}
	\le\left( \frac{e \cdot \xstar}{\xstar-N} \right)^{\xstar-N}
	\le\left( 2 e \cdot a \right)^{\xstar - N},
	$$
	
	when $\xstar > N$. This proves our Claim~\eqref{eqclaim1}.
	
	\medskip
	
	Now we bound $\binom{N}{\epsilon \cdot n}$. As $N \ge \left( 1 + \frac{10}{a} \right) \cdot \epsilon \cdot n$ and $a \ge 40$, we have
	
	$$
	\frac{e \cdot N}{N-\epsilon \cdot n} = e \cdot \left( 1 + \frac{\epsilon \cdot n}{N - \epsilon \cdot n} \right) \le  e \cdot \left(
	1 + \frac{a}{10} \right) \le \frac{ae}{5},
	$$
	
	and
	
	\begin{equation}\label{eqNepsn}
	\binom{N}{\epsilon \cdot n} 
	= \binom{N}{N - \epsilon \cdot n}
	\le \left( \frac{e \cdot N}{N-\epsilon \cdot n} \right)^{N-\epsilon \cdot n} 
	\le \left( \frac{ae}{5} \right)^{N-\epsilon \cdot n}.
	\end{equation}
	
	Putting Inequalities~\eqref{eqclaim1} and \eqref{eqNepsn} together, we have
	
	\begin{align}
	&2 \cdot \binom{\xstar}{N}  \cdot \binom{N}{\epsilon \cdot n} \cdot a^{\epsilon \cdot n}\\
	\le&2 \cdot \left( 2e \cdot a \right)^{\xstar - N} \cdot \left( \frac{ae}{5} \right)^{N-\epsilon \cdot n} \cdot a^{\epsilon \cdot n}\\
	\le& 2 \cdot (2e)^{\xstar - N} \cdot \left( \frac{e}{5} \right)^{N - \epsilon \cdot n} \cdot a^{\xstar}. \label{eqfinal}
	\end{align}
	
	As $\left( 1 + \frac{10}{a} \right) \cdot \epsilon \cdot n + 2 \le N = \left\lceil \left(1 + \frac{10}{a}\right) \cdot \epsilon \cdot n + 2 \right\rceil \le \left( 1 + \frac{10}{a} \right) \cdot \epsilon \cdot n + 3$ and $a \ge 40$, we have
	$$
	\xstar - N = \left\lceil \frac{N}{a-1} - 1 \right\rceil \le \frac{N}{a-1} \le \frac{1}{a-1} \cdot \left[ \left( 1 + \frac{10}{a} \right) \cdot \epsilon \cdot n + 3 \right] \le \frac{2}{a} \cdot \epsilon \cdot n + \frac{1}{10},
	$$
	
	and
	$$
	N -\epsilon \cdot n \ge \frac{10}{a} \cdot \epsilon \cdot n + 2.
	$$
	
	Therefore, we can further bound \eqref{eqfinal} by
	$$
	2 \cdot 2^{\frac{2}{a} \cdot \epsilon \cdot n + \frac{1}{10}} \cdot \left(\frac{e}{5} \right)^{\frac{10}{a} \cdot \epsilon \cdot n + 2} \cdot a^{\xstar} \le a^{\xstar}/2,
	$$

	which establishes Inequality \eqref{eqCxstarN} and completes the proof.
	
\end{proofof}

Now we provide the simple proof for Lemma~\ref{lm:mus}. We first restate it for the reader's convenience.

\vspace{0.2cm}
\noindent
{\bf Lemma~\ref{lm:mus} }
(restated)
{\em
	Let $f \colon D \to \{0,1\}$ with $D \subseteq \{0,1\}^M$ be a partial function, $\epsilon$ be a real in $[0,1/2)$, and $d$ be an integer such that $\dega_{\epsilon}(f) > d$.
	
	Let $\mu \colon \{0,1\}^M \to \R$ be a dual witness to the fact $\dega_{\epsilon}(f) > d$ as per Theorem~\ref{theo:dual-dega}. If $f$ satisfies the stronger condition that $\odegap_{\epsilon}(f) > d$, let $\mu$ to be a dual witness to the fact that $\odegap_{\epsilon}(f) > d$ as per Theorem~\ref{theo:dual-odega}.
	
	We further define $\mu_+(x) := \max\{0,\mu(x)\}$ and $\mu_-(x) := -\min\{0,\mu(x)\}$ to be two non-negative real functions on $\{0,1\}^M$, and $\mu_-^i$ and $\mu_+^i$ be the restrictions of $\mu_-$ and $\mu_+$ on $f^{-1}(i)$ respectively for $i \in \{0,1\}$. Then the following holds:
	
	\begin{itemize}
		\item $\mu_+$ and $\mu_-$ have disjoint supports.
		\item $ \langle \mu_+, p \rangle = \langle \mu_-, p \rangle$, for any polynomial $p$ of degree at most $d$. In particular, $\|\mu_+\|_1 = \|\mu_-\|_1 = \frac{1}{2}$.
		\item $\|\mu_+^1\|_1 > \epsilon$ and $\|\mu_-^0\|_1 > \epsilon$. 
		\item If $\odegap_\epsilon(f) > d$, then $\|\mu_+^1\|_1 = 1/2$.
	\end{itemize}
}

\begin{proofof}{Lemma~\ref{lm:mus}}
	The first two claims follows directly from Theorem~\ref{theo:dual-dega} and the definitions of $\mu_+$ and $\mu_-$.
	
	For the third claim, by Theorem~\ref{theo:dual-dega}, we have
	
	\begin{align*}
	\sum_{x \in D} f(x) \cdot \mu(x) - \sum_{x \notin D} |\mu(x)| &> \epsilon.\\
	\text{Hence, } \|\mu_+^1\|_1 - \|\mu_-^1\|_1 - \sum_{x \notin D} |\mu_-(x)| &> \epsilon.\\
	\text{This implies that } \|\mu_+^1\|_1 - \|\mu_-^1\|_1 - (0.5-\|\mu_-^1\|_1-\|\mu_-^0\|_1) &> \epsilon.\\
	\text{Hence, }\|\mu_+^1\|_1 - (0.5-\|\mu_-^0\|_1) &> \epsilon.\\
	\end{align*}
	
	Therefore, $\|\mu_+^1\|_1 > \epsilon$, and $(0.5-\|\mu_-^0\|_1) < 0.5 - \epsilon$, which means $\|\mu_-^0\|_1 > \epsilon$.
	
	Finally, the last claim follows directly from Theorem~\ref{theo:dual-odega}.
\end{proofof}


\newcommand{\domain}{D}

\section{Proof of Theorem \ref{THMCC}: $\NISZKcc \not\subseteq \UPPcc$}
\label{sec:prelimcc}
\label{ccappendix}
In this appendix, we define the communication complexity classes $\NISZKcc$
and $\UPPcc$ and prove that $\NISZKcc \not\subseteq \UPPcc$. 

\subsection{Preliminaries}
\subsubsection{Representation of Boolean Functions}
Up until this point of the paper, we have considered functions $f \colon \{0, 1\}^n \to \{0, 1\}$.
However, in order to define and reason about $\UPPcc$ communication complexity,
it will be highly convenient to consider functions $f \colon \{-1, 1\}^n \to \{-1, 1\}$ instead,
where $1$ is interpreted as logical \False\ and $-1$ is interpreted as logical \True. 
Hence, throughout this appendix, a total Boolean function $f$ will map $\{-1,1\}^n \to \{-1,1\}$. 
We will sometimes refer to the outputs of $f$ as $\True$ and $\False$ rather than $-1$ and $1$ for clarity.

The following additional convention will be highly convenient: we define partial Boolean functions to map undefined inputs to $0$. That is, a partial Boolean function on $\{-1,1\}^n$ will be thought as a map from $\{-1,1\}^n \to \{-1,0,1\}$. We use $\domain_f$ to denote the domain of $f$, i.e., $\domain_f := \{x \in \{-1,1\}^n: f(x) \in \{-1,1\} \}$.

\subsubsection{Entropy Estimation}
The definition of $\NISZK$ (cf. Definition~\ref{def:niszkclass}) is quite technical. Fortunately, there is a simple complete problem for $\NISZK$ named \textsf{Entropy Estimation}, which basically asks one to estimate the entropy of a distribution. Formally, it is defined as follows.

\begin{defi}[\textsf{Entropy Estimation} ($\EA$)~\cite{goldreich1999can}]\label{defi:EA}
	The promise problem \textsf{Entropy Estimation}, denoted $\EA = (\EA_{\mathsf{YES}},\EA_{\mathsf{NO}})$, is defined such that $\EA^{-1}(1)= \EA_{\mathsf{YES}}$ and 
	$\EA^{-1}(0)= \EA_{\mathsf{NO}}$, where
	
	\begin{align*}
	\EA_{\mathsf{YES}} &:= \{(X,k) : H(X) > k + 1 \}\\
	\EA_{\mathsf{NO}} &:= \{(X,k) : H(X) < k - 1.\}
	\end{align*}
	
	Here, $k$ is an integer specified as part of the input in binary and $X$ is a distribution encoded as a circuit outputting $n$ bits.
	\end{defi}

\subsubsection{Definition of Communication Models}
\medskip \noindent \textbf{Sign-Rank and $\UPPcc$.}\label{sec:sign-rank-upp-cc}
The original definition of \emph{unbounded error communication complexity} ($\UPPcc$) defined by Babai et al.~\cite{babai1986complexity} is for {\em total functions}, but it is straightforward to generalize the definition to partial functions.
Consider a partial Boolean function $f : X \times Y \to \{-1,0,1\}$. In a $\UPP$ protocol of $f$, 
Alice receives an input $x \in X \subseteq \{-1,1\}^{n_x}$, and Bob receives an input $y \in Y \subseteq \{-1,1\}^{n_y}$. Each has an unlimited source of private randomness, and their goal is to compute the joint function $f(x, y)$ of their inputs with minimal communication for all pair $(x,y)$ such that $f(x,y) \ne 0$.
We say the protocol computes $f$ if for any input $(x, y) \in f^{-1}(\{-1,1\})$, the output of the protocol is correct with probability strictly greater than $1/2$. The cost of a protocol for computing $f$ is the maximum number of bits exchanged on any input $(x, y)$. The unbounded error communication complexity $\UPP(f)$ of a function $f$ is the minimum cost of a protocol computing $f$. A partial function $f$ is in the complexity class $\UPPcc$ if $\UPP(f) = O(\log^c n)$ for some constant $c$, where $n = \max\{n_x,n_y\}$.

The \emph{sign-rank} of a matrix $A$ with entries in $\R$ is the least rank of a real matrix $B$ with $A_{ij} \cdot B_{ij} > 0$ for all $i, j$ such that $A_{i,j} \ne 0$.

Paturi and Simon \cite{paturisimon} showed that $\UPP(f) = \log (\text{sign-rank}([f(x, y)]_{x \in X, y \in Y})) + O(1)$.\footnote{Paturi and Simon's proof was in the context of total functions, but their result is easily seen to apply to partial functions as well.} Therefore, the sign-rank characterizes the $\UPPcc$ complexity of a communication problem.

\medskip \noindent \textbf{Definition of $\NISZKcc$.}
We now define the $\NISZK^{\cc}$ complexity of a Boolean function $f$. This is a natural extension of the original definition by Goldreich, Sahai and Vadhan~\cite{goldreich1999can}, and follows the canonical method of turning a complexity class into its communication analogue introduced by Babai, Frankl and Simon~\cite{babai1986complexity}.

\begin{defi}[$\NISZK^{\cc}$]
	In a $\NISZK$-protocol for a partial Boolean function $f : X \times Y \to \{-1,0,1\}$ with $X \subseteq \{-1,1\}^{n_x}$ and $Y \subseteq \{-1,1\}^{n_y}$, there are three computationally-unbounded parties Alice, Bob, and Merlin. Alice holds an input $x \in \mathcal{X}$ while Bob holds an input $y \in \mathcal{Y}$. The goal of Merlin is to convince Alice and Bob that $f(x,y) = 1$, in a non-interactive and zero knowledge fashion. 
	
	Specifically, there is a public random string shared between the three parties $\sigma \in \{0,1\}^r$. Additionally, Alice and Bob can also use shared randomness between them, which is not visible to Merlin. The protocol starts by Merlin sending a message $m=m(x,y,\sigma)$ to Alice (without loss of generality, we can assume the message is a function of $x,y$ and $\sigma$). Then Alice and Bob communicate, after which Alice outputs ``accept'' or ``reject''.  A $\NISZK$ communication protocol for $f$ also must satisfy additional conditions. The first two are the standard notions of completeness and soundness for probabilistic proof systems.
	
	\begin{itemize}
		\item Completeness: For all $(x,y) \in f^{-1}(1)$, there is a strategy $m^*$ for Merlin that causes Alice to output accept with probability $\ge 2/3$ (where the probability is taken over both the public random string $\sigma$ and the shared randomness between Alice and Bob that is not visible to Merlin).
		\item Soundness: For all $(x,y) \in f^{-1}(0)$, and for every strategy for Merlin, Alice outputs accept with probability $\le 1/3$.
		\end{itemize}
		Let the worst case communication cost be $w_{V}$, where this cost includes both the length of Merlin's message $m^*(x, y, \sigma)$ and the total number of bits exchanged by Alice and Bob.
		Finally, a $\NISZK$ communication protocol must also satisfy the following zero knowledge condition
		\begin{itemize}
		\item Zero Knowledge: There is a public-coin randomized communication protocol $S$ with output in $\{0,1\}^k$ and worst case communication complexity $w_{S}$, such that for all $(x,y) \in F^{-1}(1)$, the statistical distance between the following two distributions is smaller than $1/n$, where $n = \max(n_x,n_y)$.
		
		\begin{itemize}
			\item (A) Choose $\sigma$ uniformly from $\{0,1\}^r$, and output $m^*(x,y,\sigma)$.
			\item (B) The output distribution of $S$ on $(x,y)$.
		\end{itemize}
	\end{itemize}
	
	Finally, the cost of the $\NISZK$ communication protocol is defined as $r + \max(w_V,w_S)$. The quantity $\NISZK^{\cc}(f)$ is defined as the minimum of the cost of all $\NISZK$ protocols for $f$.
\end{defi}

\begin{rem}
	In the original definition in~\cite{goldreich1999can} (see also Definition~\ref{def:niszkclass} in Section \ref{sec:niszkdef} of this work), it is required that the statistical difference between the distributions (A) and (B) is negligible. In the context of communication complexity, where protocols of cost $\polylog(n)$ are considered efficient, negligible corresponds 
	to $1/\log^{\omega(1)}(n)$. However, since for polylogarthmic cost protocols, the difference between the distributions (A) and (B) can be amplified from $1/\polylog(n)$ to $1/n^{\omega(1)}$ with a polylogarithmic blowup in cost (cf. Lemma~3.1 in \cite{goldreich1999can}), we simply require the difference to be at most $1/n$ here.
\end{rem}

\subsubsection{Approximate Degree, Threshold Degree, and Their Dual Characterizations}
Since we are now considering Boolean functions mapping $\{-1, 1\}$ to $\{-1,1\}$ rather than $\{0, 1\}$, it is convenient to redefine approximate degree and threshold degree in this new setting and state the appropriate dual formulations. 

\begin{defi}\label{defi:approx-deg-new}
	Let $f: \{-1,1\}^M \to \{-1,0,1\}$ be a partial function. Recall that the domain of $f$ is $\domain_f := \{x \in \{-1,1\}^M: f(x) \in \{-1,1\} \}$.
	\begin{itemize}
	\item The \emph{approximate degree} of $f$ with approximation constant $0\le\epsilon<1$, denoted $\dega_\epsilon(f)$, is the least degree of a real polynomial $p \colon \{-1, 1\}^M \to \R$ such that $|p(x)-f(x)| \le \epsilon$ when $x \in \domain_f$, and $|p(x)| \le 1+\epsilon$ for all $x\not\in \domain_f$. We refer to such a $p$ as an \emph{approximating polynomial} for $f$. We use $\dega(f)$ to denote $\dega_{1/3}(f)$.
	
	\item The \emph{threshold degree} of $f$, denoted $\degsh(f)$, is the least degree of a real polynomial $p$ such that $p(x) \cdot f(x) > 0$ for all $x \in \domain_f$.
    \end{itemize}
\end{defi}

\begin{rem}\label{rem:change}
All the results from earlier in this paper regarding partial functions mapping (subsets of) $\{0, 1\}^n$ to $\{0, 1\}$
can be translated to results regarding functions mapping $\{-1, 1\}^n$ to $\{-1, 0, -1\}$ as considered in this appendix. Specifically, given a partial function $f : \{-1,1\}^M \to \{-1,0,1\}$ with domain $\domain_f$, 
let $D_f^{\{0, 1\}} := \left\{\left(\frac{1-x_1}{2},\dotsc,\frac{1-x_M}{2}\right) \colon (x_1, \dots, x_M) \in D_f\right\}$.
Consider the partial function $f^{\{0,1\}} \colon D_f^{\{0,1\}} \to \{0, 1\}$ defined as 
$$f^{\{0,1\}}(x_1,\dotsc,x_M) = \frac{1-f((-1)^{x_1},\dotsc,(-1)^{x_M})}{2}.
$$
Then it is easy to see that $\dega_{\epsilon}(f)$ is equal to $\dega_{\epsilon/2}(f^{\{0,1\}})$, for any $\epsilon \in [0,1)$.
\end{rem}
We recall the definitions of norm, inner product and high pure degree, now with respect to Boolean representation in $\{-1,1\}$.
For a function $\psi \colon \{-1, 1\}^M \to \R$, define the $\ell_1$ norm of $\psi$ by $\|\psi\|_1 = \sum_{x \in \{-1,1\}^M} |\psi(x)|$. 
If the support of a function $\psi \colon \{-1,1\}^M \to \R$ is (a subset of) a set $D \subseteq \{-1, 1\}^M$, we will write $\psi \colon D \to \R$. 	
For
functions $f, \psi \colon D \to \R$, denote their inner product by $\langle f, \psi \rangle := \sum_{x \in D} f(x) \psi(x)$. 
We say that a function $\psi \colon \{-1,1\}^M \to \R$ has \emph{pure high degree $d$}
if $\psi$ is uncorrelated with any polynomial $p \colon \{-1,1\}^M \to \R$ of total degree at most $d$, i.e., if
$\langle \psi, p \rangle = 0$. 

\begin{theo}\label{theo:dual-dega-new}
	
	Let $f:\{-1,1\}^M \to \{-1,0,1\}$ be a partial function and $\epsilon$ be a constant in $[0,1)$. $\dega_\epsilon(f) > d$ if and only if there is a real function $\psi: \{-1,1\}^M \to \R$ such that:
	\begin{enumerate}
		\item (Pure high degree): $\psi$ has pure high degree of $d$.
		\item (Unit $\ell_1$-norm): $\|\psi\|_{1} =1$. 
		\item (Correlation): $\sum_{x \in D_f} \psi(x) f(x) - \sum_{x\not\in D_f} |\psi(x)| > \epsilon$.
	\end{enumerate}
\end{theo}

\begin{theo}\label{theo:dual-degsh-new}
	
	Let $f:\{-1,1\}^M \to \{-1,0,1\}$ be a partial function. $\degsh(f) > d$ if and only if there is a real function $\psi: \{-1,1\}^M \to \R$ such that:
	\begin{enumerate}
		\item (Zero Outside of Domain): $\psi(x) = 0$ when $x \notin D_f$.
		\item (Pure high degree): $\psi$ has pure high degree of $d$.
		\item (Sign Agreement): $\psi(x) \cdot f(x) \ge 0$ for all $x \in D_f$.
		\item (Non-triviality): $\|\psi\|_1 > 0$.
	\end{enumerate}
\end{theo}

\medskip \noindent \textbf{Orthogonalizing Distributions.}
\label{sec:ortho}
If $\psi$ is a dual witness for the fact that $\degsh(f) > d$ as per Theorem \ref{theo:dual-degsh-new},
then $\psi \cdot f$ is a \emph{$d$-orthogonalizing distribution} for $f$, as defined next. 

\begin{defi}\label{defi:ortho}
	A distribution $\mu : \{-1,1\}^n \to [0, 1]$ is \emph{$d$-orthogonalizing} for a function $h : \{-1,1\}^n \to \{-1, 0, 1\}$ if
	\[\Ex_{x \sim \mu} [h(x) p(x)] = 0\]
	for every polynomial $p: \{-1,1\}^n \to \mathbb{R}$ with $\deg p \le d$. In other words, $\mu$ is $d$-orthogonalizing for $h$ if the function $\mu(x)h(x)$ has pure high degree $d$.
\end{defi}

\subsubsection{Pattern Matrices}\label{sec:pmat}
As indicated in Section \ref{sec:cc}, Razborov and Sherstov \cite{razborovsherstov}
showed that in order to turn a function $f \colon \{-1, 1\}^n \to \{-1, 0, 1\}$ that has high threshold degree into a matrix $M$ with high sign-rank (and hence high $\UPPcc$ complexity), it suffices
to show exhibit a dual witness $\psi$ to the fact that $\degsh(f)$ is large, such that $\psi$ satisfies an additional \emph{smoothness} condition. The transformation from $f$ to the matrix $M$ relies on the {\em pattern matrix} method introduced by Sherstov~\cite{sherstov2011pattern}. Pattern matrices are defined as follows.

Let $n$ and $N$ be positive integers for which $n$ divides $N$. Let $\mathcal{P}(N, n)$ denote the collection of subsets $S \subset [N]$ for which $S$ contains exactly one member of each block $\{1, 2, \dots, N/n\}, \{N/n + 1, \dots, 2N/n\}, \dots,\newline \{(n-1)N/n + 1, \dots, N\}$. For $x \in \{-1,1\}^N$ and $S \in \mathcal{P}(N, n)$, let $x|_S$ denote the restriction of $x$ to $S$, i.e., $x|_S = (x_{s_1}, \dots, x_{s_n})$ where $s_1 < \dots < s_n$ are the elements of $S$.

\begin{defi}\label{defi:pmat}
	For $\phi: \{-1,1\}^n \to \mathbb{R}$, the $(N, n, \phi)$-pattern matrix $M$ is given by
	\[M = [\phi(x |_S \oplus w)]_{x \in \{-1,1\}^N, (S, w) \in \mathcal{P}(N, n) \times \{-1,1\}^n}.\]
	Note that $M$ is a matrix with $2^N$ rows and $(N/n)^n2^n$ columns.
\end{defi}

When $\phi$ is partial function $\{-1,1\}^n \to \{-1,0,1\}$, the $(N, n, \phi)$-pattern matrix $M$ can be viewed as a communication problem as follows: Alice and Bob aim to evaluate the $\phi(u)$ for some ``hidden'' input $u$ to $\phi$. Alice gets a sequence of bits $x \in \{-1,1\}^N$, while Bob gets a set of coordinates $S = \{s_1,\dotsc,s_n\}$ and a shift $w \in \{-1,1\}^n$. The hidden input $u$ is simply $x |_S \oplus w$. 

Note that $M_{x,y}$ is defined (i.e., $M_{x,y} \ne 0$) if and only if $\phi(u)$ is defined ($\phi(u) \ne 0$). This is the reason we represent a partial function by a function $\{-1,1\}^n \to \{-1,0,1\}$.

\subsubsection{Symmetrization} 
We introduce the notion of symmetrization in this subsection, which is one of the key technical ingredients in our proof.

\begin{defi}
	Let $T : \{-1,1\}^k \to D$, where $D$ is a finite subset of $\R^n$ for some $n \in \mathbb{N}$. The
	map $T$ is {\em degree non-increasing} if for every polynomial $p : \{-1,1\}^k \to \R$, there exists a polynomial
	$q : D \to \R$ with $\mathrm{deg}(q) \le \mathrm{deq}(p)$ such that
	$$
	q(T(x)) = \mathop{\Ex}\limits_{y\ s.t.\ T(y)=T(x)} [p(y)]
	$$
	for every $x \in \{-1,1\}^k$. We say that a degree non-increasing map T {\em symmetrizes} a function
	$f : \{-1,1\}^k \to \R$ if $f(x) = f(y)$ whenever $T(x) = T(y)$, and in this case we say that $T$ is a
	symmetrization for $f$.
\end{defi}

For any function $\psi:\{-1,1\}^k \to \R$, a symmetrization $T: \{-1,1\}^k \to D$ for $\psi$ induces a symmetrization function $\tilde{\psi}: D \to \R$ defined as $\tilde{\psi}(z) := \Ex_{x \in T^{-1}(z)} [\psi(x)]$ (if $T^{-1}(z)$ is empty, we let $\tilde{\psi}(z)$ to be $0$). It will also be convenient to define an ``unnormalized'' version $\hatpsi$ of $\sympsi$, defined via $\hatpsi(z) := \sum_{x \in T^{-1}(z)} \psi(x)$. Observe that if $\mu$ is a distribution on $\{-1, 1\}^k$, then $\hatmu$ is a distribution
on $D$.

Let $T : \{-1,1\}^k \to D$ be a degree non-increasing map. A function $\hat\psi : D \to \R$
naturally induces an un-symmetrized function $\psi \colon \{-1,1\}^k \to \R$ by setting  $\psi(x) = \frac{1}{T^{-1}(z)} \hat{\psi}(z)$
where $z = T(x)$. That is, $\psi$ spreads the mass of $\hat\psi (z)$ out evenly over points $x \in T^{-1}(z)$. Observe
that, for any $\hat{\psi}$ and any degree non-increasing map $T$, the induced function $\psi$ is symmetrized by
$T$.

 
\label{sec:lift-upp-cc}

\subsubsection{Dual Objects}
A key technical ingredient in Bun and Thaler's \cite{bt16} methodology for proving sign-rank lower bounds
is the notion of a \emph{dual object} for a Boolean function $f$, which is roughly a dual witness $\psi$ for the high one-sided approximate degree of $f$, that satisfies additional metric properties. We introduce a related definition below. The difference between our definition of a dual object and Bun and Thaler's is that our definition only requires $\psi$ to witness the high approximate degree (rather than one-sided approximate degree) of $f$.   

\begin{defi}[Dual Object]\label{defi:dual-obj}
	Let $f: \{-1,1\}^k \to \{-1,0,1\}$ be a partial function, and let $T : \{-1,1\}^k \to D$ be a degree non-increasing symmetrization for $f$. Let $\hat\psi : D \to \R$ be any function, and let $\psi$ be the associated function on $\{-1,1\}^k$ induced by $T$. We say that $\hat\psi$ is a $(d,\epsilon,\eta)$-dual object for $f$ (with respect to $T$) if:
	
	\fourpropsname{dualobj}{$\sum\nolimits_{x \in \domain_f} \psi(x) f(x) - \sum\nolimits_{x \in \{-1,1\}^k \setminus \domain_f} |\psi(x)| > \epsilon$}{$\|\psi\|_1 = 1$}{$\psi$ has pure high degree at least $d$}{$\hat\psi(z_+) \ge \eta$ for some $z_+$ satisfying $\hat{f}(z_+) = 1$}
\end{defi}

\begin{rem}\label{rem:dualobj}
The first three Conditions~\eqref{dualobjprop1}, \eqref{dualobjprop2} and \eqref{dualobjprop3} together are equivalent to requiring that $\psi$ is a dual witness for $\dega_{\epsilon}(f) > d$ as per Theorem~\ref{theo:dual-dega-new}. Condition~\eqref{dualobjprop4} is an additional metric property that is crucial for the construction of a \emph{smooth} orthogonalizing distribution of $\Gapmajority(f)$.
\end{rem}

\subsection{The Key Technical Ingredient}
We now define a class of functions $\cd{d}{a}$, and establish that for appropriate values of $d$ and $a$, for
any function $f \in \cd{d}{a}$, it holds that the pattern matrix of $\Gapmajority(f)$ has large sign-rank.

\begin{defi}\label{defi:new-Cd} \label{defcda}
	Let $f: \{-1,1\}^k \to \{-1,0,1\}$ be a partial function and let $d, a > 0$. Then $f$ is in class $\cd{d}{a}$ if there exists a symmetrization $T : \{-1,1\}^k \to D$ for $f$ such that:
	
	\twopropsname{newcd}{There exists a $(d,\epsilon,\epsilon/2)$-dual object for $f$ with respect to $T$, such that $\frac{\epsilon}{1-\epsilon} > a$.}{$f$ evaluates to $\False$ (i.e. $f(x) = 1$) for all but at most a $2^{-d}$ fraction of inputs in $\{-1,1\}^k$.}
\end{defi}

Now we are ready to state the key technical claim that we use (cf. Section \ref{sec:niszkcc-uppcc} below) to
separate $\NISZKcc$ and $\UPPcc$.

\begin{theo}\label{theo:lifting-upp-cc}
	Let $\epsilon \in (0,0.1)$, consider a partial function $f : \{-1,1\}^k \to \{-1,0,1\} \in \cd{m}{40/\epsilon}$ such that $\epsilon \cdot m > 50$. Let $F := \GapMaj{f}{m}{1 - \epsilon}$ and $n := m\cdot k$. Then the $(2^{36 + 6 \log\epsilon^{-1}} \cdot n,n,F)$-pattern matrix $M$ has sign-rank of $\exp(\Omega(\epsilon \cdot m))$. 
\end{theo}

\subsubsection{Proof for Theorem~\ref{theo:lifting-upp-cc}}

We need the following theorem for lower bounding sign-rank, which is implicit in  \cite[Theorem~1.1]{razborovsherstov}. 

\begin{theo}[{Implicit in \cite[Theorem~1.1]{razborovsherstov}}]\label{theo:rs-sign-rank-lowb}
	Let $h \colon \{-1,1\}^n \to \{-1,0,1\}$ be a Boolean function and $\alpha > 1$ be a real number. Suppose there exists a $d$-orthogonalizing distribution $\mu$ for $h$ such that $\mu(x) \ge 2^{-\alpha d} 2^{-n}$ for all but a $2^{-\Omega(d)}$ fraction of inputs $x \in \{-1,1\}^n$. Then the $(2^{3\alpha}n,n,h)$-pattern matrix $M$ has sign rank $\exp(\Omega(d))$.
\end{theo}

\newcommand{\zeroset}{Z}

For a partial function $f:\{-1,1\}^n \to \{-1,0,1\}$ and its gapped majority version $h_m := \GapMaj{f}{m}{1-\epsilon}$, we use $\zeroset(h_m)$ to denote the set 
$$
\{
x=(x_1,x_2,\dotsc,x_m) \in \{-1,1\}^{nm} : f(x_i) = 1 \text{ for all $x_i$}
\}.
$$
That is, $\zeroset(h_m)$ is the set of inputs to $h_m$ such that all copies of $f$ evaluates to \False.

The following theorem asserts the existence of the \dort{d} distribution for $\Gapmajority(f)$ that is needed to apply Theorem \ref{theo:rs-sign-rank-lowb}.

\begin{theo}\label{theo:smooth-ort}

Let $\epsilon \in (0,0.1)$, $m$ be an integer such that $\epsilon \cdot m > 50$, and $f$ be partial function $f \colon \{-1,1\}^k \to \{-1,0,1\}$ with a $(d_1,\epsilon_2,\eta)$-dual object (with respect to symmetrization $T: \{-1,1\}^k \to D$) such that $\frac{\epsilon_2}{1-\epsilon_2} > 40/\epsilon$. 
Let $h_m := \GapMaj{f}{m}{1 - \epsilon}$ and $d = \min\{d_1,\epsilon/4 \cdot m\}$. Then there exists a \dort{d} distribution $\mu : \{-1,1\}^{mk} \to [0,1]$ for $h_m$ such that 
$$
\mu(x) \ge 2^{-2d} \cdot \binom{m}{d}^{-2} \cdot (2\eta)^{m} 2^{-mk}
$$ 
for all $x \in \zeroset(h_m)$.
\end{theo}

Before proving Theorem~\ref{theo:smooth-ort}, we show that combining it with  Theorem \ref{theo:rs-sign-rank-lowb} implies Theorem \ref{theo:lifting-upp-cc}.

\begin{proofof}{Theorem~\ref{theo:lifting-upp-cc}}

By Condition~\eqref{newcdprop1} of $f \in \cd{m}{40/\epsilon}$, $f$ has a $(m,\epsilon_2,\epsilon_2/2)$-dual object with respect to $T$ such that $\frac{\epsilon_2}{1-\epsilon_2} > 40/\epsilon$. Applying Theorem~\ref{theo:smooth-ort}, for $d = \min\{m,\epsilon/4 \cdot m\} = \epsilon/4 \cdot m$, there exists a $d$-orthogonalizing distribution $\mu :\{-1,1\}^{mk} \to [0,1]$ for $F$ such that $\mu(x) > 2^{-2d} \cdot \binom{m}{d}^{-2} \cdot \epsilon_2^{m} 2^{-mk}$ for all $x \in \zeroset(F)$. 

By the inequality $\binom{a}{b} \le \left( \frac{e\cdot a}{b} \right)^b$, we have 

\begin{equation}\label{eq:Cmd}
\binom{m}{d} \le (4e/\epsilon)^d \le 2^{(4 + \log \epsilon^{-1})d}.
\end{equation}

Since $\frac{\epsilon_2}{1-\epsilon_2} > 40/\epsilon$, we have $1-\epsilon_2 < \frac{\epsilon_2}{40/\epsilon} \le \frac{\epsilon}{40}$, hence $\epsilon_2 > 1 - \epsilon/40$. Therefore, 
\begin{equation}\label{eq:eps2m}
\epsilon_2^m \ge (1-\epsilon/40)^m \ge 4^{-\epsilon m/40} \ge 2^{-d/5},
\end{equation}
where the second inequality holds by the inequality $(1-x)^{1/x} \ge 1/4$ for all $x \in (0,0.5)$.

Putting Inequalities~\eqref{eq:Cmd} and~\eqref{eq:eps2m} together, we have for all $x \in \zeroset(F)$, 
\begin{equation}\label{eq:mux}
\mu(x) > 2^{-mk} \cdot 2^{-(2 + 2(4+\log \epsilon^{-1}) + 1/5) d} > 2^{-mk} \cdot 2^{-(12 + 2 \log \epsilon^{-1}) d}.
\end{equation}

And by Condition~\eqref{newcdprop2} of $f \in \cd{m}{40/\epsilon}$, $f(x) = 1$ for all but at most a $2^{-m}$ fraction of inputs in $\{-1,1\}^k$. Hence, by a union bound, there is at most a $m \cdot 2^{-m} \le 2^{-m/2} \le 2^{-d}$ fraction of inputs do not belong to $\zeroset(\GapMaj{f}{d}{1 - \epsilon}) = \zeroset(F)$. 

By the above fact and Inequality~\eqref{eq:mux}, we conclude that $\mu$ is a $d$-orthogonalizing for $F$ such that $\mu(x) \ge 2^{-mk} \cdot 2^{-(12 + 2 \log \epsilon^{-1}) d}$ for all but a $2^{-d}$ fraction of inputs in $\{-1,1\}^{mk}$. Therefore, invoking  Theorem~\ref{theo:rs-sign-rank-lowb}, the $(2^{36 + 6 \log\epsilon^{-1}} \cdot n, n, F)$-pattern matrix $M$ has sign-rank $\exp(\Omega(d)) = \exp(\Omega(\epsilon \cdot m))$.
\end{proofof}

\subsubsection{Proof of Theorem \ref{theo:smooth-ort}}
\label{sec:smooth}

\noindent \textbf{Additional Notation.} Let $f : \{-1,1\}^k \to \{-1,0,1\}$ be as in the statement of Theorem \ref{theo:smooth-ort}, and let $T : \{-1,1\}^k \to D$ be the symmetrization for $f$ associated with the assumed $(d_1, \epsilon_2, \eta)$-dual object for $f$. Define $T^m \colon \{-1,1\}^{mk} \rightarrow D^m$ by $T^m(x_1, \dots, x_m) := (T(x_1), \dots, T(x_m))$. Since $T$ is degree non-increasing, it is easy to see that $T^m$ is also degree non-increasing. Moreover, $T^m$ is a symmetrization for $h_m$. The map $T^m$
induces a symmetrized version $\symh_{m} \colon D^M \rightarrow \mathbb{R}$ of $h_m$ given by $\symh_{m} = \GapMaj{\symf}{m}{1-\epsilon}$.

Throughout the proof, we let $c \in \symf^{-1}(1)$ denote the point on which the dual object $\hatpsi$ for $f$ has $\hatpsi(c) \ge \eta$ (cf. Condition (\ref{dualobjprop4}) within Definition \ref{defi:dual-obj}).

\medskip \noindent \textbf{Proof Outline.}
Our proof follows roughly the same steps as in~\cite{bt16}. Let $Z^+ := T^m(\zeroset(h_m)) \subseteq D^m$. At a high level, our proof will produce, for every $z \in Z^+$, a $d$-orthogonalizing distribution $\mu_z$ that is targeted to $z$, in the sense that 
$$
\hatmu_z(z) \ge 2^{-O(d)} \cdot \binom{m}{d}^{-2} \cdot (2\eta)^{m}.
$$ Since the property of $d$-orthogonalization is preserved under averaging, we construct the final distribution by a convex combination of these constructed distributions $\mu_z$'s so that it places the required amount of probability mass on each input $x \in (T^m)^{-1}(Z^+) = \zeroset(h_m)$. The goal therefore becomes to construct these targeted distributions $\mu_z$. We do this in two stages.

\medskip
\noindent \textbf{Stage 1.}
In the first stage (see Claim \ref{claim:close} below), we construct distributions $\mu_z$ for every $z$ belonging to a highly structured subset $G \subset Z^+$ that we now describe. The set $G$ consists of inputs in $Z^+$ for which $c$ is repeated many times (specifically, at least $(1-\epsilon/4) \cdot m$ times). 

\medskip
\noindent \textbf{Stage 2.}
In the second stage (see Claim \ref{claim:far} below), we show that given the family of distributions $\{\mu_z : z \in G\}$ constructed in Stage 1, we can construct appropriate distributions $\mu_z$ for $z$ belonging to the entire set $Z^+$.

\medskip
We begin Stage 1 with a lemma.
\begin{lemma}\label{lm:tdeg} \label{f16}
	Let $\ell = (1-\epsilon/4) \cdot m$, and let $f$, $T$, and $T^{\ell}$ be as above. Consider the partial function $g_{\ell} : \{-1,1\}^{k\ell} \to \{-1,0,1\}$ defined as $g_{\ell} := \GapMaj{f}{\ell}{1-2\epsilon/3}$.
	There exists a function $\psi: \{-1,1\}^{k\ell} \to [0, 1]$ symmetrized by $T^{\ell}$ with the following properties.
	
	\threepropsname{lmtdeg}{$\psi$ is a dual witness for $\degsh(g_\ell) > d$ as per Theorem~\ref{theo:dual-degsh-new}, where $d = \min\{\epsilon/4 \cdot m, d_1\}$.}{$\|\psi\|_1 = 1$.}{$\hatpsi(\underbrace{c, \dots, c}_{\ell \text{ times}}) \ge (2\eta)^{\ell} / 6$.}
	
\end{lemma}
\begin{proof}
	The proof of this lemma is just analyzing the dual witness constructed for $\GapMaj{f}{\ell}{1-2\epsilon/3}$ in Theorem~\ref{theo:lift-UPP}. The analysis is as follows.
	
	\newcommand{\psiold}{\psi_{\textsf{old}}}
	\newcommand{\hatpsiold}{\hatpsi_{\textsf{old}}}
	
	\noindent\textbf{Properties of the dual witness constructed in Theorem~\ref{theo:lift-UPP}.} We formalize some properties of the dual witness constructed by Theorem~\ref{theo:lift-UPP} below. The original theorem deals with partial functions with signature $\{0,1\}^M \to \{0, 1\}$, but with the transformation described in Remark~\ref{rem:change}, it is straightforward to obtain from it the following result for partial functions with signature $\{-1,1\}^M \to \{-1, 0, 1\}$.
	\begin{prop}[Implicit in Theorem~\ref{theo:lift-UPP}]\label{prop:dual-prop}
	Let $f : \{-1,1\}^M \to \{-1,0,1\}$ be a partial function, $n$ be a sufficiently large integer, $d$ be an integer, $\epsilon \in (0.5,1)$ and $\epsilon_2 \in (0.98,1)$ be two constants. Let $a = \frac{\epsilon_2}{1-\epsilon_2}$, $N = \min\left(d,\left(1 - \left(1+\frac{10}{a} \right) \cdot \epsilon \right) \cdot n - 4\right)$ and $F := \GapMaj{f}{n}{\epsilon}$.
	
	Suppose $\dega_{\epsilon_2}(f) > d$, and let $\mu$ be a dual witness to this fact as per Theorem~\ref{theo:dual-dega-new}. Define $\mu_+(x) := \max\{0,\mu(x)\}$ and $\mu_-(x) := -\min\{0,\mu(x)\}$ to be two non-negative real functions on $\{-1,1\}^M$ (analogous to Lemma~\ref{lm:mus}). 
	
	Then there exists a function $\psiold : \{-1,1\}^{n \cdot M} \to \R$ such that
	
	\renewcommand{\vmid}{\vspace{-7mm}}
	\fourpropsname{psidual}{$\psiold$ takes non-zero values only on the domain of $F$.}{$\psiold(x) \in \left[\frac{1}{2} \cdot \prod\nolimits_{i} \mu_+(x_i), \frac{3}{2} \cdot \prod\nolimits_{i} \mu_+(x_i) \right]$ when $F(x) = \False$.}{$-\psiold(x) \in \left[\frac{1}{2} \cdot \prod\nolimits_{i} \mu_-(x_i), \frac{3}{2} \cdot \prod\nolimits_{i} \mu_-(x_i) \right]$ when $F(x) = \True$.}{$\psiold$ has pure high degree $N$.}
	\renewcommand{\vmid}{\vspace{-8mm}}
	\end{prop}
	
	Now, let $\hatvarphi$ be the $(d_1,\epsilon_2,\eta)$-dual object with respect to $T$, and $\varphi$ be the associated function from $\{-1,1\}^k \to \R$. By Remark~\ref{rem:dualobj}, $\varphi$ is a dual witness for $\dega_{\epsilon_2}(f) > d_1$. And by assumption, $\frac{\epsilon_2}{1-\epsilon_2} > 40/\epsilon > 400$, which means $\epsilon_2 > 0.98$. So we can invoke Proposition~\ref{prop:dual-prop} to construct the function $\psiold$ for $F := g_{\ell} = \GapMaj{f}{\ell}{1-2\epsilon/3}$.
	
	\medskip
	\noindent\textbf{Verification of Conditions~\eqref{lmtdegprop1}-\eqref{lmtdegprop3}.}
	We simply set $\psi := \psiold / \|\psiold\|_1$. Conditions \eqref{lmtdegprop2} follows immediately from the definition of $\psi$. 
	
	To check Condition~\eqref{lmtdegprop1}, note that $\psiold$ has pure degree
	\begin{align*}
	&\min\left(d_1, \left(1 - \left(1+\frac{10}{40/\epsilon} \right) \cdot (1-2\epsilon/3) \right) \cdot \ell - 4\right)\\
 \ge&\min\left(d_1, \left(1 - 1 -\frac{\epsilon}{4} + \frac{2\epsilon}{3} + \frac{\epsilon^2}{6} \right) \cdot (1-\epsilon/4) \cdot m - 4\right)\\
 \ge&\min\left(d_1, \epsilon/4 \cdot m\right). \tag{$\epsilon \cdot m \ge 50$ and $\epsilon < 0.1$}
	\end{align*}
	
	Together with Properties \eqref{psidualprop1}, \eqref{psidualprop2} and \eqref{psidualprop3} (recall that -1 represents $\True$ and $1$ represents $\False$), this implies that $\psi$ satisfies Condition~\eqref{lmtdegprop1}.
	
	Finally, we verify Condition~\eqref{lmtdegprop3}. Let $z_c = (\underbrace{c, \dots, c}_{\ell \text{ times}})$.  Since $\symf(c) = \False$, we have $\symg_\ell(z_c) = \False$ and therefore $F(x) = \False$ for $x$ such that $T^{\ell}(x) = z_c$. So we have
	
	\begin{align*}
	\hatpsiold(z_c) &= \sum_{x \in (T^{\ell})^{-1}(z_c)} \psiold(x)\\
	                &\ge \frac{1}{2} \cdot \sum_{x \in (T^{\ell})^{-1}(z_c)} \prod_{i=1}^{\ell} \mu_+(x_i) \tag{Condition~\eqref{psidualprop2}}\\
	                &= \frac{1}{2} \cdot \prod_{i=1}^{\ell} \sum_{x \in T^{-1}(c)} \mu_+(x)\\
	                &\ge \frac{1}{2} \cdot \prod_{i=1}^{\ell} \hatvarphi(c) \ge \frac{1}{2} \cdot \eta^{\ell},\\
	\end{align*}
	
	where the second last inequality holds since $\hatvarphi(c) = \sum_{x \in T^{-1}(c)} \varphi(x)$ and $\varphi(x) \le \mu_{+}(x)$ as $\mu_+(x) := \max\{0,\mu(x)\}$, and the last inequality is due to Condition~\eqref{dualobjprop4} from the definition of a $(d_1,\epsilon_2,\eta)$-dual object.
	
	Also, by Properties~\eqref{psidualprop2} and \eqref{psidualprop3}, combined with the fact that $\|\mu_+\|_1 = \|\mu_-\|_1 = \frac{1}{2}$, we have 
	$$
	\|\psiold\|_1 \le \frac{3}{2} \cdot \left(\|\mu_+ \|_1^{\ell} + \|\mu_+ \|_1^{\ell}\right) = \frac{3}{2} \cdot 2^{-\ell} \cdot 2 = 3 \cdot 2^{-\ell}.
	$$
	
	Putting them together, we have $$\hatpsi(z_c) \ge \frac{\eta^{\ell}/2}{\|\psiold\|_1} = (2\eta)^{\ell}/6.$$ 
	This establishes Condition~\eqref{lmtdegprop3} and completes the proof of Lemma \ref{f16}.
\end{proof}

With the above lemma, we are now ready to complete Stage 1 by showing that for every input $w \in D^m$ that is close in Hamming distance to the special point $z_c = (\underbrace{c, \dots, c}_{m \text{ times}})$, there is an orthogonalizing distribution for $h_{m}$ that places substantial weight on $w$.
 
Let $G$ denote the set of inputs in $Z^+$ that take the value $c$ on at least $(1-\epsilon/4) \cdot m$ coordinates. That is, 
$$
G = \{ z \in Z^+ : \sum_{i=1}^{m} \indicator_{z_i=c} \ge (1-\epsilon/4) \cdot m \}.
$$

\begin{claim}\label{claim:close}
	Let $G$ be as above. For every $w=(w_1, \dots, w_m) \in G$, there exists a $d$-orthogonalizing distribution $\nu_w: \{-1,1\}^{km} \rightarrow [0, 1]$ for $h_{m}$ such that $\nu_w$ is symmetrized by $T^m$ and $\hatnu_w(w) \ge (2\eta)^{m}/6$.
\end{claim}
\begin{proof}
Let $\ell = (1-\epsilon/4)\cdot m$, $I = \{i_1, \dots, i_{\ell}\}$ denote the first $\ell$ coordinates on which $w$ takes the value $c$.

Then we define the distribution $\hatnu_w$ by 
	\[
	\hatnu_w(z) = 
	\begin{cases}
	|\hatpsi(z_{i_1}, \dots, z_{i_\ell})| & \text{ if } z_i = w_i \text{ for all } i \notin I\\
	0 & \text{ otherwise}
	\end{cases}
	\]
where $\hatpsi$ is the function from Lemma \ref{lm:tdeg} for $g_\ell$. It is immediate from the definition that $\hatnu_w$ is a distribution on $D^m$, and hence $\nu_w$ is a distribution on $\{-1, 1\}^{km}$. Moreover, $\hatnu_w(w) \ge (2\eta)^{\ell} / 6 \ge (2\eta)^{m} / 6$.

To show that $\nu_w$ is $d$-orthogonalizing, let $p_1, \dots, p_m$ be polynomials over $\{-1,1\}^k$ whose degrees sum to at most $d$.
Let $\symp_1, \dots, \symp_m \colon D \rightarrow \mathbb{R}$ denote polynomials satisfying the property
that for all $i$ and all $z$ in the image of $T$, $\symp_i(z) := \Ex_{x \in T^{-1}(z)} [p_i(z)]$.
(Since $T$ is degree non-increasing, there exist such $\symp_i$'s whose degrees sum to at most $d$,  but 
we will not make use of this property in this proof).

Observe that:
\begin{align*} \!\!\!\!\!\!\!\!\!\!\!\!\sum_{x=(x_1, \dots, x_m) \in \{-1,1\}^{km}}\!\!\!\!\!\! \nu_w(x) h_{m}(x) \prod_{i=1}^m p_i(x_i)  & =\!\!
\sum_{z=(z_1, \dots, z_m) \in D^m} \hatnu_w(z) \symh_{m}(z) \prod_{i=1}^m \symp_i(z_i)\\ 
& =\!\! \sum_{\substack{z=(z_1, \dots, z_m) \in D^m \text{s.t.} \\ \forall i \notin I\ z_i = w_i }} |\hatpsi(z_{i_1}, \dots z_{i_{\ell}})| \cdot \symh_{m}(z)  \prod_{i=1}^m \symp_i(z_i) \\
& =\!\! \sum_{\substack{z=(z_1, \dots, z_m) \in D^m \text{s.t.} \\ \forall i \notin I\ z_i = w_i }} \hatpsi(z_{i_1}, \dots z_{i_{\ell}})  \prod_{i=1}^m \symp_i(z_i) \\
&=\!\! \left(\prod_{i \notin I} \symp_i(w_i)\right) \sum_{z=(z_1, \dots, z_{\ell}) \in D^{\ell}} \hatpsi(z) \prod_{i\in I} \symp_i(z_i) \\
&=\!\! \left(\prod_{i \notin I} \symp_i(w_i)\right) \sum_{x=(x_1, \dots, x_{\ell}) \in \{-1,1\}^{k \cdot \ell}} \psi(x) \prod_{i\in I} p_i(x_i) \\
&= 0.
\end{align*}
Here, the second equality holds by definition of $\hatnu_w$, and the final equality holds because $\psi$ has pure high degree at least $d$, and $\prod_{i\in I} p_i(x)$ is a polynomial of total degree at most $d$. To see the second inequality holds, suppose $\hatpsi(z_{i_1}, \dots z_{i_{\ell}}) > 0$, as $\psi$ agrees in sign with $g_{\ell}$, we must have $\symg_\ell(z_{i_1}, \dots z_{i_{\ell}}) = \False$. Recall that $g_\ell := \GapMaj{f}{\ell}{1-2\epsilon/3}$. This means that there are at least $(1-2\epsilon/3) \cdot \ell = (1-2\epsilon/3)(1-\epsilon/4) \cdot m \ge (1-\epsilon) \cdot m$ copies of $f$ that evaluate to $\False$. Hence $\symh_m(z)$ itself must be $\False$. So $|\hatpsi(z_{i_1}, \dots z_{i_{\ell}})| \cdot \symh_m(z) = \hatpsi(z_{i_1}, \dots z_{i_{\ell}})$. Similarly, when $\hatpsi(z_{i_1}, \dots z_{i_{\ell}}) < 0$, again as $\psi$ agrees in sign with $g_{\ell}$, we must have $\symg_\ell(z_{i_1}, \dots z_{i_{\ell}}) = \True$. So there are at least $(1-2\epsilon/3) \cdot \ell = (1-2\epsilon/3)(1-\epsilon/4) \cdot m \ge (1-\epsilon) \cdot m$ copies of $f$ evaluate to $\True$. Hence $\symh_m(z)$ must be $\True$ as well. So $|\hatpsi(z_{i_1}, \dots z_{i_{\ell}})| \cdot \symh_m(z) = -\hatpsi(z_{i_1}, \dots z_{i_{\ell}}) \cdot -1 = \hatpsi(z_{i_1}, \dots z_{i_{\ell}})$. Putting them together, we can see $|\hatpsi(z_{i_1}, \dots z_{i_{\ell}})| \cdot \symh_m(z) = \hatpsi(z_{i_1}, \dots z_{i_{\ell}})$ for all $z$ appearing in the summation, so the second inequality holds.

\end{proof}

\medskip \noindent \textbf{Stage 2.}
Now we move to Stage 2. With the distributions constructed for $w \in G$, for any point $v \in Z^+$, we construct a \dort{d} distribution that puts significant weight on it. 

\begin{claim} \label{claim:far}
	Let $G$ be as before, and suppose that for every $w \in G$ there exists a $d$-orthogonalizing distribution $\nu_w:\{-1, 1\}^{km} \to [0, 1]$ for $h_m$ that is symmetrized by $T^m$, and satisfies $\hatnu_w(w) \ge \delta$. Then for every $v \in (Z^+ \setminus G)$, there exists a $d$-orthogonalizing distribution $\rho_v$ for $h_m$ that is symmetrized by $T^m$, and 
	$$
	\hatrho_v(v) \ge 6\delta \cdot 2^{-2d} \cdot \binom{m}{d}^{-2}.
	$$
\end{claim}

The main technical ingredient in the proof of Claim \ref{claim:far} is the construction of a function $\phi: \{0, 1\}^m \to \mathbb{R}$ of pure high degree $d$ for which $\phi(1^m)$ is ``large''. This can be viewed as a dual formulation of a bound on the growth of low-degree polynomials. The construction of $\phi$ appears as part of the proof of such a bound in \cite{razborovsherstov}.

\begin{rem}
	We choose to state Lemma \ref{lem:hamming} below for a function $\phi : \{0, 1\}^m \to \mathbb{R}$, rather than applying our usual convention of working with functions over $\{-1, 1\}^m$, because it makes various statements in the proof of Claim \ref{claim:far} cleaner. To clarify the terminology below, we recall that a function $\phi : \{0, 1\}^m \to \mathbb{R}$ has \emph{pure high degree $d$} if $\sum_{x \in \{0, 1\}^m} \phi(x) \cdot p(x) = 0 $ for every polynomial $p : \{0, 1\}^m \to \mathbb{R}$ of degree at most $d$. The Hamming weight function $| \cdot | : \{0, 1\}^m \to [m]$ counts the number of $1$'s in its input, i.e. $|s| = s_1 + s_2 + \dots + s_m$.
\end{rem}

\begin{lemma}[cf. {\cite[Proof of Lemma 3.2]{razborovsherstov}}]\label{lem:hamming} \label{lemma15}
	Let $d$ be an integer with $0 \le d \le m - 1$. Then there exists a function $\phi: \{0, 1\}^m \to \mathbb{R}$ such that
	\renewcommand{\vmid}{\vspace{-7mm}}
	\fourpropsname{lemma15}{$\phi(1^m) = 1$}{$\phi(x) = 0 \text{ for all } d < |x| < m$}{$\phi \text{ has pure high degree at least } d$}{$\sum_{|x| \le d} |\phi(x)| \le 2^d {m \choose d}$}
	\renewcommand{\vmid}{\vspace{-8mm}}
\end{lemma}

\begin{proof}[Proof of Claim \ref{claim:far}]
	Fix $v \in (Z^+ \setminus G)$. Define an auxiliary function $\hatphi_v : D^m \to [0, 1]$ as follows.
	For any $z=(z_1, \dots, z_m)$, let
	\[\hatphi_v(z) := \sum_{\substack{s \in \{0, 1\}^m \text{ s.t.} \\ \forall i \ z_i = (1 - s_i) c + s_i v_i} } \phi(s),\]
	where $\phi$ is as in Lemma \ref{lem:hamming}, with $d$ set as in the conclusion of Claim \ref{claim:close} (observe that if there
	is some $z_i$ such that $z_i\neq c$ and $z_i \neq v_i$, then $\hatphi_v(z)=0$).
	
	Letting $\phi_v$ denote the function on $\{-1, 1\}^{km}$ induced from $\hatphi_v$ by $T^m$, we record some properties of $\phi_v$
	and $\hatphi_v$. 
	\fivepropsname{proof}{$\hatphi_v(v) = \phi(1^m) = 1$}{$\supp \hatphi_v \subset G \cup \{v\}$}{$\phi_v \text{ has pure high degree at least } d$}{$\|\phi_v\|_1 \le 2^d{m\choose d} + 1$}{$\hatphi_v \text{ is supported on at most } d \binom{m}{d}\text{ points in } D^m$}
	
	\noindent\textbf{Verifying Conditions (\ref{proofprop1})-(\ref{proofprop5}).} 
	For $s \in \{0,1\}^m$, we define
	$$\tau(s) := ((1-s_1)c + s_1 v_1,\dotsc, (1-s_m)c + s_m v_m).$$ 
	
	Then we can see $\hatphi_v(z) = \sum_{s \in \tau^{-1}(z)} \phi(s)$.
	
	To see Condition~\eqref{proofprop1} holds, note that $v \notin G$, so there are strictly larger than $\epsilon/4 \cdot m = d$ coordinates $v_i$ in $v$ satisfies $v_i \ne c$. Which means for $\tau(s) = v$, $|s| > d$. So by Condition~\eqref{lemma15prop2}, the only $s \in \tau^{-1}(v)$ satisfying $\phi(s) \ne 0$ is $1^m$, and we have $\hatphi_v(v) = \phi(1^m) = 1$.
	
	To verify Conditions~\eqref{proofprop2} and \eqref{proofprop5}, note that when $\hatphi_v(z) > 0$ for $z \ne v$, it means there exists some $s$ with $|s| \le d$, such that $z = \tau(s)$. By the definition of $\tau$, it means $z$ takes the value $c$ on at least $m - |s| \ge m - d  = (1-\epsilon/4) \cdot m$ coordinates, therefore $z \in G$. Moreover, one can see there are at most $\sum_{i=0}^{d} \binom{m}{i} \le d \cdot \binom{m}{d}$ such $z$'s.
	
	For Condition (\ref{proofprop3}), it is enough to show that if $p_1, \dots, p_m$ are polynomials over $\{-1, 1\}^k$ whose degrees sum to at most $d$, then 
	$\sum_{x=(x_1, \dots, x_m) \in \{-1,1\}^{km}} \phi_v(x) \prod_{i = 1}^m p_i(x_i) = 0$. To establish this,
	let $\symp_1, \dots, \symp_m \colon D \rightarrow \mathbb{R}$ denote polynomials satisfying
	$\deg(\symp_i) \leq \deg(p_i)$, and such
	that for all $i$ and all $z_i$ in the image of $T$, $\symp_i(z_i) := \Ex_{x \in T^{-1}(z_i)} [p_i(z_i)]$.
	Such polynomials are guaranteed to exist, since $T$ is degree non-increasing.
	Then:
	\begin{align*} \sum_{x=(x_1, \dots, x_m) \in \{-1,1\}^{km}} \phi_v(x) \prod_{i = 1}^m p_i(x_i) &=\sum_{z=(z_1, \dots, z_m) \in D^m} \hatphi_v(z) \prod_{i = 1}^m \symp_i(z_i)\\
	&= \sum_{z=(z_1, \dots, z_m) \in D^m} \left( \sum_{\substack{s \in \{0, 1\}^m \text{ s.t.} \\ \forall i \ z_i = (1 - s_i) c + s_i v_i} } \phi(s)\right) \prod_{i=1}^m \symp_i(z_i)\\
	&= \sum_{s \in \{0, 1\}^m} \phi(s)  \prod_{i = 1}^m \symp_i((1 - s_i) c + s_i v_i)\\
	&= 0,
	\end{align*}
	To see that the final equality holds, recall that that degrees of the polynomials $\symp_i$ sum to at most $d$.
	Hence, $p(s_1, \dots, s_m) := \prod_{i = 1}^m \symp_i((1 - s_i) c + s_i v_i)$ is a polynomial of degree strictly at most $d$ over $\{0,1\}^m$. The final equality
	then follows from the fact that
	$\phi$ has pure high degree at least $d$. 
	
	To establish Condition (\ref{proofprop4}), we check that
	\begin{align*}
	\sum_{z \in D^m, z\ne v} |\hatphi_v(z)| &\le \sum_{s \in \{0, 1\}^m, s \ne 1^m} |\phi(s)| \le  2^d {m \choose d},
	\end{align*}
	where the final inequality holds by Condition (\ref{lemma15prop4}).
	
	\noindent\textbf{Construction and analysis of $\rho_v$.}
	Up to normalization, the function $\phi_v \cdot h_m$ has all of the properties that we need to establish Claim \ref{claim:far}, except that there are locations where it may be negative. We obtain our desired orthogonalizing distribution $\rho_v$ by adding correction terms to $\hatphi_v$ in the locations where $\hatphi_v$ may disagree with $\symh_{m}$ in sign. These correction terms are derived from the distributions $\hatnu_w$ whose existence are hypothesized in the statement of Claim \ref{claim:far}. We start by defining
	
	\begin{equation} \hatRho_v(z) = \frac{\delta}{2^d {m \choose d} + 1}\symh_{m}(z)\hatphi_v(z) + \sum_{w \in (\supp \hatphi_v \setminus \{v\})} \hatnu_w(z). \label{diediedie} \end{equation}
	Observe that each $w$ appearing in the sum on the right hand side of \eqref{diediedie} is in the set $G$, owing to 
	Condition (\ref{proofprop2}). This guarantees that each term $\hatnu_w$ in the sum is well-defined.
	
	Now we check that $\hatRho_v$ is nonnegative. Since each term $\hatnu_w$ appearing in the sum on the right hand side of \eqref{diediedie} is a distribution (and hence non-negative), it suffices to check that $\hatRho_v(z) \geq 0$ for each point $z \in \supp \hatphi_v$. On each such point with $z \ne v$,  Condition (\ref{proofprop4}) guarantees that $\frac{\delta}{2^d {m \choose d}+1}\symh_{m}(z)\hatphi_v(z) \geq -\delta$. Moreover, the contribution of the sum is at least $\hatnu_z(z) \ge \delta$ by hypothesis. Hence, $\hatRho_v$ is a non-negative function.
	
	Next, we check that normalizing $\hatRho_v$ yields a distribution $\hatrho_v := \hatRho_v/\|\Rho_v\|_1$ for which $\hatrho_v(v) \ge 6\delta \cdot 2^{-2d} \cdot \binom{m}{d}^{-2}$ as required. By construction, $\hatRho_v(v) = \delta / \left(2^d {m\choose d}+1\right)$. Moreover, Conditions (\ref{proofprop1}), (\ref{proofprop4}), and (\ref{proofprop5}) together show that $\|\hatRho_v\|_1 \le \delta + d \binom{m}{d} \le 2d \binom{m}{d}$. Hence, 
	
	$$\hatRho_v(v) \geq \delta/\left(2d \binom{m}{d} \cdot \left(2^d {m \choose d} + 1\right)\right) \geq 6\delta/\left( 2^{2d} \cdot {m \choose d}^2\right),$$ 
	
	as $d = \epsilon/4 \cdot m \ge 10$ by assumption.
	
	Finally, we must check that $\rho_v = \Rho_v/\|\Rho_v\|_1$ is $d$-orthogonalizing for $h_m$.
	To see this, observe that $\Rho_v \cdot h_m$ is a linear combination of the functions $\phi_v$ and $\nu_w \cdot h_m$ for $w \in (\supp \hatphi_v \setminus \{v\})$. Moreover, each of these functions has pure high degree at least $d$ ($\phi_v$ does so by Condition (\ref{proofprop3}), while
	$\nu_w \cdot h_m$ does by the fact that $\nu_w$ is $d$-orthogonalizing for $h_m$). By linearity, it follows that
	$\Rho_v \cdot h_m$ has pure high degree at least $d$, so $\rho_v$ is $d$-orthogonalizing for $h_m$ as desired.
	
	This completes the proof of Claim \ref{claim:far}. 
\end{proof}

At last we are ready to conclude the prooof of Theorem~\ref{theo:smooth-ort}.
	By Claim \ref{claim:close}, for every $w \in G$ there exists a $d$-orthogonalizing distribution $\nu_w : \{-1,1\}^{km} \to [0, 1]$ for $h_{m}$ that is symmetrized by $T^m$, with $\hatnu_w(w) \ge (2\eta)^{m}/6$. Thus, by Claim \ref{claim:far}, it is also true that for every $v \in (Z^+ \setminus G)$, there is a $d$-orthogonalizing distribution $\rho_v : \{-1,1\}^{km} \to [0, 1]$ that is symmetrized by $T^m$, with $\hatrho_v(v) \ge (2\eta)^{m}\cdot 2^{-2d} \cdot {m \choose d}^{-2}$. Now, for each element $z \in Z^{+}$, we define its weight, $W_z = |(T^{m})^{-1}(z)|$. Consider the following distribution:
	\[
	\hatmu(z) = \left(\sum_{z \in Z^+} W_z \right)^{-1} \cdot \left(\sum_{w \in G} W_w \cdot \hatnu_w(z) + \sum_{v \in (Z^+ \setminus G)} W_v \cdot \hatrho_v(z)\right).
	\]
	We verify that the (un-symmetrized) distribution $\mu: (\{-1,1\}^k)^m \to [0, 1]$ satisfies our requirements. As $\hatmu$ is a convex combination of \dort{d} distributions for $\symh_m$, it is itself a \dort{d} distribution for $\symh_m$, therefore $\mu$ is a \dort{d} distribution for $h_m$. Now for each $x \in \zeroset(h_m)$, let $w = T^{m}(x)$, we have 
	\begin{align*}
	\mu(x) &\ge \left(\sum_{z \in Z^+} W_z \right)^{-1} \cdot \frac{1}{W_w} \cdot W_w \cdot (2\eta)^{m}\cdot 2^{-2d} \cdot {m \choose d}^{-2}\\
	       &= \left| \zeroset(h_m) \right|^{-1} (2\eta)^{m}\cdot 2^{-2d} \cdot {m \choose d}^{-2} \\
	       &\ge 2^{-mk} \cdot (2\eta)^{m}\cdot 2^{-2d} \cdot {m \choose d}^{-2}.
	\end{align*}
	
	This completes the proof.
        

\subsection{Exhibiting A Problem Separating $\NISZK^{\cc}$ From $\UPP^{\cc}$}\label{sec:niszk-cc}
 \label{sec:niszkcc-uppcc}

Now we are ready to prove the communication complexity classes separation $\NISZKcc \not\subset \UPPcc$. In order to utilize our lifting Theorem~\ref{theo:lifting-upp-cc}, we have to choose a partial function $f$ that satisfies: (1) it has an efficient $\NISZK$ protocol, so that $\Gapmajority(f)$ also is in $\NISZK$; (2) it belongs to our partial function class $\cd{d}{a}$ for appropriate choices of $d$ and $a$ (cf. Definition \ref{defcda}). However, it turns out that the $\Collision$ function, which we used in the query complexity case, does not satisfy the second condition in the definition of $\cd{d}{a}$. In fact, $\cd{d}{a}$ requires the function evaluates to $\False$ nearly everywhere, but $\Collision$ is undefined on most inputs, as a random function is neither a permutation nor $k$-to-1.

To address this issue, we use the $\PTP$ problem (cf. Definition~\ref{defi:ptp}) instead of $\Collision$. We have the following lemma, showing $\PTP$ is the function we want.

\begin{lemma}\label{lm:ptp-in-cd}
	$\PTP_n \in \cd{d}{a}$ for any $a > 1$ and $d = \Omega(\sqrt[3]{n}/a)$.
\end{lemma}

We defer the proof of Lemma \ref{lm:ptp-in-cd} to Appendix \ref{sec:ptpsec}, restricting ourselves here to a brief sketch as follows.
Bun and Thaler~\cite{bt16} gave a primal condition that implies the existence of a suitable dual object. The existence of dual object required by $\cd{d}{a}$ (namely, Condition~\eqref{newcdprop1}) can be easily proved by combining this primal condition with a simple modification of Kutin's proof for the approximate degree of $\Collision$. For Condition~\eqref{newcdprop2} of $\cd{d}{a}$, it suffices to use a Chernoff bound to show a random function from $[n] \to [n]$ is far from any permutation.
The details of the proof can be found in Appendix~\ref{sec:dual-obj-PTP}.

Now we are ready to prove the communication complexity class separation $\NISZKcc \not\subset \UPPcc$.

\medskip \noindent \textbf{Theorem.} (Restatement of Theorem \ref{THMCC})
	$\NISZKcc \not\subset \UPPcc$.
\begin{proof}
	 
	 Let $k$ be a sufficiently large integer, and $\epsilon = 1/10\log k$, then by Lemma~\ref{lm:ptp-in-cd}, $\PTP_k \in \cd{d}{40/\epsilon}$ for some $d = \Omega(\sqrt[3]{k}/\log k )$. Then we define $F := \GapMaj{\PTP_k}{d}{1-\epsilon}$, $n := d \cdot k$, and $M$ be the $(N,n,F)$-pattern matrix for $N = 2^{36 + 6 \log\epsilon^{-1}} \cdot n$. Invoking Theorem~\ref{theo:lifting-upp-cc}, we have $\UPPcc(M) = \Omega(d \cdot \epsilon) = \Omega(\sqrt[3]{k}/\log^{2} k) = \Omega(\sqrt[4]{n}/\log^2 n)$. Recall that in the communication problem (cf. Section~\ref{sec:pmat}), Alice gets an input of length $N = \polylog(n) \cdot n$, while Bob gets an input of length at most $n \cdot \ln N + n = \polylog(n) \cdot n$. Therefore, we conclude that $M$ is not in $\UPPcc$ by the definition of $\UPPcc$ in Section~\ref{sec:sign-rank-upp-cc}.

	\medskip
	\noindent\textbf{Showing it is sufficient to construct an $\NISZK$ protocol for $F$.}
	Next we show $M$ is in $\NISZKcc$. We claim that it is enough to show that $F$ admits an efficient $\NISZKdt$ protocol $P$. 
	First we recall some notation from Section~\ref{sec:pmat}. In the communication problem corresponding to $M$, Alice gets a sequence of bits $x \in \{-1,1\}^N$ , while Bob gets a set of coordinates $S = \{s_1,\dotsc,s_n\}$ and a shift $w \in \{-1,1\}^n$. They together want to compute $F(u)$ for $u = x |_S \oplus w$. Alice and Bob can simulate $P$ to solve $M$ as follows: Alice interacts with the prover as if she were the verifier in $P$. Whenever she needs the value of $u_i$, she asks Bob to send her the value of $s_i$ and $w_i$. Then she knows $u_i = x_{s_i} \oplus w_i$. The correctness of this protocol follows directly from the correctness of $P$, and it only incurs a logarithmic overhead in the running time. Therefore, if $F$ admits an efficient $\NISZK$ protocol, it can be transformed to an efficient $\NISZKcc$ protocol for $M$.
	
	\medskip
	\noindent\textbf{Reduction to $\EA$.}
	Now we prove that $F$ has an efficient $\NISZKdt$ protocol by reducing it to $\EA$, which is $\NISZK$-complete (cf. Definition~\ref{defi:EA}). Given an input $x \in \{-1,1\}^{d \cdot k}$ to $F$, let $x=(x_1,\dotsc,x_{d})$, where $x_i$ denotes the input to the $i$-th copy of $\PTP_k$. By the definition of $\PTP_k$, we further interpret $x_i$ as a function $f_i : [k] \to [k]$. 
	
	We define the distribution $\distr(x)$ as follows: pick $i \in [d]$ and $x \in [k]$ at uniformly random, and then output the pair $(i,f_i(x))$. 
	For a function $f : [k] \to [k]$, let $\distr_{f}$ be the distribution obtained by outputting $f(x)$ for a uniformly randomly chosen $x \in [k]$. Then we can express $\distr(x)$ as $\distr(x) := \frac{1}{d} \cdot \sum_{i=1}^{d} \{i\} \times \distr_{f_i}$. Note when $f$ is a permutation, i.e. $\PTP(f) = \True$, we have $H(\distr_{f}) = \log(k)$. And when $\PTP(f) = \False$, by the definition of $\PTP_k$, the size of the support of the distribution $\distr_f$ is at most $7/8 \cdot k$, therefore $H(\distr_{f}) \le \log(k \cdot 7/8) \le \log k - 0.18$. 
	
	Also, let the output of $\distr(x)$ be the random variable pair $(X,Y)$, note $Y$ depends on $X$, we have
	
	\begin{equation}\label{eq:ent-sum-2} \notag
	H(\distr(x)) = H(X,Y) = H(X) + H(Y|X) = \log d + \frac{1}{d} \cdot \sum_{i=1}^{d} H(\distr_{f_i}).
	\end{equation}
	
	With the above observation, we can bound $H(\distr(x))$ easily: when $F(x) = \True$, 
	$$
	H(\distr(x)) \ge \log d + (1-\epsilon) \cdot \log k \ge \log d + \left(1 - \frac{1}{10\log k}\right) \cdot \log k  \ge \log n - \frac{1}{10},$$ 
	as there is at least a $1-\epsilon$ fraction of $f_i$'s satisfy $H(\distr_{f_i}) = \log k$ and $\log d + \log k = \log dk = \log n$; when $F(x) = \False$, 
	$$H(\distr(x)) \le \log d + \epsilon \cdot \log k + (1-\epsilon) \cdot (\log k - 0.18) \le \log n - 0.15,$$
	as there is at least a $1-\epsilon$ fraction of $f_i$'s satisfy $H(\distr_{f_i}) \le \log k - 0.18$, and $(1-\epsilon) \cdot 0.18 \ge 0.15$ when $k$ is sufficiently large.
	
	Finally, we take the reduction to be $A(x) = \distr(x)^{\otimes 50}$, i.e., a sample from $A(x)$ is a sequence of 50 i.i.d. samples from $\distr(x)$. Then we can see when $F(x) = \True$, $H(A(x)) \ge 50 \left(\log n - \frac{1}{10}\right) \ge 50 \log n - 5$, and when $F(x) = \False$, $H(A(x)) \le 50 (\log n - 0.15) \le 50 \log n - 7.5$. Therefore, the pair $(A(x),50 \log n - 6.25)$ is a valid reduction to $\EA$ and this completes the proof.
\end{proof}

\section{$\PTP$ is in $\cd{d}{a}$}\label{sec:dual-obj-PTP}
\label{sec:ptpsec}
In order to show $\PTP \in \cd{d}{a}$, we need to prove the existence of a suitable dual object, and show that nearly all inputs to $\PTP_n$ evaluate to $\False$.

\subsection{Nearly All Inputs to $\PTP$ Evaluate to \False}

We begin with the second condition, which is relatively easy. It is tantamount to verify that nearly every function from $[n] \to [n]$ is far from any permutation, which can in turn be proved by a simple application of a Chernoff bound.

\begin{lemma}\label{lm:most-false}
	With probability at least $1-2^{\Omega(n)}$, a random function $f$ from $[n] \to [n]$ satisfies $\PTP_n(f) = \False$.
\end{lemma}
\begin{proof}
	For each $i$ in $[n]$, we define the random variable $x_i$ to be the indicator of whether $f^{-1}(i) \ne \emptyset$, and we let $X := \frac{1}{n} \cdot \sum_{i=1}^n x_i$.
	Then we have $\Pr[x_i = 0] = \left(1-\frac{1}{n}\right)^n \approx e^{-1}$. S $x_i$ takes values in $\{0, 1\}$, we have $\Ex[X]=\Ex[x_i] \approx 1-e^{-1}$.
	
	Although the $x_i$'s are not independent, by~\cite{dubhashi1996balls}, they are negatively associated, which means that we can still apply a Chernoff bound to obtain the following concentration result:
	
	$$
	\Pr[X > 0.65] \le e^{-\Omega(n)},
	$$
	
	as $1-e^{-1} \sim 0.63 < 0.65$. Note that for a function $f$ with $X \le 0.65$, we have $\PTP_n(f) = \False$, as it must differ on at least $0.35 \cdot n > n/8$ coordinates with any permutation. This completes the proof.
\end{proof}

\subsection{A Primal Condition}

In order to show the existence of a suitable dual object for the $\PTP$ problem, we introduce a sufficient {\em primal} condition that was given in~\cite{bt16}. The original statement  from \cite{bt16} only considers total functions.
But it is easy to observe that the proof in \cite{bt16} makes no use of the fact that the function is total; hence the original proof works for partial functions as well. 

\begin{defi}
	Let $T : \{-1,1\}^k \to D$ be a symmetrization for a partial function $f : \{-1,1\}^k \to \{-1,0,1\}$. Let
	$V = T^{-1}(\tilde{V})$ for some $\tilde{V} \subseteq \tilde{f}^{-1}(1)$. We say that $p : \{-1,1\}^k \to \R$ is a weak $\epsilon$-error
	one-sided approximation to $f$ under the promise that the input $x$ is in $V \cup f^{-1}(-1)$ (with respect
	to $T$) if the following holds. Define $q: \{-1,1\}^k \to \R$ by $q(x) := \Ex_{y : T(y)=T(x)}[p(y)]$. Then $q$ satisfies the following three properties:
	\begin{itemize}
		\item $q(x) \le -1 + \epsilon$ for all $x \in f^{-1}(-1)$.
		\item $|q(x) - 1| \le \epsilon$ for all $x \in V$.
		\item $|q(x)| \le 1+\epsilon$ for all $x \in \{-1,1\}^k$ such that $f(x) \ne 0$ and $x \not\in (f^{-1}(-1) \cup V)$.
	\end{itemize}
\end{defi}

\begin{theo}[Essentially Theorem~B.1 in \cite{bt16}]\label{theo:primal-cond}
	Let $T$ be a symmetrization for $f$. Let $V = T^{-1}(z_+)$ for some $z_+ \in \tilde{f}^{-1}(1)$. If there
	does not exist a weak $2\eta$-error, degree-$d$ one-sided approximation to $f$ under the promise that
	the input is in $V \cup f^{-1}(-1)$, then there exists a function $\hatpsi \colon D \rightarrow \mathbb{R}$ with $\psi$ being the associated function on $\{-1,1\}^k$ induced by $T$, such that
	\fivepropsname{btdual}{$\langle \psi, f \rangle \ge \eps$ and $\psi$ only takes non-zero values on $D_f$.}{$\|\psi\|_1 = 1$}{$\psi \text{ has pure high degree at least }d$}{$f(x) = -1 \implies \psi(x) < 0$}{$\hatpsi(z_+) \ge \eta \text{ for some } z_+ \in D \text{ satisfying } \symf(z_+) = 1$}
\end{theo}

Observe that the conditions in the above theorem are indeed strictly stronger than our requirements for a $(d,\epsilon,\eta)$-dual object.\footnote{In our definition of a $(d,\epsilon,\eta)$-dual object, we don't require $\psi$ to be zero outside of $D_f$, and Condition~\eqref{btdualprop4} is not demanded as well.} So we have the following corollary.

\begin{cor}\label{cor:cond}
	Let $T$ be a symmetrization for $f$. Let $V = T^{-1}(z_+)$ for some $z_+ \in \tilde{f}^{-1}(1)$. If there
	does not exist a weak $2\eta$-error, degree-$d$ one-sided approximation to $f$ under the promise that
	the input is in $V \cup f^{-1}(-1)$, then $f$ has a $(d,2\eta,\eta)$-dual object.
\end{cor}

\subsection{Existence of a Suitable Dual Object for $\PTP$}

We begin by defining a natural symmetrization $T :\{-1,1\}^{M} \to D$ for $\PTP_n$, which is also used in~\cite{bt16}.

Let $x \in \{-1,1\}^M$ be an input to $\PTP_n$, and let $f$ be the corresponding function from $[n] \to [n]$. We define $T_1(x) := (|f^{-1}(1)|,|f^{-1}(2)|,\dotsc,|f^{-1}(n)|) \in \R^{n}$. And for any vector $v \in \R^{n}$, we define $T_2(v)$ be the vector in $\R^n$ which sorts the coordinates in $v$ in increasing order. Our final symmetrization is defined as $T:= T_2 \circ T_1$, that is, first count the number of occurrences of each value in the image, and then sort them in ascending order. It is easy to see that $T$ symmetrizes $\PTP$, and by Lemma~C.2 in~\cite{bt16}, it is a degree non-increasing map.

Let $z_+$ be the point in the symmetrized domain $D$ representing all the $2$-to-$1$ inputs and $V = T^{-1}(z_+)$. We prove the following lemma in this subsection.

\begin{lemma}\label{lm:no-weak-approx}
	For all $\epsilon \in (0,1)$, any weak $\epsilon$-error, degree-$d$ one-sided approximation to $\PTP_n$ under the promise that
	the input is in $V \cup f^{-1}(-1)$, must have $d = \Omega\left( (1-\epsilon) \cdot n^{1/3} \right)$.
\end{lemma}

Our proof is a simple modification of Kutin's lower bound for approximate degree of $\Collision$~\cite{kutin2005quantum}. We use a more sophisticated version of a lemma from Paturi~\cite{paturi1992degree}, as we want an explicit dependence on the approximate error, rather than treating it as a small constant. 

Following Kutin~\cite{kutin2005quantum}, we define a special collection of functions which are $a$-to-1 on one part of the domain and $b$-to-1 on the other part. We call a triple of numbers $(m, a, b)$ \emph{valid} if $a | m$ and $b | (n-m)$. For each valid triple $(m, a, b)$, we define
$$
g_{m, a, b}(i) = \begin{cases}
\lceil i / a \rceil \quad \text{if } 1 \le i \le m \\
n - \lfloor (n - i) / b \rfloor \quad \text{if } m < i \le n.
\end{cases}
$$
and $R_{m,a,b} := T^{-1}(T(g_{m,a,b}))$.

We have the following important lemma from~\cite{kutin2005quantum}.

\begin{lemma}[Lemma 2.2 in~\cite{kutin2005quantum}]\label{lm:kutin-sym}
	Let $P(x)$ be a degree-$d$ polynomial in $\{-1,1\}^M$. For a valid triple $(m,a,b)$, define $Q(m,a,b)$ by
	$$
	Q(m,a,b) = \Ex_{y:y \in R_{m,a,b}}[P(y)].
	$$
	
	Then $Q(m,a,b)$ is a degree-$d$ polynomial in $m,a,b$.
\end{lemma}

Now, suppose there is a weak $\epsilon$-error, $d$-degree one-sided approximation $p$ to $f$ under the promise that the input is in $V \cup f^{-1}(-1)$. We are going to show that $d = \Omega\left( (1-\epsilon) \cdot n^{1/3} \right)$.

\newcommand{\cst}{c_1}
By Lemma~\ref{lm:kutin-sym}, let $\cst = 1 - \epsilon$, then there is a degree-$d$ polynomial $Q(m,a,b)$ such that

\begin{itemize}
	\item $Q(m,1,1) \le -1 + \epsilon = -\cst$ for any $m$.
	\item $Q(m,2,2) \in [\cst,2-\cst]$ for any $2 | m$.
	\item $Q(m,a,b) \in [-2+\cst,2-\cst]$ for any valid $(m,a,b)$ such that $\PTP(g_{m,a,b}) = \False$.
\end{itemize}

We need the following lemma by Paturi~\cite{paturi1992degree}.

\begin{lemma}
	[Paturi~\cite{paturi1992degree}]\label{theo:paturi}Let $a,b$ be two reals and $q:\mathbb{R}\rightarrow\mathbb{R}%
	$\ be a univariate polynomial such that $|q\left(  j\right)|  \leq\delta
	$\ for all integers $j\in\left[  a,b\right]  $, and suppose that $\left\vert
	q\left(  \left\lceil x\right\rceil \right)  -q\left(  x\right)  \right\vert
	\ge c \cdot  \delta $\ for some $x\in\left[  a,b\right]  $ and a real $c \in (0,1)$.\ \ Then
	$\deg\left(  q\right)  =\Omega\left(  c \cdot \sqrt{\left(  x-a+1\right)  \left(
		b-x+1\right)  }\right)  $.
\end{lemma}

Now we prove Lemma~\ref{lm:no-weak-approx} by showing the polynomial $Q$ must have degree at least $d=\Omega\left((1-\epsilon) \cdot n^{1/3} \right) \newline= \Omega\left(\cst \cdot n^{1/3} \right)$. The following proof basically mimics the original proof in~\cite{kutin2005quantum}.

\begin{proofof}{Lemma~\ref{lm:no-weak-approx}}
	Let $M = n/2$, depending on the value of $Q(M,1,2)$, there are two cases.
	
	\begin{itemize}
		\item $Q(M,1,2) \ge 0:$ Let $g(x) = Q(M,1,2x)$ and $k$ be the least positive integer such that $|g(k)| \ge 2$. Then we have $|g(x)| \le \cst$  for all positive integers $ < k$, and $g(1) - g(1/2) \ge c$ by assumption. Hence by Theorem~\ref{theo:paturi}, we have 
		$$
		d = \Omega(\cst \cdot \sqrt{k}).
		$$
		
		Now, let $c = 2k$ and consider the polynomial $h(i) = Q(n - ci,1,c)$. For any integer i with $\lceil n/4c \rceil \le i \le \lfloor n/c \rfloor$, the triple $(n-ci,1,c)$ is valid, and it is easy to see $\PTP$ evaluates to \False\ on $g_{n-c_i,1,c}$ for all those $i$'s, as at least $n/4$ inputs belong to the $c$-to-$1$ part and $c \ge 2$. Hence we have $|h(i)| \le 2-\cst$ for $i$ in that range. But $\left|h\left(\frac{n}{2c}\right)\right| = |Q(M,1,c)| = |g(k)| \ge 2$. Therefore, by Theorem~\ref{theo:paturi}, we have
		$$
		d = \Omega(c_1 \cdot n/c) = \Omega(\cst \cdot n/k).
		$$
		
		Putting them together, we have $d^3 =\Omega(\cst^3 \cdot n/k \cdot k) = \Omega( \cst^3 \cdot n)$, which means $d = \Omega( \cst \cdot n^{1/3})$.
		
		\item $Q(M,1,2) < 0:$ We let $g(x) = Q(M,2x,2)$ and $k$ be the least positive integer such that $|g(k)| \ge 2$. Then we have $g(1)-g(1/2) \ge \cst$. So again by Theorem~\ref{theo:paturi}, we have $d = \Omega(\cst \cdot \sqrt{k})$.
		
		Then, let $c=2k$ and $h(i) = Q(ci,c,2)$. For any integer $i$ with $0 \le i \le \lfloor n/c \rfloor$, the triple $(ci,c,2)$ is valid (both $n$ and $c$ is even), and clearly $\PTP(g_{ci,c,2}) = \False$ for those $i$'s. Hence we have $|h(i)| \le 2-\cst$. But $\left|h\left(\frac{n}{2c}\right)\right|=|g(k)| \ge 2$. Again by Theorem~\ref{theo:paturi}, we have $d = \Omega( c_1 \cdot n/k)$.
		
		Similarly, we also have $d = \Omega(\cst \cdot n^{1/3})$ in this case. This competes the proof.
	\end{itemize}
\end{proofof}

\subsection{Proof for Lemma~\ref{lm:ptp-in-cd}}

Finally, we prove Lemma~\ref{lm:ptp-in-cd}.
\begin{proofof}{Lemma~\ref{lm:ptp-in-cd}}
	Let $\eta = \frac{1-1/2a}{2}$, by Lemma~\ref{lm:no-weak-approx} and Corollary~\ref{cor:cond}, there exists a $(d,2\eta,\eta)$-dual object for $\PTP_n$ with respect to symmetrization $T$, for some $d = \Omega( (1-2\eta) \cdot n^{1/3}) = \Omega(n^{1/3}/a)$. Note $\frac{2\eta}{1-2\eta} > a$, hence this dual object satisfies Condition~\eqref{newcdprop1} of $\cd{d}{a}$.
	
	And by Lemma~\ref{lm:most-false}, the Condition~\eqref{newcdprop2} of $\cd{d}{a}$ follows immediately.
\end{proofof}

\section{A Weaker Polarization Lower Bound Using Fourier Analysis}
\label{app:fourier}
Here we show that, if one only cares about black box polarization in the restricted form proposed by Holenstein and Renner \cite{HR05}, then one can prove a lower bound against polarization directly using Fourier analysis alone. This may help the readers understand what's going on in the proof. But please note this result is subsumed by our oracle separation between $\SZK$ and $\PP$. 

\begin{defi}
    \label{defi:special-polarization}
    An \emph{$(n,\ell,m)$-special polarizer} is a pair of joint disributions over pairs of strings, $(S^0, R^0)$ and $(S^1,R^1)$, where $S^0$ and $S^1$ are over $\bset^n$, and $R^0$ and $R^1$ are over $\bset^\ell$.

    For any distributions $D_0$ and $D_1$, we define the polarized distributions $\widehat{D}_0$ and $\widehat{D}_1$ resulting from this polarizer as:
    \begin{align*}
      \widehat{D}_b = (D_{S^b_1}, \dots, D_{S^b_n}, R^b)
    \end{align*}
    The polarizer then provides the following guarantees:
    \begin{align*}
      ||D_0 - D_1|| > 2/3 &\implies ||\widehat{D}_0 - \widehat{D}_1|| > 1-2^{-m}\\
      ||D_0 - D_1|| < 1/3 &\implies ||\widehat{D}_0 - \widehat{D}_1|| < 2^{-m}
    \end{align*}

    An \emph{$(n,\ell)$-pseudo polarizer} is the same, except it doesn't provide the above guarantees. 
\end{defi}

It is to be noted that the technique for polarizing distance between distributions from \cite{sahai2003complete} is a special polarizer. Note also that any $(n,\ell,m)$-special polarizer is an $(n,\ell)$-pseudo polarizer.

Consider distributions over $\bset^k$. If there existed a polynomial-time computable $(n,\ell,m)$-special polarizer such that $nk+\ell < 2m$, then Theorem~\ref{thm:szkppandpolarization} implies that deciding whether pairs of such distributions are close or far can be done in $\PP$. If such a polarizer existed for every $k$, then this would imply that $\SZK$ is contained in $\PP$ because of the completeness of the Statistical Distance problem~\cite{sahai2003complete}. We rule out this approach of showing such a containment with the following theorem.

\begin{theo}
    \label{thm:special-polarization}
    For any $(n,\ell,m)$-special polarizer, $n = \Omega(m)$.
\end{theo}

Theorem~\ref{thm:special-polarization} follows immediately from the following two lemmas. For any $\alpha \in [0,1]$ and bit $b$, denote by $D^\alpha_b$ the distribution over $\bset$ that is equal to $b$ with probability $(1+\alpha)/2$. It is easy to see that $||D^\alpha_0 - D^\alpha_1||= \alpha$. We denote by $(\widehat{D}^\alpha_0,\widehat{D}^\alpha_1)$ the distributions that result from applying the special polarizer in the relevant context to $(D^\alpha_0, D^\alpha_1)$ and by $(\widetilde{D}^\alpha_0, \widetilde{D}^\alpha_1)$ the distributions resulting from the pseudo-polarizer.

\begin{lemma}
    \label{lem:special-polarization-1}
    For any $(n,\ell)$-pseudo polarizer and any $\alpha, \beta \in (0,1)$ such that $\alpha > \beta$,
    \begin{align*}
      \frac{||\widetilde{D}^\alpha_0 - \widetilde{D}^\alpha_1||}{||\widetilde{D}^\beta_0 - \widetilde{D}^\beta_1||} &\leq 2^{(n+\ell)/2} \left( \frac{\alpha}{\beta} \right)^{n}
    \end{align*}
\end{lemma}

\begin{proof}

    Throughout the proof, we use the symbols for distributions interchangeably with the symbols for vectors representing their mass functions. For each $\alpha \in (0,1)$, we define the following matrix:
    \begin{align*}
      B_\alpha = \left( \begin{array}{cc} \frac{1+\alpha}{2} & \frac{1-\alpha}{2} \\ \frac{1-\alpha}{2} & \frac{1+\alpha}{2}  \end{array}  \right)
    \end{align*}
    Consider any distribution $p$ over $\bset$. The distribution obtained by selecting a bit $b$ according to $p$ and then sampling $D^\alpha_b$ is given by $B_\alpha p$. This can be extended to the case when $p$ is over $\bset^n$ -- if $x$ is drawn according to $p$, the distribution of $(D^\alpha_{x_1}, \dots, D^\alpha_{x_n})$ is given by $B_\alpha^{\otimes n} p$.

    Further, if $p_0$ happens to be the distribution of $(S^0, R^0)$ from an $(n,\ell,m)$ special polarizer, then $\widetilde{D}^\alpha_0$, when the polarizer is applied to $(D^\alpha_0, D^\alpha_1)$, is given by $(B_\alpha^{\otimes n} \otimes I^{\otimes \ell}) p_0$, where $I$ is the $2\times 2$ identity matrix. Similarly, $\widetilde{D}^\alpha_1$ would be $(B_\alpha^{\otimes n} \otimes I^{\otimes \ell}) p_1$. Let $C_\alpha = (B_\alpha^{\otimes n} \otimes I^{\otimes \ell})$. We then have:
    \begin{align*}
      ||\widetilde{D}^\alpha_0 - \widetilde{D}^\alpha_1|| = \frac{1}{2} \lone{C_\alpha (p_1 - p_0)}
    \end{align*}

    Both $B_\alpha$ and $I$ have the vectors $\left( \begin{array}{c} 1 \\ 1 \end{array} \right)$ and $\left( \begin{array}{c} 1 \\ -1 \end{array} \right)$ as eigenvectors. The corresponding eigenvalues are $1$ and $\alpha$ for $B_\alpha$, and both $1$ for $I$. This implies that the eigenvectors of $B$ are all possible tensor products of these eigenvectors, and the eigenvalue of such a resulting vector is simply the products of the eigenvalues of the vectors that were tensored.

    In different terms, the eigenvectors are $(\chi_{T_1}\otimes\chi_{T_2})$ for any $T_1\subseteq [n]$ and $T_2\subseteq [\ell]$, which are the characters of $\mathbb{F}_2^{n+\ell}$, and the eigenvalue of this vector would be $\alpha^{|T_1|}$. Since these vectors form a basis, we can write $p_0 = \sum_{T_1,T_2} \widehat{p}_{0,(T_1,T_2)} (\chi_{T_1}\otimes\chi_{T_2})$. 

    Using the standard relationships between $L_1$ and $L_2$ norms, we have the following inequalities for any $\alpha, \beta \in (0,1)$ such that $\alpha > \beta$:
    \begin{align*}
      \frac{||\widehat{D}^\alpha_0 - \widehat{D}^\alpha_1||}{||\widehat{D}^\beta_0 - \widehat{D}^\beta_1||} &= \frac{\lone{C_\alpha (p_1 - p_0)}}{\lone{C_\beta (p_1 - p_0)}}\\
                                                                                                              &\leq 2^{(n+\ell)/2} \frac{\ltwo{C_\alpha (p_1 - p_0)}}{\ltwo{C_\beta (p_1 - p_0)}}\\
                                                                                                              &= 2^{(n+\ell)/2} \frac{\ltwo{ C_\alpha \sum_{T_1\subseteq [n], T_2\subseteq [\ell]} (\widehat{p}_{1,(T_1,T_2)} - \widehat{p}_{0,(T_1,T_2)}) (\chi_{T_1}\otimes\chi_{T_2}) }}{\ltwo{ C_\beta \sum_{T_1\subseteq [n], T_2\subseteq [\ell]} (\widehat{p}_{1,(T_1,T_2)} - \widehat{p}_{0,(T_1,T_2)}) (\chi_{T_1}\otimes\chi_{T_2}) }}\\
                                                                                                              &= 2^{(n+\ell)/2} \frac{\ltwo{\sum_{T_1\subseteq [n], T_2\subseteq [\ell]} \alpha^{|T_1|} (\widehat{p}_{1,(T_1,T_2)} - \widehat{p}_{0,(T_1,T_2)}) (\chi_{T_1}\otimes\chi_{T_2}) }}{\ltwo{\sum_{T_1\subseteq [n], T_2\subseteq [\ell]} \beta^{|T_1|} (\widehat{p}_{1,(T_1,T_2)} - \widehat{p}_{0,(T_1,T_2)}) (\chi_{T_1}\otimes\chi_{T_2}) }}\\
                                                                                                              &= 2^{(n+\ell)/2} \left(\frac{\sum_{T_1\subseteq [n], T_2\subseteq [\ell]} \alpha^{2|T_1|} (\widehat{p}_{1,(T_1,T_2)} - \widehat{p}_{0,(T_1,T_2)})^2 }{\sum_{T_1\subseteq [n], T_2\subseteq [\ell]} \beta^{2|T_1|} (\widehat{p}_{1,(T_1,T_2)} - \widehat{p}_{0,(T_1,T_2)})^2 } \right)^{1/2}\\
                                                                                                              &\leq 2^{(n+\ell)/2} \left( \frac{\alpha}{\beta} \right)^{n}
    \end{align*}
    where the last inequality follows from the readily verified fact that for any sequences of positive real numbers $\{a_i\}$, $\{b_i\}$, and $\{c_i\}$, $\frac{\sum_i c_i a_i}{\sum_i c_i b_i}$ is at most $\max_i \frac{a_i}{b_i}$.
\end{proof}

\begin{lemma}
    \label{lem:special-polarization-2}
    For any $(n,\ell)$-pseudo polarizer and any $\alpha, \beta \in (0,1)$ such that $\alpha > \beta$, there is an $(n,1)$-pseudo polarizer such that:
    \begin{align*}
      \frac{||\widetilde{D}^\alpha_0 - \widetilde{D}^\alpha_1||}{||\widetilde{D}^\beta_0 - \widetilde{D}^\beta_1||} &\geq \frac{||\widehat{D}^\alpha_0 - \widehat{D}^\alpha_1||}{||\widehat{D}^\beta_0 - \widehat{D}^\beta_1||}
    \end{align*}
\end{lemma}

\begin{proof}
    The lemma follows from the following two easily verified facts about Total Variation distance of joint distributions.

    \begin{ft}
        \label{ft:tv1}
        For random variables $X$, $Y$ and $Y'$,
        \begin{align*}
          ||(X,Y) - (X,Y')|| = \sum_x \Pr[X=x] ||Y_{|X=x} - Y'_{|X=x}||
        \end{align*}
    \end{ft}

    \begin{ft}
        \label{ft:tv2}
        For random variables $X_0$ and $X_1$ and a uniformly distributed bit $B$,
        \begin{align*}
          ||(B,X_B) - (\overline{B},X_B)|| = ||X_0 - X_1||
        \end{align*}
    \end{ft}

    For convenience, we write the resulting distributions from a polarizer as $\widehat{D}^\alpha_0 = (D^\alpha_{S^0},R^0)$, etc., which is indeed the structure that these distributions have. From the above two facts, we have the following for a uniformly distributed bit $B$:
    \begin{align*}
      \frac{||\widehat{D}^\alpha_0 - \widehat{D}^\alpha_1||}{||\widehat{D}^\beta_0 - \widehat{D}^\beta_1||} &= \frac{||({D}^\alpha_{S^0},R^0) - ({D}^\alpha_{S^1},R^1)||}{||({D}^\beta_{S^0},R^0) - ({D}^\beta_{S^1},R^1)||}\\
                                                                                                              &= \frac{||(B, {D}^\alpha_{S^B}, R^B) - (\overline{B}, {D}^\alpha_{S^B},R^B)||}{||(B, {D}^\beta_{S^B},R^B) - (\overline{B}, {D}^\beta_{S^B},R^B)||}\\
                                                                                                              &= \frac{\sum_r \Pr[R_B = r] || (B, {D}^\alpha_{S^B})_{|R^B=r} - (\overline{B}, {D}^\alpha_{S^B})_{|R^B=r}||}{\sum_r \Pr[R_B = r] ||(B, {D}^\beta_{S^B})_{|R^B=r} - (\overline{B}, {D}^\beta_{S^B})_{|R^B=r}||}\\
                                                                                                              &\leq \max_r \frac{||(B, {D}^\alpha_{S^B})_{|R^B=r} - (\overline{B}, {D}^\alpha_{S^B})_{|R^B=r}||}{||(B, {D}^\beta_{S^B})_{|R^B=r} - (\overline{B}, {D}^\beta_{S^B})_{|R^B=r}||}
    \end{align*}
    where the last inequality is from the same argument about sequences of positive numbers as the one at the end of the proof of Lemma~\ref{lem:special-polarization-1}.

    This proves what we need, as for any $r$, $((B,D_{S^B})_{|R_B=r}, (\overline{B},D_{S^B})_{|R_B=r})$ is an $(n,1)$-pseudo polarizer.
\end{proof}

\begin{proofof}{Theorem~\ref{thm:special-polarization}}

    For any $(n,\ell,m)$-special polarizer we have the following when $\alpha = 2/3$ and $\beta = 1/3$:
    \begin{align*}
      \frac{||\widehat{D}^\alpha_0 - \widehat{D}^\alpha_1||}{||\widehat{D}^\beta_0 - \widehat{D}^\beta_1||} &\geq \frac{1 - 2^{-m}}{2^{-m}} = 2^m - 1
    \end{align*}

    Lemmas~\ref{lem:special-polarization-1} and \ref{lem:special-polarization-2} imply that there is an $(n,1)$-pseudo polarizer such that:
    \begin{align*}
 \frac{||\widehat{D}^\alpha_0 - \widehat{D}^\alpha_1||}{||\widehat{D}^\beta_0 - \widehat{D}^\beta_1||} \leq \frac{||\widetilde{D}^\alpha_0 - \widetilde{D}^\alpha_1||}{||\widetilde{D}^\beta_0 - \widetilde{D}^\beta_1||} &\leq 2^{(n+1)/2} \left( \frac{\alpha}{\beta} \right)^{n} = 2^{(3n+1)/2}
    \end{align*}

    The above two inequalities tell us that $n = \Omega(m)$.
\end{proofof}

\section{Improved Polarization Places $\SZK$ in $\BPPPATH$}
\label{polarizationbpppath}

Here we show that if the Polarization Lemma of Sahai and Vadhan were strengthened in a black box manner, it would imply $\SZK \subseteq \BPPPATH$. This immediately gives that the Polarization Lemma cannot be strenghtened in this manner.

\subsection{Proof of Theorem \ref{thm:szkbpppathandpolarization}}

To prove the theorem, suppose that the statistical difference problem is $\SZK$-hard for distributions on $N$ bits which are either $\epsilon$-close or $(1-\epsilon)$-far, where $\epsilon=o(2^{-2N/3})$. We will give a $\BPPPATH$ algorithm to solve this problem, using the characterization that $\BPPPATH=\mathsf{postBPP}$. The algorithm is inspired by Aaronson, Bouland, Fitzsimon, and Lee's proof that $\SZK\subseteq \naCQP$ given in \cite{ABFL16}. We thank Tomoyuki Morimae and Harumichi Nishimura for helpful discussions on this topic.

The algorithm is as follows: flip three coins $b_1,b_2,b_3$, and draw independent samples $y_1,y_2,y_3$ from the distributions $D_{b_1},D_{b_2},D_{b_3}$, respectively. Postselect on the condition that $y_1=y_2=y_3$. Output that the distributions are far apart if $b_1=b_2=b_3$, and otherwise output that the distributions are close.

If $\epsilon=0$, then clearly this algorithm is correct. In the case the distributions are far apart, they have disjoint support, which implies the values $b_i$ must be identical, so in this case the algorithm has zero probability of error. In the case the distributions are close, they are identical, so the string $b_1b_2b_3$ is uniformly random after postselection, so the algorithm errs with probability $1/4$. Note that the correctness of this algorithm in the case $\epsilon=0$ doesn't tell us anything new in structural complexity, because in the $\epsilon=0$ case, the problem is in $\NP$ (as a witness to the fact the distributions are identical, simply provide $x_0,x_1$ such that $P_0(x_0)=P_1(x_1)$), and hence is obviously in $\BPPPATH$ and in $\PP$ as well.

We now claim that if $\epsilon=o(2^{-2N/3})$, then this algorithm still works. Note that our choice of $\epsilon$ is asymptotically tight for our algorithm; if $\epsilon=\Omega(2^{-2N/3})$, then there is a simple counterexample which foils the algorithm \footnote{Let $D_0$ be a uniform distribution, and let $D_1$ be the distribution which places an $\epsilon$ amount of weight on a single item $x$, while the remaining weight is spread uniformly on the remaining elements. These distributions are $\epsilon$-close in total variation distance, but one can easily show that this algorithm will yield the string $\hat{b}=111$ with high probability, and hence the algorithm will incorrectly identify them as being far apart. The reason this counterexample works is that postselecting the distributions on seeing the same outcome $y_1=y_2=y_3$ heavily skews the distributions towards more likely $y_i$ outputs, and in this example we will almost always have $y=x$, and hence will almost always output $\hat{b}=111$.}. To show that the algorithm works, we'll show two things. First, if the distributions are $\epsilon$-close for this small $\epsilon$, then we'll show that as $n\rightarrow \infty$, then $\hat{b}$'s value approaches the uniform distribution over all 8 possible output strings. Therefore for sufficiently large $n$, the algorithm is correct. On the other hand, if the distributions are $1-\epsilon$-far, we'll show the algorithm is correct with high probability.

Let's first handle the case in which the distributions are $\epsilon$-close. Let $\hat{b}\in \{0,1\}^n$ be the random variable corresponding to the output of $b_1b_2b_3$. Let $D_b(y)$ denote the probability that distribution $D_b$ outputs $y$. Let $S$ be the event that $y_1=y_2=y_3$, and let $S(y)$ be the event that $y_1=y_2=y_3=y$. By Bayes' rule, we have that
\begin{align*}\Pr[\hat{b}=b_1b_2b_3 | S] &=\frac{ \Pr[ S| \hat{b}=b_1b_2b_3] \Pr[ \hat{b}=b_1b_2b_3]}{\Pr[S]} \\
&= \frac{\sum_{y\in\{0,1\}^n} \Pr[S(y)|  \hat{b}=b_1b_2b_3]  \frac{1}{8}  }{\sum_{y\in\{0,1\}^n}\Pr[S(y)]} \\
&= \frac{\sum_{y\in\{0,1\}^n} D_{1}(y)^{w(\hat{b})}D_{0}(y)^{3-w(\hat{b})}  \frac{1}{8}  }{\sum_{y\in\{0,1\}^n}\Pr[S(y)]}
\end{align*}
where $w(\hat{b})$ is the Hamming weight of $\hat{b}$. 

Hence we have that 
\[\frac{\Pr[\hat{b}=b_1b_2b_3 | S]}{\Pr[\hat{b'}=b_1'b_2'b_3' | S]}=\frac{\sum_{y\in\{0,1\}^n} D_{1}(y)^{w(\hat{b})}D_{0}(y)^{3-w(\hat{b})}}{\sum_{y\in\{0,1\}^n} D_{1}(y)^{w(\hat{b'})}D_{0}(y)^{3-w(\hat{b'})}}\]

We'll now show that as $n\rightarrow \infty$, the ratio of the probabilities between each string tends to $1$. Therefore for sufficiently large $n$, the strings $\hat{b}$ can be make arbitrarily close to equiprobable, so the algorithm works. We'll break into three cases, showing that the strings $\hat{b}=111$ and $000$, $100$, and $110$ become equiprobable as $n\rightarrow \infty$. Since the probability of obtaining a string $\hat{b}$ is only a function of its hamming weight, this will imply all eight possible outcomes for $\hat{b}$ become equiprobable for large $n$, and hence the error probabilty of the algorithm approaches $1/4$ as $n\rightarrow \infty$.

\textbf{Case 1: 111 and 000}

Let's consider the extremal case, where $\hat{b}=111$ or $\hat{b}=000$. Let $\delta_y = |D_1(y)-D_0(y)|$, so $\sum_y \delta_y \leq \epsilon $, and furthermore that $D_0(y) \leq D_1(y) + \delta_y$ and $D_1(y) \leq D_0(y) + \delta_y$. Therefore we have that
\begin{align}
\frac{\Pr[\hat{b}=111 | S]}{\Pr[\hat{b}=000 | S]} &=\frac{\sum_{y\in\{0,1\}^n} D_{1}(y)^3}{\sum_{y\in\{0,1\}^n} D_{0}(y)^{3}} \\
&\leq \frac{\sum_{y\in\{0,1\}^n} (D_{0}(y) + \delta_y)^3}{\sum_{y\in\{0,1\}^n} D_{0}(y)^{3}} \\
&= \frac{\sum_{y\in\{0,1\}^n} D_{0}(y)^3 + 3D_{0}(y)^2 \delta_y+3D_{0}(y) \delta_y^2 + \delta_y^3}{\sum_{y\in\{0,1\}^n} D_{0}(y)^{3}} \\
&=1+ 3\frac{\langle \delta,D_0^2\rangle }{\langle D_0,D_0^2\rangle }+3\frac{\langle \delta^2,D_0\rangle }{\langle D_0^2,D_0\rangle }+\frac{|\delta^3|_1}{|D_0^3|_1} \label{innerprod}\\
&\leq 1 + 3\frac{\epsilon \max_y D_0(y)^2}{\langle D_0,D_0^2\rangle } + 3\frac{\epsilon^2 \max_y D_0(y)}{\langle D_0^2,D_0\rangle } + \frac{\epsilon^3}{|D_0^3|_1} \label{eq:max} \\
&\leq 1 + 3\frac{\epsilon \max_y D_0(y)^2}{\langle D_0,D_0^2\rangle } + 3\frac{\epsilon^2 \max_y D_0(y)}{\langle D_0^2,D_0\rangle } + \frac{2^{-3cn}}{2^{-2n}} \label{eq:denom}
\end{align}
where on line \ref{innerprod} we expressed these sums as inner products, on line \ref{eq:max} we used the fact the sums in the denominators are maximized when the weight of $\delta$ is placed on a single item, line \ref{eq:denom} follows from the fact the denominator is minimized by the unform distribution. We now need to bound the terms $\frac{\max_y D_0(y)^2}{\langle D_0,D_0^2\rangle }$ and $\frac{\max_y D_0(y)}{\langle D_0,D_0^2\rangle }$ as a function of the universe size $N=2^n$. One can easily show that the first is upper bounded by $\Theta(N^{2/3})$, and the second is upper bounded by $\Theta(N^{4/3})$.

To see this, let $k=\max_y D_0(y)$, so $2^{-n}\leq k \leq 1$. Then we have that 
\[\frac{\max_y D_0(y)^2}{\langle D_0,D_0^2\rangle } \leq \frac{k^2}{k^3 + \frac{(1-k)^3}{(N-1)^2}}\]
because given $k$, the denominator is minimized by spreading the remaining probability mass evenly over the remaining $N-1$ elements. By taking the derivative of this as a function of $k$ and setting it equal to zero, we see that the maximum occurs at a solution to the equation $k\left((-5-(N-1)^2)k^3+12k^2-9k+2\right)=0$. As $N\rightarrow\infty$ the real roots of this equation are $0$ and $\Theta(N^{-2/3})$ (plus two complex roots), and one can easily show the first is a minimum while the second is the maximum. Hence this quantity is maximized when $k=\Theta(N^{-2/3})$, which implies the quantity is upper bounded by
\[\frac{\max_y D_0(y)^2}{\langle D_0,D_0^2\rangle } \leq \frac{N^{-4/3}}{N^{-2} + \frac{(1-N^{-2/3})^3}{(N-1)^2}} = \frac{N^{-2/3}}{\Theta(N^-2)} = \Theta(N^{2/3})\]
A similar proof shows that the second quantity is upper bounded by $\Theta(N^{4/3})$. 

Therefore we have that
\begin{align}
\frac{\Pr[\hat{b}=111 | S]}{\Pr[\hat{b}=000 | S]} &\leq 1 + 3*2^{-cn}2^{2n/3} + 3*2^{-2cn}2^{4n/3} + \frac{2^{-3cn}}{2^{-2n}} \\
&\leq 1 +o(1)
\end{align}
since we have $c>2/3$. Note that the identical proof holds for the case where $D_0$ and $D_1$ are switched, therefore we have that $\frac{\Pr[\hat{b}=000 | S]}{\Pr[\hat{b}=111 | S]} \leq 1+o(1)$ as well. Hence we have
\[1-o(1) \leq \frac{\Pr[\hat{b}=111 | S]}{\Pr[\hat{b}=000 | S]} \leq 1+o(1)\]
So as $n\rightarrow \infty$, these strings become equiprobable.

\textbf{Case 2: 111 and 100}
We have that
\begin{align}
\frac{\Pr[\hat{b}=111 | S]}{\Pr[\hat{b}=100 | S]} &=\frac{\sum_{y\in\{0,1\}^n} D_{1}(y)^3}{\sum_{y\in\{0,1\}^n} D_{0}(y)^{2} D_{1}(y)} \\
&\leq \frac{\sum_{y\in\{0,1\}^n} D_{1}(y)(D_0(y)^2+2D_0(y)\delta_y + \delta(y)^2)}{\sum_{y\in\{0,1\}^n} D_{0}(y)^{2} D_{1}(y)} \\
&= 1 + 2\frac{\sum_{y\in\{0,1\}^n} D_{1}(y)D_0(y)\delta_y}{\sum_{y\in\{0,1\}^n} D_{0}(y)^{2} D_{1}(y)} + \frac{\sum_{y\in\{0,1\}^n} D_{1}(y)\delta_y^2}{\sum_{y\in\{0,1\}^n} D_{0}(y)^{2} D_{1}(y)} \\
&=1+2\frac{\langle\delta_y, D_0D_1\rangle }{\langle D_0, D_0D_1\rangle } + \frac{\langle \delta_y^2,D_1\rangle }{\langle D_0^2,D_1\rangle } \\
&\leq 1+2\frac{\epsilon \max_y  D_0(y)D_1(y)}{\langle D_0, D_0D_1\rangle } + \frac{\epsilon^2 \max_y D_1(y)}{\langle D_0^2,D_1\rangle } \\
&\leq 1+2\epsilon\frac{ \max_y  D_0(y)^2 + \delta_y D_0(y)}{\langle D_0, D_0D_1\rangle }+ \epsilon^2\frac{ \max_y D_1(y)}{\langle D_0^2,D_1\rangle } \\
&\leq 1+2\epsilon\frac{ \max_y  D_0(y)^2 }{\langle D_0, D_0^2\rangle  - \epsilon \max_y D_0(y)^2}+ 2\epsilon^2 \frac{\max_y D_0(y)}{\langle D_0, D_0D_1\rangle -\epsilon \max_y D_0(y)^2}  \\&+\epsilon^2\frac{ \max_y D_1(y)}{\langle D_0, D_0^2\rangle  - \epsilon \max_y D_0(y)^2} \label{eq:denomminus}
\end{align}
Where line \ref{eq:denomminus} comes from the fact that $D_1(y) \geq D_0(y)-\delta_y$ for all $y$. We now show that this is upper bounded by $1+o(1)$, by showing that the term $\frac{ \max_y  D_0(y)^2 }{\langle D_0, D_0^2\rangle  - \epsilon \max_y D_0(y)^2} $, the term  $\frac{\max_y D_0(y)}{\langle D_0, D_0D_1\rangle -\epsilon \max_y D_0(y)^2}$ and the term $\frac{ \max_y D_1(y)}{\langle D_0, D_0^2\rangle  - \epsilon \max_y D_0(y)^2}$ are upper bounded by $O(2^{2n/3})$, $O(2^{4n/3})$ and $O(2^{4n/3})$, respectively. This, combined with the fact that $\epsilon = O(2^{-cn})$ for $c>2/3$, implies that $\frac{\Pr[\hat{b}=111 | S]}{\Pr[\hat{b}=100 | S]}\leq 1+o(1)$ as desired.

For the first term, let $k=\max_y D_0(y)$. The this term is upper bounded by
\[\frac{ k^2 }{k^3 - \frac{k^3}{(N-1)^2} - \epsilon k^2} \]
because the denominator is minimized by spreading the remaining probability mass evenly over the remaining $N-1$ elements. One can easily show this function is maximized by setting $k=\Theta(N^{-2/3})$. Indeed, taking the derivative of this equation and setting it equal to 0, one can see that the extreme values of $k$ satisfy $k(-2(N-1)^2k^3 -3k+2)=0$. Hence the optimal value of $k$ satisfies $k=\Theta(N^{-2/3})$. For this value of $k$, the term evalues to $O(N^{2/3})$.

For the second term, if we let $k$ be defined as above, then by the same reasoning we have that the term is uppor bounded by 
\[\frac{ k }{k^3 - \frac{k^3}{(N-1)^2} - \epsilon k^2} \]
Again by a similar proof, one can easily show the function is maximized by setting $k=O(N^{-2/3})$, which implies the term is upper bounded by $O(2^{4n/3})$.

For the third term, let $k = \max_y D_1(y)$. By a similar argument as above, we have that
\[\frac{ \max_y D_1(y)}{\langle D_0^2,D_1\rangle } \leq \frac{ \max_y D_1(y)}{\langle D_1^2,D_1\rangle  - 2\langle \delta, D_1^2\rangle  + \langle \delta^2,D_1\rangle } \leq \frac{k}{k^3 - \frac{(1-k)^3}{(N-1)^2} -2\epsilon k^2 +\epsilon^2 k}\]
One can show that this term is maximized be setting $k=\Theta(N^{-2/3})$, and therefore this term is upper bounded by $O(N^{4/3})$. Indeed, taking the derivative of this quantity with respect to $k$ and setting it equal to zero, one can see that the maximum value of $k$ satisfies $(-2(N-1)^2+2)k^3+(2\epsilon -3)k^2 +1=0$, which implies the maximum value satisfies $k=\Theta(N^{-2/3})$.

We've now shown that $\frac{\Pr[\hat{b}=111 | S]}{\Pr[\hat{b}=100 | S]} \leq 1+ o(1)$. Now consider the opposite ratio. By the same reasoning as before, we have that
\begin{align}\frac{\Pr[\hat{b}=100 | S]}{\Pr[\hat{b}=111 | S]}  &= \frac{\sum_{y\in\{0,1\}^n} D_{0}(y)^2 D_1(y)}{\sum_{y\in\{0,1\}^n} D_{1}(y)^{3} } \\
&\leq 1+ 2\frac{\langle \delta,D_1^2\rangle }{\sum_{y\in\{0,1\}^n} D_{1}(y)^{3} } + \frac{\langle \delta^2,D_1\rangle }{\sum_{y\in\{0,1\}^n} D_{1}(y)^{3} } \\
&\leq 1 + 2\epsilon \frac{\max_y D_1(y)}{\sum_{y\in\{0,1\}^n} D_{1}(y)^{3} } + \epsilon^2 \frac{\max_y D_1(y)^2}{\sum_{y\in\{0,1\}^n} D_{1}(y)^{3} }\\
&\leq 1+o(1) \label{eq:appealtocase1}
\end{align}
Where on line \ref{eq:appealtocase1} we used the fact that we previously upper bounded these terms when handling Case 1. Hence as $n\rightarrow\infty$ the strings $\hat{b}=111$ and $\hat{b}=100$ become equiprobable.

\textbf{Case 3: 111 and 110}
We have that
\begin{align}
\frac{\Pr[\hat{b}=111 | S]}{\Pr[\hat{b}=110 | S]} &=\frac{\sum_{y\in\{0,1\}^n} D_{1}(y)^3}{\sum_{y\in\{0,1\}^n} D_{0}(y) D_{1}(y)^{2}} \\
&\leq 1+\frac{\sum_{y\in\{0,1\}^n} \delta_y D_{1}(y)^2}{\sum_{y\in\{0,1\}^n} D_{0}(y) D_{1}(y)^{2}} \\
&= 1+\frac{\langle \delta_y,D_1^2\rangle }{\langle D_0,D_1^2\rangle } \\
&\leq 1+ \frac{\epsilon \max_y D_1(y)^2}{\langle D_0,D_1^2\rangle } \\
&\leq 1+ \frac{\epsilon \max_y D_1(y)^2}{\langle D_1,D_1^2\rangle  - \langle \delta_y, D_1^2\rangle } \label{eq:minusepsilon} \\
&\leq 1+ \frac{\epsilon \max_y D_1(y)^2}{\langle D_1,D_1^2\rangle  - \epsilon \max_y D_1(y)^2} \label{eq:epsilonk2} \\
&\leq 1+ o(1) \label{eq:appealtocase2}
\end{align}
Where on line \ref{eq:minusepsilon} we used the fact that $D_0(y)\geq D_1(y) - \delta(y)$, on line \ref{eq:epsilonk2} we used that fact that the numerator is minimized if all the mass of $\delta_y$ is placed on the maximum likelihood event of $D_1$, and on line \ref{eq:appealtocase2} we used the fact that this is the same as the first term we bounded in Case 2. 

Now consider the opposite ratio. We have that 
\begin{align}
\frac{\Pr[\hat{b}=110 | S]}{\Pr[\hat{b}=111 | S]} &=\frac{\sum_{y\in\{0,1\}^n} D_{0}(y) D_{1}(y)^2}{\sum_{y\in\{0,1\}^n} D_{1}(y)^3} \\
&\leq \frac{\sum_{y\in\{0,1\}^n} (D_{1}(y)+\delta(y)) D_{1}(y)^2}{\sum_{y\in\{0,1\}^n} D_{1}(y)^3} \\
&= 1+ \frac{\sum_{y\in\{0,1\}^n} \delta(y) D_{1}(y)^2 }{\sum_{y\in\{0,1\}^n} D_{1}(y)^3} \\
&\leq 1+o(1)
\end{align}
Where the last line follows from our previous arguments in Case 1. Hence we have that $1-o(1)\leq \frac{\Pr[\hat{b}=111 | S]}{\Pr[\hat{b}=110 | S]} \leq 1+o(1)$ as desired.

Hence we have shown $1-o(1) \leq \frac{\Pr[\hat{b}=111 | S]}{\Pr[\hat{b}=x | S]} \leq 1+o(1)$ for any three-bit string $x$. Hence all strings are equiprobable, so in the case the distributions are $\epsilon$-close, the algorithm's error probability tends to $1/4$ as $n\rightarrow\infty$, and hence the algorithm is correct in this case.

To complete the proof, we now show that the probability of error is low then the distributions are $1-\epsilon$ far apart in total variation distance. 

Suppose the distributions are $1-\epsilon$ far apart in total variation distance. By the definition of total variation distance, there must exist some event $T\subseteq \{0,1\}^n$ for which $|D_0(T)-D_1(T)|\geq 1-\epsilon$, where the notation $D_0(T)$ indicates the probability that $D_0$ outputs an element of the set $T$, i.e. $D_O(T)=\sum_{y\in T} D_0(y)$. Without loss of generality we have that $D_0(T)-D_1(T)\geq 1-\epsilon$, which implies $D_1(\bar{T}) - D_0(\bar{T})\geq 1-\epsilon$. Since $D_0$ and $D_1$ are probability distributions, this implies $D_0(T)\geq 1-\epsilon$ and $D_1(T)\leq \epsilon$, and likewise $D_1(\bar{T})\geq 1-\epsilon$ and $D_0(\bar{T})\leq \epsilon$. In other words $D_0$ has almost all its probability mass in $T$ and $D_1$ has almost all its probability mass in $\bar{T}$.

We'll now show that under these distributions, one will almost certainly see the output $\hat{b}=000$ or $\hat{b}=111$. As before, we'll show this by proving that for large $n$, the strings $\hat{b}=000$ or $111$ are far more likely than $\hat{b}=001$ or $\hat{b}=011$, which implies the algorithm almost always outputs the correct answer.

Let $k_0 =\max_{y\in T} D_0(y)$ and let $k_1 = \max_{y\in\bar{T}} D_1(y)$. Suppose without loss of generality that $k_0\geq k_1$ (otherwise exchange $D_0$ and $D_1$ in the argument). Then we have that
\begin{align}
\frac{\Pr[\hat{b}=000 | S]}{\Pr[\hat{b}=100 | S]} &= \frac{\sum_{y\in\{0,1\}^n} D_{0}(y)^3}{\sum_{y\in\{0,1\}^n} D_{0}(y)^2 D_{1}(y)} \label{eq:usualargs} \\
&= \frac{\langle D_0^2,D_0\rangle }{\langle D_0^2,D_1\rangle } \\
&\geq \frac{k_0^3 + \frac{(1-k_0)^3}{(N-1)^2}}{\epsilon k_0^2 + \epsilon^2 k_1} \label{eq:k0}  \\
&\geq \frac{k_0^3 + \frac{(1-k_0)^3}{(N-1)^2}}{\epsilon k_0^2 + \epsilon^2 k_0} \label{eq:k02}
\end{align}
where line \ref{eq:usualargs} follows from the same arguments as the previous section, and line \ref{eq:k0} follows because the numerator is minimized by placing the uniform distribution on all elements other than the element responsible for $k_0$, and the denominator is maximized if all the weight that $D_0$ has on $\bar{T}$ is placed on the element of maximal weight under $D_1$, and vice versa. Line \ref{eq:k02} follows from the fact that $k_0\geq k_1$

Now we show that this quantity is $\omega(1)$, i.e. it approaches infinity as $n\rightarrow \infty$. Suppose by contradiction that there exists a constant $c>1$ which is an upper bound for this quantity. Since $\epsilon = o(N^{-2/3})$, there exists an $n_0$ such that for all $n>n_0$, $\epsilon < \frac{1}{2c}N^{-2/3}$. We claim that for all $n>n_0$, this quantity is greater than $c$, which is a contradiction. 

To see this, we break into three cases. 

\textbf{Case 1:} $k_0 \geq N^{-2/3}$

In this case, the numerator is at least $k_0^3$, while the denominator is at most $\frac{1}{10c} N^{-2/3}k_0^2 + \frac{1}{100c^2}N^{-4/3}k_0  \leq \frac{1}{2c} k_0^3 + \frac{1}{4c^2}k_0^3 < \frac{1}{c}k_0^3$, where the last step follows from the fact that $\frac{1}{2c} \frac{1}{4c^2} < \frac{1}{c}$ for any $c>1$. Therefore the quantity on line \ref{eq:k02} is strictly greater than $\frac{k_0^3}{\frac{1}{c}k_0^3}=c$ as desired.

\textbf{Case 2:} $k_0 \leq N^{-2/3}$

In this case the numerator is at least $\frac{(1-k_0)^3}{(N-1)^2}$ which is $\geq 0.75N^{-2}$ for sufficiently large $n$, while the denominator is $\epsilon k_0^2 + \epsilon^2 k_0 \leq \frac{1}{2c} N^{-2} +\frac{1}{4c^2}N^{-2} < \frac{3}{4c}N^{-2}$ for $c>1$, which follows from our upper bounds on $\epsilon$ and $k_0$. Hence the quantity on line \ref{eq:k02} is strictly greater than $c$ as desired. 

Therefore we have shown that and $n\rightarrow\infty$, the string $\hat{b}=000$ (or $\hat{b}=111$, if $k_0<k_1$) is much more likely to occur than $\hat{b}=001$.

A similar proof holds to show that the string $\hat{b}=000$ (or $\hat{b}=111$) is more likely to occur than $\hat{b} =110$; indeed by the same arguments as above, assuming $k_0\geq k_1$, we have
\[\frac{\Pr[\hat{b}=000 | S]}{\Pr[\hat{b}=110 | S]} \geq \frac{k_0^3 + \frac{(1-k_0)^3}{(N-1)^2}}{\epsilon^2 k_0 + \epsilon k_1} \geq \frac{k_0^3 + \frac{(1-k_0)^3}{(N-1)^2}}{\epsilon^2 k_0 + \epsilon k_0} \geq \omega(1)
\]

Hence the string $\hat{b}=000$ is far more likely to occur than the strings $\hat{b}=001$ or $\hat{b}=011$ (assuming $k_0>k_1$, otherwise the string $\hat{b}=111$ is more likely to occur than $001$ or $011$), and hence the algorithm errs with probability $o(1)$ when the distributions are $1-\epsilon$ far apart. This completes the proof.


\end{document}